\documentclass{article}
\usepackage{amsmath}
\usepackage{amssymb}
\usepackage{amsthm}
\usepackage{amsfonts}
\usepackage{float} %prevent floating

\usepackage{pdflscape}%rotate table

\usepackage{graphicx} 
\usepackage{tikz}
\usepackage{color}
\usepackage{multicol} 
\usepackage{xargs}
\usepackage{nicefrac}
\usepackage{bbm}

\usepackage[english]{babel}

\usepackage[latin1]{inputenc}

\usepackage{hyperref}
\def\gs{r}
\def\gr{s}
\def\codim{\mathrm{codim}\ }

\def\bbbn{\mathbbm N}
\def\bbbc{\mathbbm C}
\def\CC{\mathbbm C}
\def\CCP{\mathbbm{CP}}
\def\bbbq{\mathbbm Q}
\def\bbbz{\mathbbm Z}
\def\bbbd{\mathbbm D}
\def\bbbt{\mathbbm T}

\def\bbbo{\mathbbm O}
\def\bbbi{\mathbbm Y}

\def\alphat{\alpha_{4,1}}
\def\alphao{\alpha_{6,1}}
\def\alphai{\alpha_{12}}
\def\balphat{\bar{\alpha}_{6}}
\def\bbetat{\bar{\beta}_{8}}
\def\bgammat{\bar{\gamma}_{12}}
\def\betat{\beta_{4,2}}
\def\betai{\beta_{20}}
\def\betao{\beta_{8,0}}
\def\balphao{\bar{\alpha}_{8}}
\def\bbetao{\bar{\beta}_{12}}
\def\bgammao{\bar{\gamma}_{18}}
\def\gammat{\gamma_{6,0}}
\def\gammao{\gamma_{12,1}}
\def\gammai{\gamma_{30}}

\def\cp1{{\mathbbm C\mathbb P}^1}

\newtheorem{Remark}{Remark}

\newtheorem{Example}{Example}
\newtheorem{Lemma}{Lemma}

\newtheorem{Corollary}{Corollary}

\newtheorem{Theorem}{Theorem}

\newtheorem{Definition}{Definition}

\newcommand{\mf}[1]{\mathfrak{#1}}
\newcommand{\mc}[1]{\mathcal{#1}}

\newcommandx{\AIJ}[3][1=I,2=G,3=\alpha]{\mathbbm{#1}_{#3}^{#2}}
\newcommand{\mb}[1]{\mathbbm{#1}}

\newcommand{\gal}{g_\alpha} %generators of a group
\newcommand{\gbe}{g_\beta}
\newcommand{\gga}{g_\gamma}
\newcommand{\coa}{\kappa_\alpha}
\newcommand{\cob}{\kappa_\beta}
\newcommand{\coc}{\kappa_\gamma}

\newcommand{\SL}{\mathrm{SL}}

\newcommand{\PSL}{\mathrm{PSL}}

\newcommand{\Aut}{\mathrm{Aut}}

\newcommand{\diag}{\mathrm{diag}}

\newcommand{\inv}[1]{\hat{#1}}

\newcommand{\bpf}{\begin{proof}}
\newcommand{\epf}{$\qed$\end{proof}}

\newcommand\rlie[3]{\mf{#1l}_{#2}(#3)}
\newcommand\lie[2]{\mf{#1l}_{#2}(\bbbc)}
\renewcommand{\sl}{\lie{s}{2}}

\newcommand{\sln}{\lie{s}{n}}

\newcommand{\splitk}{\mathbbm{k}}

\newcommand\alias[4]{(#1\otimes \splitk(\lambda))_{#2}^{#4}}
\newcommand\alian[3]{(\mf{sl}_{#3}({#1})\otimes \splitk(\lambda))_{\alpha}^{#2}}

\newcommand\nf[1]{\|\mf{sl}({#1})\|}
\renewcommand{\sl}{\mf{sl}_2(\mathbbm{C})}

\def\cp1{{\mathbbm C\mathbbm P}^1}

\newcommand{\Zn}[1]{\bbbz/{#1}}
\begin{document}

% BIOBLIOGRAPHY STYLES
\bibliographystyle{plain}
\title{Higher dimensional Automorphic Lie Algebras}
%\title{Higher dimensional Automorphic Lie Algebras;\\ the special linear case}

\author{Vincent Knibbeler, Sara Lombardo\\
Department of Mathematics and Information Sciences\\
 Northumbria University, Newcastle upon Tyne, UK\\
\and
Jan A.~Sanders \\
Department of Mathematics, Faculty of Sciences\\
Vrije Universiteit, De Boelelaan 1081a, 1081 HV Amsterdam, The Netherlands
}

%\keywords{}
%\subjclass{}
%\date{\today,\version}
%\date{\today}
\date{}

\maketitle
%\numberwithin{equation}{section}
%\numberwithin{table}{section}
\numberwithin{Theorem}{section}
\numberwithin{Lemma}{section}
\numberwithin{Definition}{section}
\numberwithin{Remark}{section}
\numberwithin{Corollary}{section}
\numberwithin{Example}{section}

\begin{abstract}
The paper presents the complete classification of Automorphic Lie Algebras based on 
$\sln$,
where the symmetry group $G$ is finite and the orbit is any of the exceptional $G$-orbits in $\overline{\bbbc}$.
A key feature of the classification is the study of the algebras in the context of \emph{classical invariant theory}. 
This provides on one hand  a powerful tool from the computational point of view, on the other it opens new questions from an algebraic perspective, which suggest further applications of these algebras, beyond the context of integrable systems. 
In particular, the research shows that Automorphic Lie Algebras associated to the $\bbbt\bbbo\bbbi$ groups  (tetrahedral, octahedral and icosahedral groups) depend on the group through the automorphic functions only, thus they are group independent as Lie algebras. 
This can be established by defining a \emph{Chevalley normal form} for these algebras, generalising this classical notion to the case of Lie algebras over a polynomial ring.

\vspace{1cm}
\noindent\textbf{AMS  Subject Classification Numbers}\\
%15A54: Linear and multilinear algebra; matrix theory: matrices over function rings in one or more variables;\\
16Z05: Associative rings and algebras: computational aspects of associative rings;\\
17B05: Nonassociative rings and algebras: structure theory;\\
17B65: Nonassociative rings and algebras:  infinite-dimensional Lie (super)algebras;\\
17B80: Nonassociative rings and algebras: applications to integrable systems.

\end{abstract}

\newpage

%%%%%% Temporary
%\tableofcontents
%\listoffigures
%\listoftables
%%%%%%

%\numberwithin{equation}{section}
%\numberwithin{table}{section}
%\numberwithin{Theorem}{section}
%\numberwithin{Lemma}{section}
%\numberwithin{Definition}{section}
%\numberwithin{Remark}{section}
%\numberwithin{Corollary}{section}
%\numberwithin{Example}{section}
\section{Introduction} 
\label{sec:intro}
An Automorphic Lie Algebra (ALiA in what follows) is the space of invariants 
\[
(\mf{g}\otimes\mc{M}(\overline{\mathbbm{C}}))^G_\Gamma
\] 
obtained by imposing a finite group symmetry on a current algebra of Krichever-Novikov (KN) type \cite{schlichenmaier2013virasoro}
$\mf{g}\otimes\mc{M}(\overline{\mathbbm{C}})$ where $\mf{g}$ is a Lie algebra, $\mc{M}(\overline{\mathbbm{C}})$ the field of meromorphic functions on the Riemann sphere $\overline{\mathbbm{C}}=\mathbbm{C}\cup\{\infty\}$, $G$ a subgroup of $\textrm{Aut}(\mf{g}\otimes\mc{M}(\overline{\mathbbm{C}}))$ and where $\Gamma\subset\overline{\mathbbm{C}}$ is a $G$-orbit, to which poles are confined.
Since their introduction in  \cite{lm_cmp05} automorphic algebras  have been extensively studied (see \cite{MR2718933} and references therein, but also \cite{Bury2010} and \cite{Chopp2011}). 
%%%%%%%%%%%%%%%
ALiAs arose originally in the context of algebraic reductions of integrable equations \cite{lm_cmp05}, motivated by the problem of algebraic reduction of Lax pairs \cite{Mikhailov81}. While the classification problem is a stand-alone one, its solution could have an impact also in applications to the theory of integrable systems and beyond. In particular, the Chevalley normal form (see Section \ref{sec:Chev}) can be used as starting point to analyse Lax pairs and consequently associated integrable equations.

A first step towards the classification of ALiAs was presented in \cite{lm_cmp05}, 
where automorphic algebras associated  to finite groups were considered. 
These groups are those of Klein's classification, namely, the cyclic groups $\Zn{N}$, the dihedral groups $\bbbd_N$, the tetrahedral group $\bbbt$, the octahedral group $\bbbo$ and the icosahedral group $\bbbi$. In \cite{lm_cmp05} the authors study automorphic algebras associated to the dihedral group $\bbbd_N$, starting from the finite dimensional algebra $\sl$; examples of ALiAs based on $\lie{s}{3}{}$ were also discussed. 
In \cite{KLS-DN} the authors present a complete classification of automorphic algebras associated to the dihedral group $\bbbd_N$. 
A further, crucial, step toward the full classification appears in \cite{MR2718933}, where the problem is formulated in a uniform way using the theory of invariants. This allows for a complete classification of $\sl$-based ALiAs with finite group symmetry. The new approach inspires the present results; however the simplifying assumption that the representations of $G$ acting on the spectral parameter $\lambda$ as well as on the base Lie algebra are the same, as in \cite{MR2718933}, can no longer be made when considering higher dimensional Lie algebras.

The aim of this paper is to present the complete classification of Automorphic Lie Algebras for the case \(\mf{g}=\sln\)
with poles at an exceptional $G$-orbit. Exceptional orbits $\Gamma$ are those with less than $|G|$ elements; they are labelled by $\zeta=\alpha, \beta, \gamma$, where $\alpha, \beta, \gamma$ refer to the forms with zeros at $\Gamma_\zeta$.
%\footnote{In the following we will refer to the $G$-orbit either as $\Gamma$ or $\Gamma_\zeta$, depending on the context, and where no confusion ca be made}.
A key feature of this approach is the study of these algebras in the context of \emph{classical invariant theory}. 
In brief, the Riemann sphere is identified with the complex projective line $\CCP^1$ consisting of quotients $\nicefrac{X}{Y}$ of two complex variables by setting $\lambda=\nicefrac{X}{Y}$. M\"obius transformations on $\lambda$ then correspond to linear transformations on the vector $(X,Y)$ by the same matrix. Classical invariant theory is then used to find the $G$-invariant subspaces of $\CC[X,Y]$-modules, where $\CC[X,Y]$ is the ring of polynomials in $X$ and $Y$. These ring-modules of invariants are then localised by a choice of multiplicative set of invariants. This choice corresponds to selecting a $G$-orbit $\Gamma_\zeta$ of poles. The set of elements in the localisation of degree zero, i.e.~the set of elements which can be expressed as functions of $\lambda$, generate the ALiA.
Once the algebra is computed, it is transformed into a Chevalley normal form in the spirit of the standard 
Chevalley 
basis \cite{MR0323842}; we believe this is the most convenient form for analysis. The \emph{isomorphism question} can finally be answered in the \(\sln\) case and a more refined isomorphism conjecture is formulated:

\emph{Let $G$ and $G'$ be two of the groups from \(\bbbt, \bbbo, \bbbi\) or \(\bbbd_N\) and let $\Gamma_\zeta$ and $\Gamma'_{\zeta'}$ be exceptional $G$- and $G'$-orbits, respectively. Then, the Automorphic Lie Algebras $(\mf{g}\otimes\mc{M}(\overline{\mathbbm{C}}))^G_{\Gamma_{\zeta}}$ and $(\mf{g}'\otimes\mc{M}(\overline{\mathbbm{C}}))^{G'}_{\Gamma'_{\zeta'}}$ are isomorphic as Lie algebras if $\mf{g}\cong\mf{g}'$ and $\kappa_\zeta=\kappa_{\zeta'}$} (cf.~Table \ref{tab:codimg} --
see Theorem \ref{teo:AL} for the precise statement). 

%%%%%%%%%%%%
Classical invariant theory  provides a powerful tool of analysis from the point of view of computations. Indeed, one of the obstacles to a complete classification so far were computational difficulties arising on one hand from choosing two different group representations, which implies a \emph{ground form} of higher degree, rather than of degree two as in \cite{MR2718933}, on the other hand the intrinsic difficulty arising from the higher dimensionality of the problem (moving from $\sl$ to $\sln$, $n>2$). It is worth noting here that in this paper we will consider only inner automorphisms in $\textrm{Aut}(\mf{g}\otimes\mc{M}(\overline{\mathbbm{C}}))$. 
This is however not so restrictive as it might seem at first, as only the octahedral group $\bbbo$ might admit outer automorphisms in the case of  $\sln$, $n>2$ \cite{knibbeler2014}. The analysis of all admissible automorphisms in $\textrm{Aut}(\mf{g}\otimes\mc{M}(\overline{\mathbbm{C}}))$ given a Lie algebra  \(\mf{g}\) is a very interesting one, and it is left for further investigation.

The main results of the classification can be summarised as follows:
\begin{enumerate}
\item The long-standing isomorphism conjecture, 
due to Mikhailov, is now a theorem for \(\mf{g}=\sln\) (see Theorem \ref{teo:AL}).
The proof  relies on the explicit Chevalley normal form of the algebras.
\item The number of automorphic functions present in each normal form is an invariant (see Sections \ref{sec:Chev} and \ref{sec:Invariants}).
\end{enumerate}

The results also suggest a natural interpretation of these algebras as finitely generated over the ring \(\splitk[\AIJ[I][][\Gamma]]\), 
where \(\splitk\) is an extension of $\bbbq$ with a root of unity depending on the irreducible representations of the group $G$, and \(\AIJ[I][][\Gamma]\) is a $G$-automorphic function with poles at the orbit $\Gamma$ (note that the field and the automorphic function are group dependent, but we do not want to overload the notation by calling it \(\splitk_G\); this also
underlines the fact that the group dependency does not play a big role).

The alternative is to consider it as an infinite dimensional Lie algebra over \(\splitk\),  graded by powers of \(\AIJ[I][][\Gamma]\), as has been done in earlier publications, cf.~\cite{MR2718933}, where both approaches are used in parallel, or in \cite{lm_cmp05}, \cite{Bury2010}, and, in the context of KN type algebras, in \cite{MR2058804}. While the former approach adds some computational complications, one is rewarded with classical looking Chevalley normal form results (see Section \ref{sec:Chev}) and the Cartan matrix is the same as the one from the original Lie algebra.
It is worth pointing out that
in both approaches one can ask whether the ALiA can
be brought into normal form, as for instance in the case of the
Chevalley basis for simple Lie algebras over \(\bbbc\).
As with any normal form question, one has to determine the transformation group.
In the context of infinite dimensional Lie algebras, there are now two approaches in use:
(i) the graded approach, where one allows invertible linear
transformations on the algebra
respecting the grading. This approach in particular keeps the grading
depth invariant \cite{lm_cmp05}.
(ii) The filtered approach, used in this paper and introduced in \cite{MR2718933}, where one allows invertible linear
transformations of filtering degree \(0\), where the filtering is
induced by the grading in the usual manner.
Here the quasigrading is respected, but the grading depth may increase.
Since the second group of transformations contains the first,
the normal form space will be smaller.
Explicitly, if the algebra $(\mf{g}\otimes\mc{M}(\overline{\mathbbm{C}}))^G_\Gamma$ is generated by $m$ matrices over the ring \(\splitk[\AIJ[I][][\Gamma]]\),
then the first approach uses the transformation group \(\{T\in Mat_{m\times m}(\mb{k})\,|\,\det(T)\in\mb{k}^\ast\}=GL(\mb{k}^{m})\)
and the second uses \(\{T\in Mat_{m\times m}(\mb{k}[\AIJ[I][][\Gamma] ])\,|\,\det(T)\in\mb{k}^\ast\}=GL(\mb{k}^{m})\oplus\,\bigoplus_{d=1}^{\infty} End (\mb{k}^{m})\,\AIJ[I][d][\Gamma] \),  %=GL((\mf{g}\otimes\mc{M}(\overline{\mathbbm{C}}))^G_\Gamma)\)., 
namely the general linear group of the vector space $(\mf{g}\otimes\mc{M}(\overline{\mathbbm{C}}))^G_\Gamma$.

We remark  that the finite group theory used here is completely classical, with the exception of the results 
in Section \ref{sec:Invariants},
whereas the Lie algebra theory over a polynomial ring is slightly more modern, but it is the combination of the two that poses the central question in  this paper.

Finally, it is worth pointing out that the classification is driven by computational inputs: many of the necessary computations were done using the FORM package \cite{Form00}, calling on GAP \cite{GAP} and Singular \cite{MR2363237}.

%%%%%%%%%%%%%%%%%%%%%%%%%%%%%%%%%%%%%%%%%%%%%%%%%%%%%%%%%%%%%%%%%%%%
The paper is organised as follows: 
in the next section the computational challenges are presented and addressed in two ways (the difficulties arising from the increasing dimensionality of the problem are discussed in Section \ref{sec:computing} but ultimately addressed in Section \ref{sec:moi}): first, by using classical invariant theory, 
thus working with polynomials in $X$ and $Y$ (Section \ref{sec:motivations}), rather than rational functions of $\lambda$, 
until the very last stage when the Riemann sphere is identified with the complex projective line $\CCP^1$ by setting $\lambda=\nicefrac{X}{Y}$. 
%Next, the difficulties arising from the increasing dimensionality of the problem, are discussed in Section \ref{sec:computing} and ultimately addressed in Section \ref{sec:moi}.
Section \ref{sec:gap} recalls the necessary background from representation theory of finite groups, considering in particular the $\bbbt\bbbo\bbbi$ groups. 
Sections \ref{sec:gap} and \ref{sec:Molien} recall basic notions from invariant theory, such as decompositions into irreducible representations and Molien series. In Section \ref{sec:comp_inv_mat} invariant matrices are computed by means of transvection (Section \ref{sec:trans}).  The second major computational challenge of the problem is addressed in Section \ref{sec:moi} introducing the  concept of \emph{matrices of invariants}, which in turn allows one to define \emph{Chevalley normal form} for ALiAs. Normal forms for \(\lie{s}{n}\)--based ALiAs are given in Section \ref{sec:Chev}. 
In this Section we consider an extension of  the Jacobson-Morozov construction to the case of $\mb{k}[\AIJ[I][][\Gamma]]$-Lie algebras.  Section \ref{sec:Invariants} introduces the concept of \emph{invariant of Automorphic Lie Algebras}. The predicting power of invariants is discussed in the Conclusions (Section \ref{sec:conclusions}) where the main findings are commented upon.

%%%%%%%%%%%%%%%%%%%%%%%%%%%%%%%%%%%%%%%%%%%%%%%%%%%%%%%%%%%%%%%%%%%%

%%%%%%%%%%%%%%%%%%%%%%%%%%%%%%%%%%%%%%%%%%%%%%%%%%%%%%%%%%%%%%%%%%%%

\section{Computing Automorphic Lie Algebras}
\label{sec:computing} 
One of the obstacles to a complete classification of Automorphic Lie Algebras so far has been of computational nature: difficulties arising on one hand from the choice of two different group representations, which implies a \emph{ground form} of higher degree, rather than of degree two as in \cite{MR2718933}. On the other hand the intrinsic difficulty arising from the higher dimensionality of the problem, moving from $\sl$ to $\sln$, $n>2$
. These difficulties are overcome here in two ways: first, by using classical invariant theory, 
thus working with polynomials in $X$ and $Y$ rather than rational functions of $\lambda$, 
until the very last stage when the Riemann sphere is identified with the complex projective line $\CCP^1$ by setting $\lambda=\nicefrac{X}{Y}$. 
This allows us a better control of the degrees of the invariants at each step of the computation 
and it enables the use of Molien's theory to predict the degree of the invariants, 
and to check the outcome of the computations as well. Working on $\bbbc[X,Y]$ allows us also to use transvectants,  
an easy to implement computational tool in classical invariant theory (see Section \ref{sec:trans}). 
The difficulty arising from the higher dimensionality of the problem is instead dealt with introducing \emph{matrices of invariants} (see Section \ref{sec:moi}), 
which are computationally very effective. They are defined by considering the action of invariant matrices on invariant vectors, by multiplication. 
The description of the invariant matrices in terms of this action yields greatly simplified matrices, whose entries  are indeed $G$-invariant. 
The map to matrices of invariants preserves the structure constants of the Lie algebra. 
We emphasise that the matrices of invariants are not invariant under
the usual group action, because they are expressed in a
$\lambda$-dependent basis that trivialises the conjugation action on
the matrices, leaving only the action on the spectral parameter
$\lambda$ (see next section).

We start by defining \emph{Polynomial Automorphic Lie Algebras}.

%%%%%%%%%%%%%%%%%%%%%%%%%%%%%%%%%%%%%%%%%%%%%%%%%%%%%%%%%%%%%%%%%%%%

\subsection{Polynomial Automorphic Lie Algebras}
\label{sec:motivations} 
Let $G$ be a finite group and let $\sigma$ be a faithful, projective $G$-representation:
\[
\sigma:G\rightarrow \textrm{GL}_2(\bbbc)\,.
\]
This restricts $G$ to the groups 
\[
\bbbz/N,\;\bbbd_N,\;\bbbt,\;\bbbo,\;\bbbi
\] 
of Klein's classification \cite{MR1315530,MR0080930} where $\bbbz/N$ is the cyclic group, $\bbbd_N$ the dihedral group, $\bbbt$ the tetrahedral group, $\bbbo$ the octahedral group and $\bbbi$ the icosahedral group. In this paper we focus on the exceptional cases (since they are not part of infinite families), the $\bbbt$$\bbbo$$\bbbi$ groups. 
The  $\bbbd_N$-classification has been presented in \cite{KLS-DN}, both for generic and exceptional $G$-orbits, since the $\bbbd_N$ computations can be done explicitly without the use of a computer. In addition, this is the only non abelian group in Klein's classification whose order depends on $N$, which is a 
complication from a computational point of view and we prefer to keep it separate. 

Let $\tau:G\rightarrow \textrm{PGL}(V)$ be an irreducible $G$-representation, consider the Lie algebra
\[
\mf{g}(V)\otimes \bbbc[X,Y]\,
%\textrm{ with } \mf{g}(V)=\clie{gl}{}{V,\bbbc},
\]
where $\mf{g}(V)$ is a complex Lie algebra in $\mf{gl}(V)$ and $\bbbc[X,Y]$ is the ring of polynomials in $X$ and $Y$. The representations $\sigma$ and $\tau$ induce a $G$-action on $\mf{g}(V)\otimes \bbbc[X,Y]$ (see \cite[Section 1.5, 1.6]{MR0450380}) by identifying $\mf{gl}(V)=V\otimes V^{\ast}$, where $V^{\ast}$ is the dual space,
\[
g\cdot\big(M\otimes p(X,Y)\big)=\tau(g) M \tau(g^{-1})\otimes p\big(\sigma(g^{-1})(X,Y)\big).
\]  
Notice that this defines a Lie algebra automorphism of $\mf{g}(V)\otimes\bbbc[X,Y]$.

\begin{Definition}
\label{def:relativeinv}
Let \(V\) be a \(G\)-module. An element $v\in V$ is called \textbf{$\chi$-relative invariant}
if there exists a homomorphism \(\chi\,:\, G\to\bbbc^{\ast}\), the multiplicative group of \(\bbbc\), 
such that \( g\,v=\chi (g)\, v\).
%if there exists a representation \(\chi\) of the abelianisation \(G/[G,G]\) such that \( g\,v=\chi (g)\, v\), \(\,\forall g\in G\) (\(\chi\) is a one-dimensional representation since \(G/[G,G]\)  is abelian). 
If \(\chi \) is trivial then \(v\) is called \textbf{invariant}. The space of $\chi$-relative invariants in $V$ will be denoted by $V^{\chi}_G$ (or simply $V^{\chi}$ if there is no confusion with respect to the group), the space generated by all relative invariants by $V_G$ and the subspace of invariants by $V^G$. 
\end{Definition}
\begin{Remark}
\label{rm:det}
An example of a homomorphism \(\chi\,:\, G\to\bbbc^{\ast}\)
%a representation \(\chi\) of the abelianisation \(G/[G,G]\) 
is the determinant of a $G$-representation $\rho$, $\Delta_\rho (g)=\det \rho (g)$.
\end{Remark}
\begin{Definition}
The algebra $(\mf{g}(V)\otimes\bbbc[X,Y])^G$ defines a \textbf{Polynomial Automorphic Lie Algebra} based on $\mf{g}(V)$ cf.~\cite{MR2718933}.
\end{Definition}
Our first goal will be to compute Polynomial ALiAs, $(\mf{g}(V)\otimes\bbbc[X,Y])^G$, where $G$ is one of the $\bbbt$$\bbbo$$\bbbi$ groups.

In the following we fix a group $G$, its representation $\sigma$ and vary $\tau$ through all possible irreducible projective $G$-representations. %In particular we consider here projective $G$-representations.

%%%%%%%%%%%%%%%%%%%%%%%%%%%%%%%%%%%%%%%%%%%%%%%%%%%%%%%%%%%%%%%%%%%%

\subsection{Irreducible representations}
\label{sec:gap} 

We recall that our ultimate goal is to construct and classify all Automorphic Lie Algebras, $(\mf{g}(V)\otimes\mc{M}(\overline{\mathbbm{C}}))^G_\Gamma$, where $G$ is  a finite group, $\mc{M}(\overline{\mathbbm{C}})$ is the field of meromorphic functions on the Riemann sphere and where $\Gamma\subset\overline{\mathbbm{C}}$ is a $G$-orbit. Using the identification $\lambda=\nicefrac{X}{Y}\in\CCP^1$ the space $\mc{M}(\overline{\mathbbm{C}})$ is identified with the space of quotients of two homogeneous polynomials in $X$ and $Y$ of the same degree. M\"{o}bius transformations on $\lambda$ correspond to linear transformations on $X$ and $Y$ by the same matrix. Moreover, two matrices yield the same M\"{o}bius transformation if and only if they are scalar multiples of one another. Therefore, in order to cover all possibilities, we allow the action on $X$ and $Y$ to be projective. We recall that a faithful projective representation \(\sigma\) of \(G\) in \(\bbbc^2\) is a mapping from \(G\) to \(GL_2(\bbbc)\) obeying the following 
\begin{equation}
\label{eq:pr}
\sigma(g)\,\sigma(h)=c(g,h)\,\sigma(g\,h)\,,\quad \forall g\,,h\in G\,,
\end{equation}
where \(c(g,h)\,:\; G\times G\to\bbbc^{\ast}\) in (\ref{eq:pr}) is a %nontrivial
2-cocycle over \(\bbbc^{\ast}\)  (see for example \cite{MR1157729}), satisfying the cocycle identity 
\[
c(x,y)c(xy,z)=c(y,z)c(x,yz).
\]
If the cocycle is trivial  the projective representation \(\sigma\)  is a representation. 
Projective representations of $G$ correspond to representations of the Schur cover $G^\flat$ of $G$ in $\SL_2(\bbbc)$. We define the Schur cover $G^\flat$ of $G$ in $\SL_2(\bbbc)$ as the preimages of $G\subset \PSL_2(\bbbc)$, under the canonical projection $\pi:\SL_2(\bbbc)\rightarrow\PSL_2(\bbbc)$:
%%%%%%%%
\[
G^\flat=\pi^{-1}G\,.
\]
Alternatively, this group can be defined by the presentation
\[
G^\flat=\langle \gal, \gbe, \gga\;|\;\gal^{d_G}=\gbe^3=\gga^2=\gal\gbe\gga\rangle,
\]
cf.~\cite{suter2007quantum}, where $d_G=3,4$ and $5$ for $\bbbt$, $\bbbo$ and $\bbbi$, respectively.
We can readily see that $\gal\gbe\gga$ is a central element because it commutes with each generator, e.g.~$\gal (\gal\gbe\gga)=\gal \gal^{d_G}=\gal^{d_G} \gal=(\gal\gbe\gga)\gal$. If $G^\flat$ is nonabelian then this is the only nontrivial central element and represented by minus the identity matrix in $\SL_2(\bbbc)$. In particular it has order $2$ and the projection $\pi$ maps it to the identity.
Another presentation is given by \[r=\gal,\qquad s=\gga.\]
Then $\gbe=\gal^{-1}(\gal\gbe\gga)\gga^{-1}=\gal^{-1}(\gga^2)\gga^{-1}=\gal^{-1}\gga=r^{-1}s$
 and we obtain
\[
G^\flat=\langle r,s\;|\;r^{d_G}=(r^{-1}s)^3=s^2\rangle.
\]
In Appendix \ref{app:cover} we give an explicit construction of the Schur cover $G^\flat$ we work with, for completeness.

From a computational point of view it is more convenient to work with representations, rather than projective  representations. For example, in order to use GAP to compute generating elements, character tables (Sections \ref{sec:t}--\ref{sec:i}) and Molien functions (Section \ref{sec:Molien}), one needs to replace the projective representation by a representation.

Linear representations of $\bbbt^\flat$, $\bbbo^\flat$, $\bbbi^\flat$ can be easily computed by GAP (see Sections \ref{sec:t} to \ref{sec:i} for further details); in what follows we label irreducible representations (irreps) by $G^\flat_i$, where $G$ is one of the $\bbbt\bbbo\bbbi$ groups, and we drop $\flat$ when the representation is also a linear representation of $G$. we denote this set as $\textrm{Irr}(G^\flat)$.
The representations with a $\flat$-index are those with nontrivial cocycle
%i.e. nontrivial cocycle; see discussion following Table \ref{tab:charT} and Appendix \ref{app:cover})
(see Tables \ref{tab:charT}, \ref{tab:charO}, \ref{tab:charY}); these are the representations which are not linear representations of $G$.

\begin{Definition}[Natural representation]
\label{def:natural}
A monomorphism 
\[
\sigma:G^\flat\rightarrow \textrm{SL}_2(\bbbc)
\] is called a \emph{natural representation}.
\end{Definition}
The chosen natural representations of the $\bbbt\bbbo\bbbi$ groups are underlined
in the Tables \ref{tab:charT}, \ref{tab:charO} and \ref{tab:charY}.

%%%%%%%%%%%%%%%%%%%%%%%%%%%%%%%%%%%%%%%%%%%%%%%%%%%%%%%%%%%%%%%%%

\subsubsection{Dynkin diagrams of the irreducible representations}
\label{sec:dynkin}
%Before proceeding with a list of irreducible $G^\flat$-representations, let us recall here some results from \cite{MR909219}. We will closely follow the notations in \cite{MR909219} for the purpose of the diagrams, so for this section only we rename the natural representation $\sigma$ with $x$ and the linear representations $G^\flat_i$ with $u_{i}$.  
Before proceeding with a list of irreducible $G^\flat$-representations, let us recall here some results from \cite{MR909219}. Let \(\bbbt^\flat,\bbbo^\flat,\bbbi^\flat\) be the double covers of the $\bbbt\bbbo\bbbi$ groups; they are characterised by the solutions of the equation
\begin{equation}
\label{eq:abc}
\frac{1}{a}+\frac{1}{b}+\frac{1}{c}=1\,,\qquad\,a,\,b,\,c\in\bbbn\,.
\end{equation}
The solutions are well known, they are $(6, 3, 2)$ for $\bbbi^\flat$, $(4, 4, 2)$ for $\bbbo^\flat$ and $(3, 3, 3)$ for $\bbbt^\flat$, up to permutation.

We will closely follow the notations in \cite{MR909219}, so for the purpose of the diagrams we rename the natural representation $\sigma$ with $x$ and denote by $x_h$ the $h$-th symmetric power of $x$. Notice that $x_0$ is the trivial representation and $x_1=x$ the natural representation.
The Clebsch-Gordan formula from classical invariant theory is
\begin{equation}
\label{eq:CG}
x\otimes x_h =x_{h-1}\oplus x_{h+1}\,\quad h\geq 1\,.
\end{equation}
Let  $x_0$, $y$ and $z$ be the three different endpoints of the Dynkin diagram of affine type (this is also called extended Dynkin diagram, as it contains the trivial representation $x_0$ - see Figure \ref{dynkin}). 
The diagram is formed by taking the irreducible representations as nodes.
Every representation is connected to those irreducible representations that occur in
the decomposition of its tensor product with the natural representation into irreducible representations.
Let $a\geq 2$ be such that $x_0$, $x_1$,...,$x_{a-1}$ are irreducible as $G^\flat$-modules and $x_a$ is not, then $x_{a-1}$ is called branch point (of the Dynkin diagram). There are integers $b,\, c\geq 2$ such that the two other branches of the Dynkin diagram are given by $y, x_1 y,\cdots,x_{b-2} y$ and $z,x_1 z,\cdots,x_{c-2}z $, respectively and it follows that $x_{a}$ splits into two irreducibles according to the rule
\[
x\otimes x_{a-1}=x_{a-2}\oplus x_{a}=x_{a-2}\oplus x_{b-2}\otimes y\oplus x_{c-2}\otimes z\,
\]
(see \cite{MR909219} for details). The branch point is characterised by $x_{a-1}=x_{b-1}\otimes y=x_{c-1}\otimes z$ and $(a, b, c)$ satisfy equation (\ref{eq:abc}).

\begin{figure}[h]
\begin{center}
\begin{tikzpicture}
  \path node at ( 0,0) [shape=circle,draw,label=270: $x_0$] (first) {$$}
  	node at ( 2,0) [shape=circle,draw,label=270: $x_1$] (second) {$$}
  	node at ( 4,0) [shape=circle,draw,label=270: $x_i$] (third) {$$}
  	node at (6,0) [shape=circle,draw,label=270: $x_{a-2}$] (fourth) {$$}
  	node at ( 8,0) [shape=circle,draw,label=270: $x_{a-1}$] (fifth) {}
  	node at (10,0) [shape=circle,draw,label=270: $x_{b-2}\otimes y$] (sixth) {}
  	node at ( 12,0) [shape=circle,draw,label=270: $x_{i}\otimes y$] (seventh) {}
  	node at (14,0) [shape=circle,draw,label=270: $y$] (eight) {}
  	node at (8,2) [shape=circle,draw,label=0: $x_{c-2}\otimes z$] (nine) {}
  	node at (8,4) [shape=circle,draw,label=0: $z$] (ten) {};
	\draw (first.east) to node [sloped,below]{} (second.west);
	\draw[dotted] (second.east) to node [sloped,above]{} (third.west);
	\draw[dotted] (third.east) to node [sloped,above]{} (fourth.west);
	\draw (fourth.east) to node [sloped,above]{} (fifth.west);
	\draw (fifth.east) to node [sloped,above]{} (sixth.west);
	\draw[dotted] (sixth.east) to node [sloped,above]{} (seventh.west);
	\draw[dotted] (seventh.east) to node [sloped,above]{} (eight.west);
	\draw (fifth.north) to node [sloped,above]{} (nine.south);
	\draw[dotted] (nine.north) to node [sloped,above]{} (ten.south);
\end{tikzpicture}
\end{center}
\caption{Affine Dynkin diagrams of $G^\flat$, where $G$ is one of the $\bbbt\bbbo\bbbi$ groups.
The dimensions of the irreducibles are $1, 2,\ldots, a$; $\nicefrac{a}{b},2\nicefrac{a}{b},\ldots,(b-1)\nicefrac{a}{b}$; $\nicefrac{a}{c},2\nicefrac{a}{c},\ldots,(c-1)\nicefrac{a}{c}$.}
\label{dynkin}
\end{figure}
%In the picture the Dynkin diagrams of $G^\flat$, where $G$ is one of the $\bbbt\bbbo\bbbi$ groups.
%The dimensions of the irreducibles are $1, 2,\ldots, a; \nicefrac{a}{b},2\nicefrac{a}{b},\ldots,(b-1)\nicefrac{a}{b};\nicefrac{a}{c},2\nicefrac{a}{c},\ldots,(c-1)\nicefrac{a}{c}$.

%%%%%%%%%%%%%%%%%%%%%%%%%%%%%%%%%%%%%%%%%%%%%%%%%%%%%%%%%%%%%%%%%

\subsubsection{Tetrahedral group $\bbbt$}
\label{sec:t}
A regular tetrahedron is a Platonic solid composed of  four equilateral triangular faces, three of which meet at each vertex. It has four vertices and six edges. A regular tetrahedron has twelve rotational (or orientation-preserving) symmetries; the set of orientation-preserving symmetries forms a group referred to as \(\bbbt\), isomorphic to the alternating subgroup \(\mc{A}_4\). As an abstract group it is generated by two elements, \(\gs\) and \(\gr\), satisfying the identities $\gs^{3}=\gr^{2}=(\gs\,\gr)^3=id\,.$

In Table \ref{tab:charT} the character table of the Schur cover $\bbbt^\flat=\langle r,s\;|\;r^{3}=(r^{-1}s)^3=s^2\rangle$ in $SL_2 (\bbbc)$ (see Section \ref{sec:gap}) is given.
%The Schur cover $\bbbt^\flat=\langle r,s\;|\;r^{3}=(r^{-1}s)^3=s^2\rangle$ in $SL_2 (\bbbc)$ (see Section \ref{sec:gap}) has the following character table:
%\begin{center}
%\begin{table}[h!] 
%\begin{center}
%\begin{tabular}{|c|c|c|c|c|c|c|c|c|c|c|} \hline
%irrep & Dynkin &   $id$ & $ [r^2]$ &$ [s]$ &$ [s^2]$ &  $[r]$&$  [sr^2]$&$  [s^2r]$&$\Delta$&$\iota$\\
%\hline
%\hline
%$\bbbt_1$ & $x_0$ &  1 &  1 & 1 & 1 &  1&  1&  1&$\bbbt_1$&$1$\\
%$\bbbt_2$ &  $y$ & 1 & A & 1 & 1 & /A & A&/A&$\bbbt_2$&$0$\\
%$\bbbt_3$ &  $z$ &  1 &  /A & 1 & 1 &  A &/A & A&$\bbbt_3$&$0$\\
%$\underline{\bbbt_4^\flat}$ &  $x_1$ &     2  &-1 & . &-2 & -1&  1&  1&$\bbbt_1$&$-1$\\
%$\bbbt_5^\flat$ & $x_1\otimes z$ & 2& -/A & .& -2 & -A& /A & A&$\bbbt_2$&$0$\\
%$\bbbt_6^\flat$ &   $x_1\otimes y$ &  2 & -A & . &-2 &-/A & A &/A&$\bbbt_3$&$0$\\
%$\bbbt_7$ &   $x_2$ &   3  & .& -1&  3 &  .&  .&  .&$\bbbt_1$&$1$\\
%\hline 
%\end{tabular}
%\end{center}
%\caption{Character table for \(\bbbt^\flat\), \(A = \omega_3^2\), \(/A=\omega_3\), in GAP notation.}
%\label{tab:charT}
%\end{table}
%\end{center}
The first column 
%in Table \ref{tab:charT} 
contains the seven irreducible representations of \(\bbbt^\flat\); they can be obtained by e.g.~GAP \cite{GAP}; the irreducible representation $\bbbt_4^\flat$ is the \emph{natural representation} (see Definition \ref{def:natural}). 
The representations with a $\flat$-index are those with nontrivial cohomology (see Appendix \ref{app:cover}); the 
 $\flat$ is dropped when the representation is also a linear representation of $\bbbt$. The second column contains the same representations in the language of \cite{MR909219} to allow drawing the Dynkin diagram as in Section \ref{sec:dynkin}. The next columns list the conjugacy classes and the corresponding values of the characters, following the GAP notation, where a dot indicates the zero and where \(A = \omega_3^2\), \(/A=\omega_3\). Here, and in what follows, $\omega_n=\exp{\nicefrac{2\pi i}{n}}$, so $\omega_3$ is a primitive cubic root of unity. The penultimate column contains determinants of the representation (see Remark \ref{rm:det}).
% ; 
%one computes the values of \(\Delta\) for a given \(\bbbt_n\) from the table 
%by looking up the entry \(\bbbt_m\) in the \(\bbbt_n\)-row and the \(\Delta\)-column.
%Then one looks in the \(\bbbt_m\)-row of Table \ref{tab:charT} and \([g]\)-column to find the value of the corresponding \(\Delta\). 
Determinants have been included since they suggest the pairing of relative invariants in order to get invariants from transvection (Section \ref{sec:trans}) and (for future reference) play a role in the determination of the building blocks of \(\mf{sl}(V)\).
%The reason why determinants are interesting will be clear in Section \ref{sec:trans}, following the definition of transvectants.
Finally, the last column contains the value of the Frobenius-Schur indicator \(\iota\), computed by \(\iota_\chi=\frac{1}{|G|}\sum_{g\in G}\chi(g^2)\). Complex irreducible representations with Frobenius-Schur indicator $1$, $0$ or $-1$ are respectively known as representations of real type, complex type or quaternionic type \cite{ful91a}. This last column is included here purely for future reference, as it gives information about the existence of irreducible \(\mf{so}\) and \(\mf{sp}\) representations.
\begin{center}
\begin{table}[h!] 
\begin{center}
\begin{tabular}{|c|c|c|c|c|c|c|c|c|c|c|} \hline
irrep & Dynkin &   $id$ & $ [r^2]$ &$ [s]$ &$ [s^2]$ &  $[r]$&$  [sr^2]$&$  [s^2r]$&$\Delta$&$\iota$\\
\hline
\hline
$\bbbt_1$ & $x_0$ &  1 &  1 & 1 & 1 &  1&  1&  1&$\bbbt_1$&$1$\\
$\bbbt_2$ &  $y$ & 1 & A & 1 & 1 & /A & A&/A&$\bbbt_2$&$0$\\
$\bbbt_3$ &  $z$ &  1 &  /A & 1 & 1 &  A &/A & A&$\bbbt_3$&$0$\\
$\underline{\bbbt_4^\flat}$ &  $x_1$ &     2  &-1 & . &-2 & -1&  1&  1&$\bbbt_1$&$-1$\\
$\bbbt_5^\flat$ & $x_1\otimes z$ & 2& -/A & .& -2 & -A& /A & A&$\bbbt_2$&$0$\\
$\bbbt_6^\flat$ &   $x_1\otimes y$ &  2 & -A & . &-2 &-/A & A &/A&$\bbbt_3$&$0$\\
$\bbbt_7$ &   $x_2$ &   3  & .& -1&  3 &  .&  .&  .&$\bbbt_1$&$1$\\
\hline 
\end{tabular}
\end{center}
\caption{Character table for \(\bbbt^\flat\), \(A = \omega_3^2\), \(/A=\omega_3\), in GAP notation.}
\label{tab:charT}
\end{table}
\end{center}

A concrete projective representation of $\bbbt_4^\flat$ is given by
\begin{equation}
\label{eq:Tnrep}
\sigma(r)= \begin{pmatrix} -1 -\omega_3&  0 \\
1 &  \omega_3 \end{pmatrix}\,,
\quad 
\sigma(s)= \begin{pmatrix} -\omega_3& 1+\omega_3\\
1+\omega_3 &  \omega_3 \end{pmatrix}\,.
\end{equation}
One could in principle make one of the generator diagonal, but we rather work with GAP given representations.
Table \ref{tab:charT} suggests the following field extension: \(\splitk=\bbbq[\omega_3]/(1+\omega_3+\omega_3^2)\); the nonzero elements are denoted by \(\splitk^\ast\). 

%%%%%%%%%%%%%%%%%%%%%%%%%%%%%%%%%%%%%%%%%%%%%%%%%%%%%%%%%%%%%%%%%

\subsubsection{Octahedral group \(\bbbo\)}\label{sec:o}
A regular octahedron is a Platonic solid composed of eight equilateral triangles, four of which meet at each vertex; it has six vertices and eight edges.
A regular octahedron has twenty four rotational (or orientation-preserving) symmetries. A cube has the same set of symmetries, since it is its dual. The group of orientation-preserving symmetries is denoted by \(\bbbo\) and it is isomorphic to \(\mc{S}_4\), or the group of permutations of four objects, since there is exactly one such symmetry for each permutation of the four pairs of opposite sides of the octahedron. As an abstract group it is generated by two elements, \(\gs\) and \(\gr\), satisfying the identities $\gs^{4}=\gr^{2}=(\gs\,\gr)^3=id\,.$
%
%The Schur cover $\bbbo^\flat=\langle r,s\;|\;r^{4}=(r^{-1}s)^3=s^2\rangle$ in $SL_2 (\bbbc)$ (see Section \ref{sec:gap}) 
%has the following character table:
\begin{center}
\begin{table}[h!] 
\begin{center}
\begin{tabular}{|c|c|c|c|c|c|c|c|c|c|c|c|} \hline
irrep &   Dynkin &  $id$ &  $[s]$ & $[r^2s^2]$ & $[r^2]$ &  $[s^2]$&  $[r]$&  $[rs]$&$[r^3]$&$\Delta$&$\iota$\\
\hline
\hline
$\bbbo_1$& $x_0$ &1  &  1 & 1 &  1&  1 & 1&  1&   1&$\bbbo_1$&$1$\\
$\bbbo_2$& $y$ & 1  &-1  & 1 & 1 & 1  &-1& 1 &-1&$\bbbo_2$&$1$\\
$\bbbo_3 $ &  $z$ & 2&  .& -1&  2&  2&  .& -1&  .&$\bbbo_2$&$1$\\
$\underline{\bbbo}_4^\flat $&  $x_1$ &  2 & .& -1 & . &-2 & A & 1& -A&$\bbbo_1$&$-1$\\
$\bbbo_5^\flat $  & $x_1\otimes y$ & 2 & .& -1 & .& -2& -A&  1&  A&$\bbbo_1$&$-1$\\
$\bbbo_6 $& $x_2\otimes y$ &    3&  1 & .& -1&  3& -1 & .& -1&$\bbbo_2$&$1$\\
$\bbbo_7 $&   $x_2$ &  3 &-1 & . &-1 & 3 & 1 & .&  1&$\bbbo_1$&$1$\\
$\bbbo_8^\flat $&   $x_3$ &  4 & .&  1&  .& -4&  .& -1 & .&$\bbbo_1$&$-1$\\
\hline 
\end{tabular}
\end{center}
\caption{Character table for \(\bbbo^\flat\), \(A = -\omega_8+\omega_8^3=-\sqrt{2}\), in GAP notation.}
\label{tab:charO}
\end{table}
\end{center}
The character table of the Schur cover $\bbbo^\flat=\langle r,s\;|\;r^{4}=(r^{-1}s)^3=s^2\rangle$ in $SL_2 (\bbbc)$ (see Section \ref{sec:gap}) is given in Table \ref{tab:charO}. The irreducible representation $\bbbo_4^\flat$ is the natural representation.

The concrete projective representation we work with is given by
\begin{equation}
\label{eq:Onrep}
\sigma(r)=\begin{pmatrix}  \omega_8^2 & -1\\
 \omega_8+\omega_8^3&  -\omega_8-\omega_8^2+\omega_8^3\end{pmatrix}\,,\quad
\sigma(s)=\begin{pmatrix} 1& - \omega_8 +\omega_8^2 + \omega_8^3\\
  \omega_8 + \omega_8^2 - \omega_8^3& - 2\end{pmatrix}\,.\\
\end{equation}
As in the previous case, the chosen field is determined by the occurrence of roots of unity in the representation matrices.
In the $\bbbo^\flat$ case both $\omega_3$ and $\omega_8$ occur (e.g.~$\omega_3$ occurs in $\bbbo_8$), leading to a mix of values of roots of unity and hence to $\omega_{24}$. The minimal polynomial is then the one for $\omega_6$ 
but expressed for $\omega_{24}$. Hence the field extension in this case is \(\splitk=\bbbq[\omega_{24}]/(\omega^8_{24}-\omega^4_{24}+1)\).

%%%%%%%%%%%%%%%%%%%%%%%%%%%%%%%%%%%%%%%%%%%%%%%%%%%%%%%%%%%%%%%%%
\subsubsection{Icosahedral group  \(\bbbi\)}
\label{sec:i}

An icosahedron is a convex regular polyhedron (a Platonic solid) with twenty triangular faces, thirty edges and twelve vertices. 
A regular icosahedron has sixty rotational (or orientation-preserving) symmetries; the set of orientation-preserving symmetries forms a group referred to as $\bbbi$; $\bbbi$ is isomorphic to $\mc{A}_5$, the alternating group of even permutations of five objects. As an abstract group it is generated by two elements, \(\gs\) and \(\gr\), satisfying the identities $\gs^{5}=\gr^{2}=(\gs\,\gr)^3=id\,.$

The Schur cover $\bbbi^\flat=\langle r,s\;|\;r^{5}=(r^{-1}s)^3=s^2\rangle$ in $SL_2 (\bbbc)$ (see Section \ref{sec:gap}) 
has the following character table:
\begin{center}
\begin{table}[h!] 
\begin{center}
\begin{tabular}{|c|c|c|c|c|c|c|c|c|c|c|c|c|} \hline
irrep &   Dynkin & $id$ &  $[r]$ & $[r^2]$ & $[rs^3]$ &  $[s]$ &  $[rs]$ &  $[rs^2]$ & $[s^2]$ & $[r^2s^2]$ &$\Delta$&$\iota$\\
\hline
\hline
$\bbbi_1$ & $x_0$ & 1 & 1 &  1&  1 & 1&  1&   1&  1&   1&$\bbbi_1$&$1$\\
$\underline{\bbbi}_2^\flat$ &$x_1$ &  2 & A & *A&  1 & .& -1&  -A& -2& -*A&$\bbbi_1$&$-1$\\
$\bbbi_3^\flat$ & $y$ & 2 &*A &  A&  1 & .& -1& -*A& -2&  -A&$\bbbi_1$&$-1$\\
$\bbbi_4$  & $z$ & 3 &-*A&  -A&  .& -1&  .& -*A&  3&  -A&$\bbbi_1$&$1$\\
$\bbbi_5$ & $x_2$ & 3 & -A& -*A&  .& -1&  .&  -A&  3& -*A&$\bbbi_1$&$1$\\
$\bbbi_6$  & $x_1 \otimes y$ & 4 & -1&  -1&  1&  .&  1&  -1&  4&  -1&$\bbbi_1$&$1$\\
$\bbbi_7^\flat$ & $x_3$ & 4 & -1&  -1& -1&  .&  1&   1& -4&   1&$\bbbi_1$&$-1$\\
$\bbbi_8$   & $x_4$ & 5 &  .&   .& -1&  1& -1&   .&  5&   .&$\bbbi_1$&$1$\\
$\bbbi_9^\flat$ & $x_5$ &  6 &  1&   1&  .&  .&  .&  -1& -6&  -1&$\bbbi_1$&$-1$\\
\hline 
\end{tabular}
\end{center}
\caption{Character table for \(\bbbi^\flat\), \(A = \omega_5+\omega_5^4\), \(*A=1-A=A^2=-1/A\), in GAP notation.}
\label{tab:charY}
\end{table}
\end{center}
The concrete projective representation we work with is given by
\begin{equation}
\label{eq:Ynrep}
\sigma(r)= \begin{pmatrix} \omega_5  &  -\omega_5^2-\omega_5^3  \\
 0   & \omega_5^4  \end{pmatrix},
\quad \sigma(s)= \begin{pmatrix} -\omega_5^2 & \omega_5^4 \\
 -1-\omega_5 &   \omega_5^2  \end{pmatrix}
\end{equation}
and  \(\splitk=\bbbq[\omega_5]/(1+\omega_5+\omega_5^2+\omega_5^3+\omega_5^4)\).

%%%%%%%%%%%%%%%%%%%%%%%%%%%%%%%%%%%%%%%%%%%%%%%%%%%%%%%%%%%%%%%%%%%%%%%%%%%%%%
\subsubsection{Decomposition of $\mf{sl}(V)$ into irreducible representations}
\label{sec:decomposition}
We compute the decomposition of $\mf{sl}(V_j)\cong V_j\otimes V_j^\ast-V_1$ into irreducible representations using GAP, where $V_1$ is the trivial representation and list them in Tables \ref{tab:decompTsl} -- \ref {tab:decompYsl}.
This is the first moment we specialise to \(\mf{sl}(V)\); we remark that similar decompositions exist for \(\mf{so}(V)\) and \(\mf{sp}(V)\) and this paper contains all the necessary information to analyse these cases as well. 
%Some data in the tables have been put there with these cases in mind so that they can be used in the future.
The irreducible representations $V_j$ are labelled using the group name, so $\bbbt_1$ corresponds to the first irreducible representation in the list of $\bbbt^\flat$ (see Tables \ref{tab:charT} -- \ref {tab:charY}).

\begin{center}
\begin{table}[H] 
\begin{center}
\begin{tabular}{|c|c|c|} \hline
$\mf{sl}(\bbbt_j)$&dim & decomposition\\
\hline
\hline
$\mf{sl}(\bbbt_4^\flat)$& 3&  $\bbbt_7   $\\
$\mf{sl}(\bbbt_5^\flat)$& 3&  $\bbbt_7    $\\
$\mf{sl}(\bbbt_6^\flat)$& 3&  $\bbbt_7    $\\
$\mf{sl}(\bbbt_7)$& 8&  $\bbbt_2\oplus\bbbt_3\oplus 2\bbbt_7   $\\
\hline 
\end{tabular}
\end{center}
\caption{Decomposition of $\mf{sl}(\bbbt_j^\flat)$.}
\label{tab:decompTsl}
\end{table}
\end{center}
\begin{center}
\begin{table}[H] 
\begin{center}
\begin{tabular}{|c|c|c|} \hline
$\mf{sl}(\bbbo_j)$,&dim & decomposition\\
\hline
\hline
$\mf{sl}(\bbbo_3)$& 3&  $   \bbbo_{2}\oplus\bbbo_3  $\\
$\mf{sl}(\bbbo_4^\flat)$& 3&  $\bbbo_7   $\\
$\mf{sl}(\bbbo_5^\flat)$& 3&  $\bbbo_7    $\\
$\mf{sl}(\bbbo_6)$& 8&  $\bbbo_3\oplus\bbbo_6\oplus\bbbo_7    $\\
$\mf{sl}(\bbbo_7)$& 8&  $\bbbo_3\oplus\bbbo_6\oplus\bbbo_7   $\\
$\mf{sl}(\bbbo_8^\flat)$& 15&  $\bbbo_2\oplus\bbbo_3\oplus2\bbbo_6\oplus2\bbbo_7 $\\
\hline 
\end{tabular}
\end{center}
\caption{Decomposition of $\mf{sl}(\bbbo_j^\flat)$.}
\label{tab:decompOsl}
\end{table}
\end{center}
\begin{center}
\begin{table}[H] 
\begin{center}
\begin{tabular}{|c|c|c|} \hline
$\mf{sl}(\bbbi_j)$&dim& decomposition\\
\hline
\hline
$\mf{sl}(\bbbi_2^\flat)$& 3&  $  \bbbi_{5} $\\
$\mf{sl}(\bbbi_3^\flat)$& 3&  $   \bbbi_{4}  $\\
$\mf{sl}(\bbbi_4)$& 8&  $\bbbi_4\oplus\bbbi_8     $\\
$\mf{sl}(\bbbi_5)$& 8&  $\bbbi_5\oplus\bbbi_8    $\\
$\mf{sl}(\bbbi_6)$& 15&  $\bbbi_4\oplus\bbbi_5\oplus\bbbi_6\oplus\bbbi_8  $\\
$\mf{sl}(\bbbi_7^\flat)$& 15&  $\bbbi_4\oplus\bbbi_5\oplus\bbbi_6\oplus\bbbi_8   $\\
$\mf{sl}(\bbbi_8)$& 24&  $\bbbi_4\oplus\bbbi_5\oplus2\bbbi_6\oplus2\bbbi_8
  $\\
$\mf{sl}(\bbbi_9^\flat)$& 35&  $2\bbbi_4\oplus2\bbbi_5\oplus2\bbbi_6\oplus3\bbbi_8
  $\\
\hline 
\end{tabular}
\end{center}
\caption{Decomposition of $\mf{sl}(\bbbi_j^\flat)$.}
\label{tab:decompYsl}
\end{table}
\end{center}

%%%%%%%%%%%%%%%%%%%%%%%%%%%%%%%%%%%%%%%%%%%%%%%%%%%%%%%%%%%%%%%%%%
\subsection{Molien functions}
\label{sec:Molien}
In the search for invariants in $\mf{sl}(V)\otimes\bbbc[X,Y]$ we use the decomposition of $\mf{sl}(V)$ in the irreducible representations listed in Tables \ref{tab:decompTsl} -- \ref {tab:decompYsl}:
\[
\mf{sl}(V)=\bigoplus_k \langle \mf{sl}(V),V_k \rangle  V_k\,.
\]
This reduces the problem to describing $(V_k \otimes\bbbc[X,Y])^{G^\flat}$. The generating functions of invariants in $V_k \otimes\bbbc[X,Y]$ can be computed using the following theorem (See \cite[Section 4.3]{MR1328644}).

\begin{Theorem}[Molien, \cite{mol98a}]
Let \(\sigma\,: \: G^{\flat} \hookrightarrow GL_2(\bbbc)\) be a representation defining an action of $G^\flat$ on $\bbbc[X,Y]$ by $g\cdot p(X,Y)=p(\sigma(g^{-1})(X,Y))$, $g\in G^{\flat}$, $p(X,Y)\in\bbbc[X,Y]$, and let $\chi_k$ be the character of $V_k$.
Then the Poincar\'{e} series of invariants in $V_k \otimes \bbbc[X,Y] $ is given by
\begin{equation}
\label{eq:molien}
M(({V_k}\otimes\bbbc[X,Y])^{G^{\flat}},t)=\frac{1}{|G^{\flat}|}\sum_{g\in G^{\flat}} \frac{\chi_k (g)}{\det(1-\sigma(g^{-1})\,t)}\,.
\end{equation}
We will call this the Molien function of the irreducible representation $V_k$.
\end{Theorem}

Recall the irreducible representations $x_i$, $i=0,\ldots,a-1$, $x_i\otimes y$, $i=0,\ldots,b-2$ and $x_i\otimes z$, $i=0,\ldots,c-2$ from Section \ref{sec:dynkin}. The following holds (see  \cite{MR909219})
\begin{equation}
\label{eq:molien2}
M(\square,t)=\frac{N(\square,t)}{(1-t^{2a})(1-t^{4a-4})}
\end{equation}
with $N(\square,t)$ defined by
\begin{eqnarray}
\label{eq:irrrep_molien}
N(({x_i}\otimes\bbbc[X,Y])^{G^{\flat}},t)&=&t^i +t^{6a -6 -i}+(t^{2a-i}+t^{4a-4-i})\frac{1-t^{2i}}{1-t^2}\,,\quad i=0,\ldots,a-1\,,\nonumber\\
N(({x_i\otimes y}\otimes\bbbc[X,Y])^{G^{\flat}},t)&=&t^{a+b-i-2} (1+t^{2a -2})\frac{1-t^{2a}}{1-t^{2b}}\frac{1-t^{2i+2}}{1-t^2}\,,\quad i=0,\ldots,b-2\,,\\
N(({x_i\otimes z}\otimes\bbbc[X,Y])^{G^{\flat}},t)&=&t^{a+c-i-2} (1+t^{2a -2})\frac{1-t^{2a}}{1-t^{2c}}\frac{1-t^{2i+2}}{1-t^2}\,,\quad i=0,\ldots,c-2\,.\nonumber
\end{eqnarray}

\begin{Example}
\label{ex:primary_inv}
%Consider the Poincar\'{e} series of invariants in $\bbbt_1\otimes \bbbc[X,Y]$, where $\bbbt_1$ coincides with $x_0$ in the notation above. The Dynkin diagram of $\bbbt^\flat$ is 
Consider the Poincar\'{e} series of invariants in $\bbbt_1\otimes \bbbc[X,Y]$, with $x_0$ in the notation above. The affine Dynkin diagram of $\bbbt^\flat$, where $\bbbt_1$ coincides with $x_0$, is 

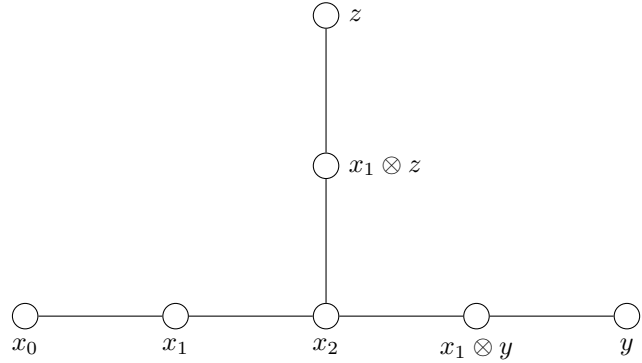
\begin{center}
\begin{tikzpicture}
  \path node at ( 0,0) [shape=circle,draw,label=270: $x_0$] (first) {$$}
  	node at ( 2,0) [shape=circle,draw,label=270: $x_1$] (second) {$$}
  	node at ( 4,0) [shape=circle,draw,label=270: $x_{2}$] (third) {$$}
  	node at ( 6,0) [shape=circle,draw,label=270: $x_{1}\otimes y$] (fourth) {}
  	node at (8,0) [shape=circle,draw,label=270: $y$] (fifth) {}
  	node at (4,2) [shape=circle,draw,label=0: $x_{1}\otimes z$] (sixth) {}
  	node at (4,4) [shape=circle,draw,label=0: $z$] (seventh) {};
	\draw (first.east) to node [sloped,below]{} (second.west);
	\draw (second.east) to node [sloped,above]{} (third.west);
	\draw (third.east) to node [sloped,above]{} (fourth.west);
	\draw (fourth.east) to node [sloped,above]{} (fifth.west);
	\draw (third.north) to node [sloped,above]{} (sixth.south);
	\draw (sixth.north) to node [sloped,above]{} (seventh.south);
\label{dynkin_ex}
\end{tikzpicture}
\end{center}

and it is characterised by $(a=3,b=3,c=3)$ (see Section \ref{sec:dynkin}). It follows from (\ref{eq:irrrep_molien}) that
\[
N((x_0\otimes \bbbc[X,Y])^{\bbbt^\flat},t)=1+t^{12},
\]
thus
\[
M((\bbbt_1\otimes \bbbc[X,Y])^{\bbbt^\flat},t)=\frac{1+t^{12}}{(1-t^6)(1-t^8)}.
\]
\end{Example}

Using the scheme illustrated above (and the natural representation $\sigma=x_1$) we rewrite the Molien function for the irreducible representations in (\ref{eq:irrrep_molien}) in a form which is relevant for the computations of the generators of the invariants in  $V_k\otimes \bbbc[X,Y]$ (see Tables \ref{tab:irrMolienT}--\ref{tab:irrMolienY}). The choice of the powers in the denominators is determined by the existence of invariants at those degrees. These invariants are called the \emph{primary invariants},
while the ones corresponding to the terms in the numerator are called the  \emph{secondary invariants}.

Consider $\bbbt^\flat$ primary invariants at degree six and eight, so that \(M((\bbbt_k\otimes \bbbc[X,Y])^{\bbbt^\flat},t)=\frac{N}{(1-t^6)(1-t^8)}\). The numerators $N$ are then given in Table \ref{tab:irrMolienT}.

\begin{center}
\begin{table}[H] 
\begin{center}
\begin{tabular}{|c|c|c|l|} \hline
irrep&Dynkin &dim& Molien function numerator $N$\\
\hline
\hline
$\bbbt_1$&$x_0$ &1&$1+t^{12}$ \\
$\bbbt_2$&$y$ &1&$t^{4}+t^{8}$\\
$\bbbt_3$&$z$ &1&$t^{4}+t^{8}$\\
$\underline{\bbbt}_{4}^\flat$&$x_1$ & 2&$t+t^5+t^{7}+t^{11}$\\
$\bbbt_{5}^\flat$&$x_1\otimes z$ & 2&$t^{3}+t^5+t^{7}+t^{9}$\\
$\bbbt_{6}^\flat$&$x_1\otimes y$ & 2&$t^{3}+t^5+t^{7}+t^{9}$\\\
$\bbbt_{7}$&$x_2$ &3&$t^2+t^4+2t^{6}+t^8+t^{10}$\\
\hline
\end{tabular}
\end{center}
\caption{Molien functions of the irreducible representations $M((\bbbt_k\otimes \bbbc[X,Y])^{\bbbt^\flat},t)$.}
\label{tab:irrMolienT}
\end{table}
\end{center}
Similarly, considering $\bbbo^\flat$ and $\bbbi^\flat$ primary invariants at degree eight and twelve, and twelve and twenty, respectively, one  obtains Molien functions  \(M((\bbbo_k\otimes \bbbc[X,Y])^{\bbbo^\flat},t)\) and \(M((\bbbi_k\otimes \bbbc[X,Y])^{\bbbi^\flat},t)\) - see Tables \ref{tab:irrMolienO} and \ref{tab:irrMolienY} for the respective numerators.
\begin{center}
\begin{table}[H] 
\begin{center}
\begin{tabular}{|c|c|c|l|} \hline
irrep&Dynkin &dim& Molien function numerator $N$\\
\hline
\hline
$\bbbo_1$&$x_0$ & 1&  $   1+t^{18} $\\
$\bbbo_2$&$y$ & 1&  $   t^6+t^{12} $\\
$\bbbo_3$&$z$ &2&  $  t^4+t^8 +t^{10}+t^{14}$\\
$\underline{\bbbo}_4^\flat$&$x_1$ & 2&  $ t+t^{7}+t^{11}+t^{17} $\\
$\bbbo_5^\flat$&$x_1\otimes y$ &  2&  $  t^5+t^{7}+t^{11}+t^{13}$\\
$\bbbo_6$& $x_2\otimes y$ &3&  $ t^4+t^6+t^8+t^{10}+t^{12}+t^{14} $\\
$\bbbo_7$&$x_2$ & 3&  $ t^{2}+t^6+t^8+t^{10}+t^{12}+t^{16} $\\
$\bbbo_8^\flat$& $x_3$ &4&  $ t^3+t^{5}+t^7+2t^9+t^{11}+t^{13}+t^{15}$\\
\hline
\end{tabular}
\end{center}
\caption{Molien functions of the irreducible representations:  $M((\bbbo_k\otimes \bbbc[X,Y])^{\bbbo^\flat},t)$. }
\label{tab:irrMolienO}
\end{table}
\end{center}
\begin{center}
\begin{table}[H] 
\begin{center}
\begin{tabular}{|c|c|c|l|} \hline
irrep & Dynkin & dim   & Molien function numerator $N$\\
\hline
\hline
$\bbbi_1$&$x_0$ & 1&  $   1+t^{30} $\\
$\underline{\bbbi}_2^\flat$& $x_1$ &2&  $   t+t^{11}+t^{19}+t^{29} $\\
$\bbbi_3^\flat$& $y$ & 2&  $   t^7+t^{13}+t^{17}+t^{23} $\\
$\bbbi_4$& $z$ & 3&  $   t^6+t^{10}+t^{14}+t^{16}+t^{20}+t^{24} $\\
$\bbbi_5$& $x_2$ & 3&  $   t^2+t^{10}+t^{12}+t^{18}+t^{20}+t^{28}$\\
$\bbbi_6$& $x_1 \otimes y$ & 4&  $   t^6+t^8+t^{12}+t^{14}+t^{16}+t^{18}+t^{22}+t^{24}$\\
$\bbbi_7^\flat$& $x_3$ & 4&  $   t^3+t^9+t^{11}+t^{13}+t^{17}+t^{19}+t^{21}+t^{27}$\\
$\bbbi_8$& $x_4$ & 5&  $   t^4+t^8+t^{10}+t^{12}+t^{14}+t^{16}+t^{18}+t^{20}+t^{22}+t^{26}$\\
$\bbbi_9^\flat$& $x_5$ & 6&  $   t^5+t^7+t^9+t^{11}+t^{13}+2t^{15}+t^{17}+t^{19}+t^{21}+t^{23}+t^{25} $\\
\hline 
\end{tabular}
\end{center}
\caption{Molien functions of the irreducible representations:  $M((\bbbi_k\otimes \bbbc[X,Y])^{\bbbi^\flat},t)$.}
\label{tab:irrMolienY}
\end{table}
\end{center}
If one would like to compute the Molien function of a reducible representation, this is done by adding the Molien functions of the irreducible components with the corresponding multiplicities.

\section{Invariant matrices}
\label{sec:comp_inv_mat}
A brute-force computational approach towards invariant matrices consists in taking a general element in $\mf{g}(V)\otimes \bbbc[X,Y]$ of the degree dictated by the Molien function of $\mf{g}(V)$, and average over the group $G^\flat$. 
The Molien function of $\mf{g}(V)$ can be computed from the Molien functions of Tables \ref{tab:irrMolienT}-\ref{tab:irrMolienY} and the decompositions in Tables \ref{tab:decompTsl}-\ref{tab:decompYsl}, using the additive property of the Molien function. \\
This approach is however not very effective  computationally, as, for example, it would imply averaging an element in $\mf{sl}(\bbbi^\flat_9)\otimes \bbbc_{28}[X,Y]$ (that is, of \(X,Y\)-degree twentyeight).

Instead one could use the method of classical invariant theory to compute higher order invariants by transvection, starting from lower degree $\mf{g}(V)$-\emph{ground forms}, where $V$ is an irreducible $G^\flat$-representation. Hence, this reduces the problem to finding lower degree $\mf{g}(V)$-ground forms.
Moreover, transvection only involves multiplication and differentiation with respect to $X$ and $Y$, thus it is computationally very effective and easy to implement.

In order to systematically find the lower degrees $\mf{g}(V)$-ground forms one can use of the decomposition of $\mf{g}(V)$ into irreducible representations. The degree of the ground form is the lowest degree in the Taylor expansion at $t=0$ of the Molien function (see Section \ref{sec:Molien}) of the irreducible component in the decomposition (see Section \ref{sec:decomposition}); e.g.~the degree for the $\bbbi_8$-ground form is four, see Tables \ref{tab:decompYsl}
and \ref{tab:irrMolienY}; such ground form will be notated by $\mf{A}_8^4$, where the upper index indicates the degree while the lower one the corresponding $V$.%, see \ref{sec:Molien}. 
 The explicit projection on the irreducible components will be given in the next section.

\subsection{Fourier transform}
\label{sec:fourier}
Let \(W\) be a finite dimensional representation of a finite group \(G^\flat\) and let  $\{w_i\,|\,i=1,\ldots,\dim W\}$ be a basis of $W$. Then \( W \) can be decomposed as a direct sum of irreducible representations of \(G^\flat\)
as follows. 

Let \(V\) be such an irreducible  \(G^\flat\)--representation and let $\{v^i\,|\,i=1,\ldots,\dim V^\ast\}$ be a basis of \(V^\ast\). Let  \(\langle W, V\rangle\) be the multiplicity of $V$ in $W$ (that is, \(V\)  occurs as a direct summand in \(W\) \(\langle W, V\rangle\) times) and consider the space of invariants
\[
(W\otimes V^{\ast})^{G^\flat} =\{\eta^{k}\;|\; k=1,\ldots,\langle W, V\rangle \},\quad \eta^{k}=\sum_{i,j}\eta^{k}_{i,j} \,w_i \otimes v^j\,.
\]
The $\eta^k$ are traces of the basis of $V^\ast$ and its canonical dual basis, a basis for $V$. From the expression for $\eta^k$ we find $\langle W,V\rangle$ $V$-bases $\{v_j^k=\sum_{i}\eta_{i,j}^k w_i\;|\;j=1,\ldots,\dim V\}$, $k=1,...,\langle W, V\rangle$.

In practice we take a general element $\sum_{i,j} \zeta_{i,j} \,w_i \otimes v^j$ in $W\otimes V^\ast$ and require this element to be invariant under the action of the generators of $G^\flat$ to obtain elements $\eta^k=\sum_{i,j}\eta^{k}_{i,j} \,w_i \otimes v^j$.

If we now do the same construction for $U\otimes V$ we find $V^\ast$-bases in $U$. Taking the trace with each $V$-basis in $W$ results in  $\langle W, V\rangle \langle U, V^\ast\rangle $ linearly independent elements of $(W\otimes U)^{G^\flat}$. The space spanned by these elements will be denoted by  $(W\otimes U)_V^{G^\flat}$. We have
\[(W\otimes U)^{G^\flat}=\bigoplus_{V\in\textrm{Irr}G^\flat}(W\otimes U)_V^{G^\flat}\]

We return to the original problem of finding invariant matrices of degree $d$ in $\mf{sl}(V)\otimes\splitk[X,Y]$. To this end we apply the above construction to the $G^\flat$-representations $\mf{sl}(V)$ and \(\splitk_d[X,Y]\) and obtain  \((\mf{sl}(V)\otimes\splitk_d[X,Y])_{V'}^{G^\flat}\), with $V'\in\textrm{Irr}(G^\flat)$.

%%%%%%%%%%%%%%%%%%%%%%%%%%%%%%%%%%%%%%%%%%%%%%%%%%%%%%%%%%%%%%%%%%%%%%%%%%%%%%%%%%%%%%%%%%%%%%%%%%%%%%%%%%%%%%%%%%%%%%%%%%%%%%%%%%%%%%%%%%%%%%%%%%%%%%%%%%%%

\subsection{Transvectants}
\label{sec:trans} 
In classical invariant theory the basic computational tool is the \emph{transvectant}: given any two invariants (in the context of invariant theory these are called covariants), it is possible to construct a number of (possibly new) invariants by computing transvectants. As a simple example consider two linear forms \(aY+bX\), \(cY+dX\); their first transvectant is the determinant of the coefficients, i.e.~\(ad-cb\). 
A transformation on \((X,Y)\) induces a transformation on  \((a,b)\) such that \(aY+bX\) remains constant, and similarly for \((c,d)\). Then  \(ad-cb\) is invariant under the joint induced transformations on \((a,b)\)  and \((c,d)\).
Similarly, the discriminant \( a_0 a_2-a^{2}_{1}\) of a quadratic form \(a_0 Y^2+2 a_1 X Y+a_2 X^2\) is the second transvectant of the quadratic form with itself. 
While the transvectant language has been superseded by more general constructions,
working for all finite dimensional Lie algebras,
and sounds rather old-fashioned to present day algebraists,
it is still a very effective computational tool when it can be applied and it is easy to program. 
The only assumption one makes 
is that the group acts linearly and faithfully on \(\bbbc^2\), that the group elements are represented by matrices in $SL_{2}(\bbbc)$, as it is indeed the case for the natural representation $\sigma$ (see Definition \ref{def:natural}).
If one would like to replace \(\bbbc^2\) by a higher-dimensional space,
the transvectant mechanism is no longer available, but while the transvectant technique
is very efficient, the results in this paper could also have been obtained without transvectants, e.g.~using group averaging as mentioned at the beginning of Section \ref{sec:comp_inv_mat}.

In this section we will adapt the idea of transvection to compute invariant Lie algebras. We start from the classical work by Klein about automorphic functions 
and generalise it to the context of automorphic algebras. To do so, we need first to recall some definitions and facts about transvectants and generalise some of the concepts to the present set up.\\
Recall the Definition \ref{def:relativeinv} of relative invariant; in the literature, relative invariants are also called semi-invariants 
or covariants. 
\begin{Definition}[Polynomial ground form]
\label{def:poly_groundform}
A {\bf polynomial ground form} is a relative invariant polynomial \(\alpha\) of minimal degree. The divisor of zeros of such a polynomial is  an exceptional (or degenerate) $G$-orbit of minimal order.
\end{Definition}
\begin{Definition}[Ground form]
\label{def:groundform}
A {\bf ground form} is an invariant  \(\mf{A}\in \mathcal{V}\) of minimal degree, where \(\mathcal{V}\) is a  \(G\)-module and a $\splitk[X,Y]$-module. %with its divisor of zeros equal to an exceptional (or degenerate) $G$-orbit.
\end{Definition}
The computations of polynomial ground forms for the $\bbbt$$\bbbo$$\bbbi$ groups can be found, for instance, in  \cite{Dolgachev10}, \cite[II.6]{MR845275} and \cite{MR1315530}.
\begin{Definition}[Transvectant]
\label{def:trans}
Let \(\mathcal{V}\) and \(\mathcal{W}\) be  \(G\)-modules and $\splitk[X,Y]$-modules.
Let \(\phi\in \mathcal{V}_G\) and 
\(\phi_{k,l}=\frac{\partial^{k+l}\phi}{\partial X^k\partial Y^l}\);
we define the \(k\)th--\textbf{transvectant} of \(\phi\) with \(\psi \in \mathcal{W}_G\) 
\[\mf{F}=(\phi\,,\psi )^k =\sum_{i=0}^k (-1)^i\binom{k}{i}  \phi_{i,k-i}\otimes
\psi_{k-i,i}\in (\mathcal{V}\otimes \mathcal{W})_G\,.\]
\end{Definition}
\begin{Lemma}
Let \(\phi\in \mathcal{V}_G\) and \(\psi \in \mathcal{W}_G\); the transvectant transforms as 
\[
g(\phi\,,\psi )^k =(g\phi\,,g\psi )^k \,, \quad g\in G.
\]
This implies that \((\phi\,,\psi )^k\in (\mathcal{V}\otimes \mathcal{W})_G\), and if $\phi$ and $\psi$ are invariant, so is $(\phi\,,\psi )^k$.
\end{Lemma}
\begin{Corollary}
Let \(\mf{A}\in\mathcal{V}\) be a ground form and \(\bar{\alpha}\) an invariant polynomial.
Then \((\bar{\alpha},\mf{A})^l \in \mathcal{V}^G\). 
\end{Corollary}
\begin{Corollary}
Let \(\phi\in(\mathcal{V}\otimes V)^G\) and \(\psi\in(V^\ast \otimes \splitk[X,Y])^G\).
Let \(\mf{A}=\mathrm{Trace\ } \phi\otimes \psi\in \mathcal{V}^G\) be an invariant form,
Then \((\bar{\alpha},\mf{A})^l=\mathrm{Trace\ } \phi\otimes (\bar{\alpha},\psi)^l\in \mathcal{V}^G\), with \(\bar{\alpha}\) a polynomial invariant.
\end{Corollary}
This justifies the way we compute a sequence of invariants from a ground form using the Molien function of the irreducible representation \(V\) (see Section \ref{sec:comp_inv_mat}).

%%%%%%%%%%%%%%%%%%%%% EXAMPLES %%%%%%%%%%%%%%%%%%%%%%

\begin{Example}
The polynomial ground forms for \(\bbbt,\bbbo\) and \(\bbbi\), in the bases given 
by (\ref{eq:Tnrep}), (\ref{eq:Onrep}) and (\ref{eq:Ynrep}) respectively, are:
\begin{eqnarray*}
\alphat&=&
\omega_3 (X - 1/3(1-\omega_3)Y)
( X - (1+ \omega_3)Y)
(X+ (1 + \omega_3)Y)
(X  -(1- \omega_3)Y)\,.
\\
\alphao&=&
( X - 1/2(1  + \omega_8 +\omega_8^2 -\omega_8^3)Y)
( X - 1/2(1  - \omega_8 - \omega_8^2 + \omega_8^3)Y)
( X -(1  - \omega_8 + \omega_8^2)Y)
\\ &\times & 
(X -(1  - \omega_8^2 - \omega_8^3)Y)
(X - (\omega_8 - \omega_8^2)Y)
(X - (\omega_8^2 + \omega_8^3)Y)\,,
\\
\alphai&=&
X
Y
(X + (1  + \omega_5^2 + \omega_5^3)Y)
(X + \omega_5^3 Y)
(X+ 1/5( 2 - \omega_5 + \omega_5^2 + 3\omega_5^3)Y)
\\ &\times & 
(X - \omega_5Y)
(X+( 1  + \omega_5 + \omega_5^2 + \omega_5^3)Y)
(X + (\omega_5^2 + \omega_5^3)Y)
\\ &\times & 
(X +(1 + \omega_5^3)Y)
(X - (\omega_5 + \omega_5^2)Y)
(X+Y )
(X - (\omega_5 - \omega_5^3)Y)\,.
\end{eqnarray*}
The subindex of \(\alpha_{i,j}\) is determined as follows: $i$ is the \(X,Y\)-degree and $j$ identifies the element in the group of one dimensional characters  describing how $\alpha_{i,j}$ transforms. For example, the one dimensional characters of $\bbbt$ constitute the group $\bbbz/3=\{0,1,2\}$ by identifying $\bbbt_{j+1}$ with $j\in \bbbz/3$. 
In $\alphai$ the second grading is trivial, so it is omitted (see also Examples \ref{ex:T}-\ref{ex:Y}).
\end{Example}
%%%%%%%%%%%%%%%%%%%%%%%%%%%%%%%%%%%%%%%%%%%%%%%%%%%%%%%%%%%%%%%%%%%
\begin{Example}[Classical Invariant Theory]
\label{ex:cit}
Let $V=W=\splitk[X,Y]_G$ and replace in the Definition \ref{def:trans} the tensor product by the ordinary product of polynomials. Then
 \(\mf{F}\in \splitk[X,Y]_G\). Let $\alpha$ be the lowest degree relative invariant, then it follows from the classical theory that if  \(G\)  is either \(\bbbt\), \(\bbbo\) or \(\bbbi\) the classical relative invariants \cite{MR0080930,MR1315530} are given by 
\[\alpha\,,\quad \beta=(\alpha,\alpha)^{2}\,,\quad\gamma=(\alpha,\beta)^{1}\,.\]
\begin{center}
\begin{table}[h!] 
\begin{center}
\begin{tabular}{|c|c|c|c|} \hline
\(G\)  & \(\deg_\alpha\) &  \(\deg_\beta=2\deg_\alpha-4 \) & \(\deg_\gamma=3\deg_\alpha-6\) \\ 
\hline 
\hline
\( \bbbt\) & 4 & 4 & 6\\ 
\hline 
\( \bbbo\) & 6 & 8 & 12\\ 
\hline 
\( \bbbi\) & 12 & 20 & 30\\ 
\hline
\end{tabular}
\end{center}
\caption{Degrees of the  classical relative invariants of \(\bbbt, \bbbo, \bbbi\).}
\label{tab:invariants}
\end{table}
\end{center}
If one denotes the degree of a form \(\alpha\) by \(\deg_\alpha\) it follows that (see Table \ref{tab:invariants})
\[\deg_\beta=2\,\deg_\alpha-4\,,\quad \deg_\gamma=3\,\deg_\alpha-6\,.\]
The degree of \(\beta\) is the number of faces  of the Platonic solid and determines its name.
We observe that \(\deg_\alpha-\deg_\gamma+\deg_\beta=2\), the Euler characteristic, and that \(\deg_\alpha+\deg_\beta+\deg_\gamma=|G|+2\).
\end{Example}
The next examples illustrate how the Molien series information is combined with the concept of transvectant to construct a basis for the relative invariants. We write $\splitk[V]=\splitk[X,Y]$ when $\{X,Y\}$ is a basis for the dual of a representation $V$. %$V^\ast$

%%%%%%%%%%%%%%%%%%%%%%%%%%%%%%%%%%%%%%%%%%%%%%%%%%%%%%%%%%%%%%%%%%%
\begin{Example}[Tetrahedral group $\bbbt$]
\label{ex:T}
The ring generated by the relative invariants is determined as follows.
From GAP we obtain the Molien function
\[
M(\splitk[\bbbt_4^\flat]_{\bbbt^\flat},t)=\frac{1+2t^4+2t^8+t^{12}}{(1-t^6)(1-t^8)}
=\frac{1-t^{12}}{(1-t^4)^2(1-t^6)}
=\frac{1+t^{6}}{(1-t^4)^2}\,.
\]
To find the ground form
\(\alphat \) we look in \({\bbbt_2}\otimes \splitk_4[\bbbt_4^\flat]\).
Then \(\betat=(\alphat,\alphat)^2\in \splitk_4[\bbbt_4^\flat]^{\bbbt_3}\)
and \(\gammat=(\alphat,\betat)^1\in \splitk_6[\bbbt_4^\flat]^{\bbbt^\flat}\), in analogy with classical invariant theory.
This follows from Table \ref{tab:invariants}. 
Thus one finds that
\[
\splitk[\bbbt_4^\flat]_{\bbbt^\flat}=\splitk[\alphat,\betat](1\oplus\gammat)
\]
where
\[
\alphat=   Y^4 - 8/3 XY^3 + 2X^2Y^2 - 4/3X^3Y - 4/3\omega_3 XY^3 +2\omega_3 X^2Y^2 - 8/3\omega_3 X^3Y+\omega_3 X^4,
\]
\[
\betat= - 128XY^3 +128X^3Y -256 \omega_3 XY^3 + 384\omega_3 X^2 Y^2 - 128\omega_3 X^3Y
\]
and
\[
\begin{array}{rl}
\gammat=&
-512Y^6 + 2560X^2Y^4 - 5120 X^4Y^2 + 3072X^5Y - 512 X^6
 - 1024\omega_3 Y^6+
\\\\ &
+3072 \omega_3 X Y^5 - 2560\omega_3 X^2Y^4 - 2560\omega_3X^4Y^2 + 3072\omega_3 X^5 Y - 1024\omega_3 X^6 ,\end{array} 
\]
in the basis given by (\ref{eq:Tnrep}).
One expects from the Molien function a relation at degree \(12\) of the form
\[
 \alphat^3+C_\alpha^\beta\betat^3+C_\alpha^\gamma \gammat^2=0,\quad C_\alpha^\beta,\,C_\alpha^\gamma\in \splitk^\ast
\]
and one finds \(C_\alpha^\beta= 1/884736 \) and \(C_\alpha^\gamma= 1/786432\).
The Molien function of the invariants is given by
\[
M(\splitk[\bbbt_4^\flat]^{\bbbt^\flat},t)=\frac{1+t^{12}}{(1-t^6)(1-t^8)}\,.
\]
Thus the invariants corresponding to these terms are $\gammat\equiv\balphat$ for $t^6$, $\alphat\betat\equiv\bbetat$ for $t^{8}$ and $\alphat^3\equiv\bgammat$ for $t^{12}$ (or equivalently $\betat^3$). Hence. the ring of invariants can be written as
\[
%\splitk[\bbbt_4^\flat]^{\bbbt^\flat}=\splitk[\balphat,\bbetat]\oplus\splitk[\balphat,\bbetat]\bgammat\,.
\splitk[\bbbt_4^\flat]^{\bbbt^\flat}=\splitk[\balphat,\bbetat](1\oplus\bgammat)\,.
\]
\end{Example}
%%%%%%%%%%%%%%%%%%%%%%%%%%%%%%%%%%%%%%%%%%%%%%%%%%%%%%%%%%%%%%%%%%%
\begin{Example}[Octahedral group $\bbbo$]
\label{ex:O}
Similarly, the ring generated by the $\bbbo$-relative invariants is determined as follows.
From GAP we obtain the Molien function
\[
M(\splitk[\bbbo_4^\flat]_{\bbbo^\flat},t)=\frac{1+t^6+t^{12}+t^{18}}{(1-t^8)(1-t^{12})} 
=\frac{1+t^{12}}{(1-t^6)(1-t^8)}
\]
and the individual generating function for \(\bbbo_2\) is
\[
M(\splitk[\bbbo_4^\flat]^{\bbbo_2},t)=\frac{t^6+t^{12}}{(1-t^8)(1-t^{12})}
\]
and for \(\bbbo_1\) is
\[
M(\splitk[\bbbo_4^\flat]^{\bbbo^\flat},t)=\frac{1+t^{18}}{(1-t^8)(1-t^{12})}.
\]
To find the basic covariant \(\alphao \) we look in \(\splitk_6[\bbbo_4^\flat]^{\bbbo_2}\).
Then \(\betao=(\alphao,\alphao)^2\in \splitk_8[\bbbo_4^\flat]^{\bbbo^\flat}\)
and \(\gammao=(\alphao,\betao)^1\in \splitk_{12}[\bbbo_4^\flat]^{\bbbo_2}\).
Thus one finds that
\[
\splitk[\bbbo_4^\flat]_{\bbbo^\flat}=\splitk[\alphao,\betao](1\oplus\gammao)\,.
\]
We identify the terms in the Molien function for \(\bbbo_1\) as:
the \(t^8\) is \(\balphao=\betao\), the \(t^{12}\)-term is \(\bbetao=\alphao^2\) and the \(t^{18}\)-term
is \(\bgammao=\alphao\gammao\).
We identify the terms in numerator of the \(\bbbo_2\)-Molien function as follows.
The \(t^6\) term is \(\alphao\), the \(t^{12}\) term is \(\gammao\).

One can check that the relative invariants satisfy a relation of the form 
\[
\alphao^4+C_\alpha^\beta \betao^3+C_\alpha^\gamma \gammao^2=0\,.
\]
It follows that the invariants have the following relation
\[
C_\alpha^\beta  \balphao^3\bbetao+\bbetao^3+C_\alpha^\gamma \bgammao^2=0
\]
and that the ring of invariants can be written as
\[
%\splitk[\bbbo_4^\flat]^{\bbbo^\flat}=\splitk[\balphao,\bbetao]\oplus\splitk[\balphao,\bbetao]\bgammao\,.
\splitk[\bbbo_4^\flat]^{\bbbo^\flat}=\splitk[\balphao,\bbetao](1\oplus\bgammao)\,.
\]
\end{Example}

%%%%%%%%%%
\begin{Example}[Icosahedral group $\bbbi$]
 \label{ex:Y}
The Molien function of the invariants is
\[
M(\splitk[\bbbi_2^\flat]^{\bbbi^\flat},t)=\frac{1+t^{30}}{(1-t^{12})(1-t^{20})}.
\]
The invariants are \(\alphai\), \(\betai=(\alphai,\alphai)^2\) and
\(\gammai=(\alphai,\betai)^1\), and they satisfy the following relation
\[
\alphai^5+C_\alpha^\beta \betai^3+C_\alpha^\gamma \gammai^2=0\,.
\]
The ring of invariants can be written as
\[
\splitk[\bbbi_2^\flat]^{\bbbi^\flat}=\splitk[\alphai,\betai](1\oplus\gammai)\,.
\]

\end{Example}

%%%%%%%%%%%%%%%%%%%%%%%%%%%%%%%%%%%%%%%%%%%%%%%%%%%%%%%%%%%%%%%%%%%%%%%%%%%%%%%%%%%%%%%%%%%%%%%%%%%%%%%%%%%%%%%%%%%%%%

\subsection{$\bbbt\bbbo\bbbi$-Invariant matrices}
\label{sec:inv_mat_a}
Our goal is to determine the structure of the Lie algebra of invariant matrices.
Once the ground forms are computed, the other degrees can be realised by taking appropriate transvectants with the relative invariants.
The choice of transvectants is completely independent of the dimension we are working in, thus the construction is completely uniform.\\
We observe in first place that it is possible to predict that the number of generators of $(V\otimes \splitk[X,Y])^{G^\flat}$ is twice the dimension of $V$. This follows from the following Lemma, a modification of a method by Stanley \cite{stanley1979invariants}.
%\cite[4.2 \& 4.3]{stanley1979invariants}.
\begin{Lemma}
Let $G^\flat$ be a finite subgroup of $\textrm{SL}(2,\CC)$ and let  $V$ be one of its irreducible representation with  character $\chi$. The space of invariants $(V\otimes \splitk[X,Y])^{G^\flat}$ is a Cohen-Macaulay module of Krull dimension 2. Say  
\[(V\otimes \splitk[X,Y])^{G^\flat}=\bigoplus_{i=1}^{k_\chi}\splitk[\bar{\alpha},\bar{\beta}]\rho_i\]
and set $e_i=\deg \rho_i$.
Then 
\begin{eqnarray}
k_\chi |G^\flat| &=&\deg_{\bar{\alpha}} \deg_{\bar{\beta}} \chi(1)\\
\frac{2}{k_\chi}\sum_{i=1}^{k_\chi}e_i&=&\deg_{\bar{\alpha}} +\deg_{\bar{\beta}} -2
\end{eqnarray}
\end{Lemma}
\begin{proof}
The two equations follow from the first two coefficients, $A$ and $B$,
of the Laurent expansion around $t=1$ of the Molien series
\[M((V\otimes \splitk[X,Y])^{G^\flat},t)=\frac{A}{(1-t)^2}+\frac{B}{1-t}+\mc{O}(1).\]
We have two ways to express this series. Namely by Molien's theorem
and by the expression of $(V\otimes \splitk[X,Y])^{G^\flat}$ as a
Cohen-Macaulay module.

First Molien's theorem:
$P(V\otimes \splitk[X,Y])^{G^\flat},t)=\frac{1}{|G^\flat|}\sum_{g\in
G^\flat}\frac{\overline{\chi(g)}}{\det (1-t\sigma(g))}$.
Considering $\sigma(g)$ to be diagonal we see that the only
contribution to the term of order $(1-t)^{-2}$ in the Laurent
expansion
comes from the identity element $g=1$, so
$A=\frac{\chi(1)}{|G^\flat|}.$ The terms
$\frac{\overline{\chi(g)}}{\det (1-t\sigma(g))}$ that
contribute to the coefficient of $(1-t)^{-1}$ in the Laurent expansion
come from elements $\sigma(g)$ that have precisely one eigenvalue
equal to \(1\).
However, since $\det \sigma(g)=1$ there are no such elements: $B=0$.

On the other hand we notice that
\begin{eqnarray*}
P(\bigoplus_{i=1}^{k_\chi}\splitk[\bar{\alpha},\bar{\beta}]\rho_i,t)=\frac{\sum_{i=1}^{k_\chi}t^{e_i}}{(1-t^{\deg_{\bar{\alpha}}})(1-t^{\deg_{\bar{\beta}}})}
\end{eqnarray*}
and the first two coefficients of the Laurent expansion around $t=1$ are $A=\frac{k_\chi}{\deg_{\bar{\alpha}}\deg_{\bar{\beta}}}$ and $B=\frac{k_\chi}{2 \deg_{\bar{\alpha}}\deg_{\bar{\beta}}}(\deg_{\bar{\alpha}}-1)+\frac{k_\chi}{2\deg_{\bar{\alpha}}\deg_{\bar{\beta}}}(\deg_{\bar{\beta}}-1)-\frac{1}{\deg_{\bar{\alpha}}\deg_{\bar{\beta}}}\sum_{i=1}^{k_\chi}e_i$. The result follows.
\end{proof}

In Section \ref{sec:moi} we then repeat the procedure of Section \ref{sec:comp_inv_mat}, with a slight variation, to produce a basis for relative invariant vectors.
%Our goal is to determine the structure of the Lie algebra of invariant matrices
%with trace zero.
%We know that \(\mf{g}(V)\) splits into irreducible representation of \(G^\flat\).

In the following sections we compute a basis for \(|G|\)-homogeneous $G$-invariant matrices; this is a minimal generating set for the module of $G$-invariant matrices (over the primary invariants $\alpha^{d_G}$ and $\beta^3$) whose homogeneous elements have degree divisible $|G|$. This will be enough to construct a minimal generating set for the Automorphic Lie Algebra (see \cite{KLS-DN, knibbeler2015isotypic}).

\subsubsection{Tetrahedral group invariant matrices}
\label{sec:it}
From Table \ref{tab:decompTsl}
 it follows that \(\mf{g}(V)\) splits into a direct sum of \(\bbbt_i, i=2,3,7\).
We then consider \((\bbbt_i\otimes \splitk_{12}[\bbbt_4^\flat])^{\bbbt^\flat}\), as it is sufficient to consider entries of degree equal to the order of the group $|\bbbt|$ (see \cite{KLS-DN, knibbeler2015isotypic}).

The groundforms and transvectants are listed in Table \ref{tab:MolienT}. Notice that the degrees in column Molien and Multiplier add up to the order of the group.

\renewcommand{\arraystretch}{1.3}
\begin{center}
\begin{table}[H]
\begin{center}
\begin{tabular}{|c|c|c|c|c|} \hline
irrep &   Molien & ground form &invariant matrix & multiplier  \\
\hline 
\hline
$\bbbt_1$&$1$&$\mf{A}_1^0$&$\mf{M}_1^0=\mf{A}_1^0$ &$\balphat^2$\\
\hline
$\bbbt_2$&$t^4$&$\mf{A}_2^4$&$\mf{M}_2^4=\mf{A}_2^4$ &$\bbetat$\\
\hline
$\bbbt_3$&$t^4$&$\mf{A}_3^4$&$\mf{M}_3^4=\mf{A}_3^4$&$\bbetat$\\
\hline
$\bbbt_7$&$t^4$&$\mf{A}_7^2$&$\mf{M}_7^4=(\balphat,\mf{A}_7^2)^2$&$\bbetat$\\
&$t^6$&&$\mf{M}_7^{6}=(\balphat,\mf{A}_7^{2})^1$&$\balphat$\\
&$t^6$&&$\mf{N}_7^{6}=(\bbetat,\mf{A}_7^{2})^2$&$\balphat$\\
\hline
\end{tabular}
\end{center}
%\caption{Table of Molien functions for the invariant matrices as a \(\splitk[\balphat,\bbetat]\)-module}
%\caption{Generators of $\bbbt^\flat$-invariant matrices of degree $|\bbbt|$} 
\caption{Generators of $\bbbt$-invariant matrices of degree $|\bbbt|$.} 
\label{tab:MolienT}
\end{table}
\end{center}
%\begin{center}
%\begin{table}[H]
%\begin{center}
%\begin{tabular}{|c|c|c|c|c|} \hline
%irrep &   Molien & multiplier & ground form &invariant matrix \\
%\hline 
%\hline
%$\bbbt_1$&$1$&$\balphat^2$&$\mf{A}_1^0$&$\mf{M}_1^0=\mf{A}_1^0$ \\
%\hline
%$\bbbt_2$&$t^4$&$\bbetat$&$\mf{A}_2^4$&$\mf{M}_2^4=\mf{A}_2^4$ \\
%\hline
%$\bbbt_3$&$t^4$&$\bbetat$&$\mf{A}_3^4$&$\mf{M}_3^4=\mf{A}_3^4$\\
%\hline
%$\bbbt_7$&$t^4$&$\bbetat$&$\mf{A}_7^2$&$\mf{M}_7^4=(\balphat,\mf{A}_7^2)^2$\\
%&$t^6$&$\balphat$&&$\mf{M}_7^{6}=(\balphat,\mf{A}_7^{2})^1$\\
%&$t^6$&$\balphat$&&$\mf{N}_7^{6}=(\bbetat,\mf{A}_7^{2})^2$\\
%\hline
%\end{tabular}
%\end{center}
%%\caption{Table of Molien functions for the invariant matrices as a \(\splitk[\balphat,\bbetat]\)-module}
%\caption{Generators of $\bbbt^\flat$-invariant matrices} 
%\label{tab:MolienT}
%\end{table}
%\end{center}

Table \ref{tab:MolienT} is constructed by considering first the decomposition in Table \ref{tab:decompTsl}; one observes that the only representations playing a role are 
$\bbbt_2$, $\bbbt_3$ and $\bbbt_7$, so they are listed in the first column of Table \ref{tab:MolienT}.
The trivial representation $\bbbt_1$ is added for future reference. 
Next one considers the numerators of their corresponding Molien functions (see Table \ref{tab:irrMolienT}): the lowest order terms ($t^4$, $t^4$ and $t^2$), computed using the technique of Section \ref{sec:fourier} are the ground forms $\mf{A}_2^4$,  $\mf{A}_3^4$ and $\mf{A}_7^2$ in the third column, where the upper index denotes the degree in $X$ and $Y$ and the lower index refers to the irreducible representation (see the first column). The fourth column contains the invariant matrices; the last three entries correspond to $t^4$ and $2t^6$ in the $\bbbt_7$-row are obtained by taking the first transvectant with the primary invariants $\bbetat$, $\balphat$. 
It is worth noticing that not all terms in the numerator of the Molien function are present. 
This is due to the fact that not all invariant matrices can be made $|G|$-homogeneous: for instance, looking at the Table \ref{tab:irrMolienT}  for $\bbbt_2$, we observe that the $t^{8}$ term is missing, indeed in this case one would need to solve 
%requires that $\deg (\balphat^n\,\bbetat^m\,(\mf{A}_3^4)^2)=|\bbbt|$
%$\deg (\balphat^n\,(\bbetat)^m\,\mf{A}_2^4)=|\bbbt|$ 
%which leads to 
the linear diophantine equation $6n+8m+8=|\bbbt|=12$ which has no solutions for $n$ and $m$ non-negative integer. 
%The third column of the Table \ref{tab:MolienT} illustrates that one can consider the invariant matrices as a \(\splitk[\balphat,\bbetat]\)-module,
%where $\balphat$ are $\bbetat$ are  the primary invariants. 
The last column of the Table \ref{tab:MolienT} illustrates that one can solve the diophantine equation for the terms in the second column, hence a basis for $|\bbbt|$-homogeneous $\bbbt^\flat$-invariant matrices is given by the products of the elements in the last two columns.
%Indeed, one can read from the third column that there is always an invariant multiplying the ground form so that the degree of the product equals the order of the group $\bbbt$. 

%The fifth column contains the invariant matrices; the last three entries correspond to $t^4$ and $2t^6$ in the $\bbbt_7$-row are obtained by taking the first transvectant with the primary invariants $\bbetat$, $\balphat$. 

\begin{Example}\label{ex:psiGround}
From Table \ref{tab:decompTsl} one has \(\rlie{s}{2}{\bbbt_5^\flat}\cong \bbbt_7\). To find a concretisation of $\mf{A}_7^2$ we consider an embedding $\vartheta^{\rlie{s}{2}{\bbbt_5^\flat}}$ of $\bbbt_7$ into $\rlie{s}{2}{\bbbt_5^\flat}$:
%Let 
%In the case of
%\(\rlie{s}{2}{\bbbt_5^\flat}\)
%one has
%\iffalse
%\begin{quote}
%\#$DEG1x1r7=2;\\
%L A1x1r7x1=v(1)*V(1)*Y^2 - 2*v(1)*V(1)*X*Y - v(1)*V(1)*X^2 - 2*v(1)*V(1)*\omega_3*X*Y\\
% + v(1)*V(1)*\omega_3*X^2 - v(1)*V(2)*Y^2 + v(1)*V(2)*X^2 + 2*v(1)*V(2)*\omega_3*X*Y \\
%+ v(1)*V(2)*\omega_3*X^2 + v(2)*V(1)*Y^2 - 4*v(2)*V(1)*X*Y + 3*v(2)*V(1) *X^2 \\
%- 2*v(2)*V(1)*\omega_3*X*Y + 3*v(2)*V(1)*\omega_3*X^2 - v(2)*V(2)*Y^2 + 2*v(2)*V(2)*X*Y 
%+ v(2)*V(2)*X^2 + 2*v(2)*V(2)*\omega_3*X*Y - v(2)*V(2)*\omega_3*X^2;\\
%.sort\\
%\#$dim1x1r7=1;\\
%\end{quote}
%\fi
\[
%\psi_5(\mf{A}_7^2)= \begin{pmatrix}
\vartheta^{\rlie{s}{2}{\bbbt_5^\flat}}(\mf{A}_7^2)= \begin{pmatrix}
  Y^2-2(1+\omega_3) XY+(\omega_3-1) X^2 & 
-Y^2 +2\omega_3 XY +(\omega_3+1) X^2 \\ 
Y^2  -2 (2+\omega_3) XY+3(\omega_3+1) X^2&
 - Y^2 +2(1+\omega_3) XY+(1-\omega_3) X^2 \end{pmatrix}\,.
\]
In the case of
\(\rlie{s}{3}{\bbbt_7}\cong\bbbt_2\oplus\bbbt_3\oplus 2 \bbbt_7\) one has two concretisations of the ground form $\mf{A}_7^2$, namely $\vartheta_{1}^{\rlie{s}{3}{\bbbt_7}}(\mf{A}_7^2)$ and $\vartheta_{2}^{\rlie{s}{3}{\bbbt_7}}(\mf{A}_7^2)$, since the multiplicity of \(\bbbt_7\) in \(\rlie{s}{3}{\bbbt_7}\)  is two.
%,
%where the map
%\(\vartheta_{7,i}^{\rlie{s}{3}{\bbbt_7}}\) embeds irreducible summands into $\mf{sl}(V_i)$ 
%%\(\bbbt_i\rightarrow \bbbt_i\otimes\bbbt_i^\ast-\bbbt_1\),
%that is, a section to the decomposition into irreducibles of the tensor product.
\end{Example}
\begin{Example}
We compute a set of generators for $\mf{sl}_3(\bbbt_7)$, linearly independent over the ring $\splitk[\balphat,\bbetat]$ of primary invariants. We know that $\mf{sl}_3(\bbbt_7)\cong\bbbt_2\oplus\bbbt_3\oplus 2 \bbbt_7$. Therefore we have ground forms $\mf{A}_2^4$, $\mf{A}_3^4$ and $\mf{A}_{7}^2$.
%$\mf{A}_2^4$, $\mf{A}_3^4$, $\mf{A}_{7,1}^2$ and $\mf{A}_{7,2}^2$. 
Thus we compute the generators 
$\vartheta^{\rlie{s}{3}{\bbbt_7}}(\mf{M}_2^4)$, 
$\vartheta^{\rlie{s}{3}{\bbbt_7}}(\mf{M}_3^4)$,
$\vartheta_{1}^{\rlie{s}{3}{\bbbt_7}}(\mf{M}_{7}^4)$,
$\vartheta_{1}^{\rlie{s}{3}{\bbbt_7}}(\mf{M}_{7}^6)$, 
$\vartheta_{1}^{\rlie{s}{3}{\bbbt_7}}(\mf{N}_{7}^6)$,
$\vartheta_{2}^{\rlie{s}{3}{\bbbt_7}}( \mf{M}_{7}^4)$,
$\vartheta_{2}^{\rlie{s}{3}{\bbbt_7}}(\mf{M}_{7}^6)$,
$\vartheta_{2}^{\rlie{s}{3}{\bbbt_7}}( \mf{N}_{7}^6)$.
%$$
%\begin{array}{ll}
%\{
%\vartheta_{2}^{\rlie{s}{3}{\bbbt_7}}(\mf{A}_2^4), \;
%\vartheta_{3}^{\rlie{s}{3}{\bbbt_7}}(\mf{A}_3^4),\; 
%\vartheta_{7,1}^{\rlie{s}{3}{\bbbt_7}}(\mf{M}_{7}^4),\;
%\vartheta_{7,1}^{\rlie{s}{3}{\bbbt_7}}(\mf{M}_{7}^6),\; 
%\vartheta_{7,1}^{\rlie{s}{3}{\bbbt_7}}(\mf{N}_{7}^6),\;\\
%\vartheta_{7,2}^{\rlie{s}{3}{\bbbt_7}}( \mf{M}_{7}^4),\;
%\vartheta_{7,2}^{\rlie{s}{3}{\bbbt_7}}(\mf{M}_{7}^6),\; 
%\vartheta_{7,2}^{\rlie{s}{3}{\bbbt_7}}( \mf{N}_{7}^6) 
%\}.
%\end{array}
%$$
%$$
%\{\mf{A}_2^4, \,\mf{A}_3^4,\, \mf{M}_{7,1}^4,\, \mf{M}_{7,1}^6,\, \mf{N}_{7,1}^6,\, \mf{M}_{7,2}^4,\, \mf{M}_{7,2}^6,\, \mf{N}_{7,2}^6 \}.
%$$
Once we have tested their independence, we know from the Molien function that they span the space $(\mf{sl}(\bbbt_7)\otimes\splitk[\bbbt_4^\flat])^{\bbbt^\flat}$.
\end{Example}
%If $\bbbt_i^\flat$ is not a representation of $\bbbt$, there are no invariants in $\bbbt_i^\flat\otimes \splitk[X,Y]$ of degree $|\bbbt|$. In this case
%one can try as an alternative the lowest degree for which the dimension is the same as the dimension of the irreducible representation. 
%This is listed in Table \ref{tab:vecbasest}-\ref{tab:vecbasesi}. 

%%%%%%%%%%%%%%%%%%%%%%%%%%%%%%%%%%%%%%%%%%%%%%%%%%%%%%%%%%%%%%%%%%%%%%%%%%

\subsubsection{Octahedral group invariant matrices}
\label{sec:io}

Table \ref{tab:MolienO} is computed in the same spirit as in the previous section; also in this case, not all terms in the numerator of the Molien function (see Table \ref{tab:irrMolienO}) correspond to invariant matrices which can be made zero homogeneous, hence they are not listed below.
\begin{center}
\begin{table}[H]
\begin{center}
\begin{tabular}{|c|c|c|c|c|} 
\hline 
irrep &   Molien  & ground form &invariant matrix & multiplier\\
\hline
\hline
$\bbbo_1$&$1$&$\mf{A}_1^0$&$\mf{M}_1^{0}=\mf{A}_1^0$& $\bbetao^2$ \\
\hline
$\bbbo_2$&$t^{12}$&$\mf{A}_2^6$&$\mf{M}_2^{12}=(\balphao,\mf{A}_2^6)^1$ & $\bbetao$\\
\hline
$\bbbo_3$&$t^4$&$\mf{A}_3^4$&$\mf{M}_3^4=\mf{A}_3^4$& $\balphao\bbetao $\\
&$t^8$&&$\mf{M}_3^{8}=(\balphao,\mf{A}_3^{4})^2$& $\balphao^2$\\
\hline
$\bbbo_6$ & $t^4$&$\mf{A}_6^4$&$\mf{M}_6^4=\mf{A}_6^4$ & $\balphao\bbetao $\\
&$ t^8$&&$\mf{M}_6^{8}=(\balphao,\mf{A}_6^{4})^2$& $\balphao^2$ \\
&$ t^{12}$ &&$\mf{M}_6^{12}=(\balphao,\mf{A}_6^{8})^2$& $\bbetao$ \\
\hline
$\bbbo_7$& $ t^8$&$\mf{A}_7^2$&$\mf{M}_7^{8}=(\balphao,\mf{A}_7^{2})^1$& $\balphao^2$ \\
&$ t^{12}$&&$\mf{M}_7^{12}=(\balphao,\mf{M}_7^{8})^2$ & $\bbetao $\\
&$ t^{16}$ &&$\mf{M}_7^{16}=(\balphao,\mf{M}_7^{12})^2$ & $\balphao$\\
\hline
\end{tabular}
\end{center}
%\caption{Table of Molien functions for the invariant matrices as \(\splitk[\balphao,\bbetao]\)-module}
\caption{Generators of $\bbbo$-invariant matrices of degree $|\bbbo|$.} 
\label{tab:MolienO}
\end{table}
\end{center}

%%%%%%%%%%%%%%%%%%%%%%%%%%%%%%%%%%%%%%%%%%%%%%%%%%%%%%%%%%%%%%%%%%%%%%%%%%%

\subsubsection{Icosahedral group invariant matrices}
\label{sec:ii}
The invariant matrices for \(\bbbi^\flat\) are presented in the Table \ref{tab:MolienY}; as before, not all terms in the numerator of the Molien function (see Table \ref{tab:irrMolienY}) correspond to invariant matrices which can be made zero homogeneous, hence they are not listed below.
\begin{center}
\begin{table}[H]
\begin{center}
\begin{tabular}{|c|c|c|c|c|} 
\hline
irrep &   Molien  & ground form & invariant matrix & multiplier \\
\hline
\hline 
$\bbbi_1$  &$1$ &$\mf{A}_1^0$& $ \mf{M}_1^{0}=\mf{A}_1^0$ & $\alphai^5 $ \\
\hline
$\bbbi_4$  &$t^{16}$&$\mf{A}_4^6$& $ \mf{M}_4^{16}=(\alphai,\mf{A}_4^6)^1$& $\alphai^2\betai $  \\
  &  $t^{20}$ && $\mf{M}_4^{20}=(\alphai,\mf{M}_4^{16})^4$ & $\betai^2 $ \\
  &  $t^{24}$ && $ \mf{M}_4^{24}=(\alphai,\mf{M}_4^{20})^4$ & $\alphai^3$\\
\hline
$\bbbi_5$ &$t^{12}$ &$\mf{A}_5^2$&$\mf{M}_5^{12}=(\alphai,\mf{A}_5^2)^1$ & $\alphai^4 $\\
 &$ t^{20}$&&$\mf{M}_5^{20}=(\alphai,\mf{M}_5^{12})^2$& $\betai^2 $\\
 &$t^{28}$&&$\mf{M}_5^{28}=(\alphai,\mf{M}_5^{20})^2$& $\alphai\betai $\\
\hline
$\bbbi_6$  & $t^8$&$\mf{A}_6^6$&$\mf{M}_6^{8}=(\alphai,\mf{A}_6^6)^5$& $\alphai\betai^2 $\\
           & $t^{12} $&&$\mf{M}_6^{12}=(\alphai,\mf{M}_6^{8})^4$& $\alphai^4$\\
           & $t^{16}$&&$\mf{M}_6^{16}=(\alphai,\mf{M}_6^{12})^4$& $\alphai^2\betai $\\
           & $t^{24}$&&$\mf{M}_6^{24}=(\alphai,\mf{M}_6^{16})^2$& $\alphai^3$\\
\hline
$\bbbi_8$   &  $t^4$&$\mf{A}_8^4$&$\mf{M}_8^4=\mf{A}_8^4$& $\alphai^3\betai $\\
            &  $t^8$&&$\mf{M}_8^{8}=(\alphai,\mf{A}_8^4)^4$& $\alphai \betai^2 $\\
            &  $t^{12} $&&$\mf{M}_8^{12}=(\alphai,\mf{M}_8^8)^4$& $\alphai^4$\\
            &  $t^{16}$&&$\mf{M}_8^{16}=(\alphai,\mf{M}_8^{12})^4$& $\alphai^2\betai $\\
            &  $t^{20}$&&$\mf{M}_8^{20}=(\betai,\mf{A}_8^4)^2$& $\betai^2$\\
\hline
\end{tabular}
\end{center}
%\caption{Table of Molien functions for the invariant matrices as \(\splitk[\alphai,\betai]\)-module}
\caption{Generators of $\bbbi$-invariant matrices of degree $|\bbbi|$.} 
\label{tab:MolienY}
\end{table}
\end{center}

%%%%%%%%%%%%%%%%%%%%%%%%%%%%%%%%%%%%%%%%%%%%%%%%%%%%%%%%%%%%%%%%%%%%%%%%%%%%%%%%%%%%%%%%%%%%%%%%%%%%%%%%%%%%%%%%%%%%%%%%%%%%%%%%%%%%%%%%%%%%%%%%%%%%%%%%%%%%%%%%%%%%%%%%%%%%%%%%%%%%%%%%%%%%%%%%%%%%%%%%%%%%%%%%%%%%%%%%%%%%%%

%We have now a minimal generating set for the \(|G|\)-homogeneous $G^\flat$-invariant matrices (over the primary invariants $\bar{\alpha}$ and $\bar{\beta}$); 
%We now have a basis for \(|G|\)-homogeneous $G$-invariant matrices; this is a minimal generating set for the module of $G$-invariant matrices (over the primary invariants $\alpha^{d_G}$ and $\beta^3$) whose homogeneous elements have degree divisible $|G|$. 
At this stage one could in principle fix any $G$-orbit (exceptional or generic), divide the matrices by the corresponding invariant form (the invariant form vanishing at those points) in order to obtain
zero-homogeneous matrices depending on \(\lambda=X/Y\). 
In this paper we only consider the case of exceptional orbits. This correspond to dividing the matrices by \(\alpha^{d_G}\), \(\beta^3\) or \(\gamma^2\), where $d_G=3,4$ and $5$ for $\bbbt$, $\bbbo$ and $\bbbi$, respectively. These then form a minimal generating set (over the 
invariant \(\AIJ[I][\beta][\alpha]\), \(\AIJ[I][\alpha][\beta]\), \(\AIJ[I][\alpha][\gamma]\), respectively -- see next Section \ref{sec:af}). 
We denote this minimal generating set by $\langle \hat{M}^1\,,\cdots\,,\hat{M}^{n^2-1}\rangle$; it generates the \(G\)--Automorphic Lie Algebra.

\begin{Definition}
By \(\alias{\mf{sl}(V)}{\zeta}{\alpha}{G}\)
we denote the \(G\)--Automorphic Lie Algebra based on \(\mf{g}=\rlie{s}{}{V}\)
with homogeneous coefficients having poles at the $G$-orbit $\Gamma_\zeta$, 
or, equivalently, at the zeros of \(\zeta=\alpha\), \(\beta\) or \(\gamma\).
\end{Definition}

\begin{Remark}[Towards Lax Pairs]
Defining a Lax operator \( L\) \(\in\alias{\mf{sl}(V)}{\zeta}{\alpha}{G}
\) gives us a \( G\)--invariant (automorphic) Lax operator and therefore a \( G\)--invariant (automorphic) integrable systems of equations (see \cite{lm_jpa04}).
\end{Remark}

%%%%%%%%%%%%%%%%%%%%%%%%%%%%%%%%%%%%%%%%%%%%%%%%%%%%%%%%%%%%%%%%%%%%
\subsection{Zero-homogeneous automorphic functions}
\label{sec:af}
For the \(\bbbt\bbbo\bbbi\)-groups, the basic relative invariants \(\alpha,\beta \) and \(\gamma\) have a relation of the form
\[
C_\zeta^\alpha \alpha^{d_G} +C_\zeta^\beta \beta^3 +C_\zeta^\gamma \gamma^2=0, \quad \zeta=\alpha,\beta,\gamma.
\]
Dividing this relation by \(\zeta^{\nu_\zeta}\), with \(\nu_\alpha=d_G\), \(\nu_\beta=3\), \(\nu_\gamma=2\), and fixing \(C_\zeta^\zeta=1\), we obtain
a linear relation between two zero-homogeneous invariants \(\AIJ[I][][\zeta]\) and \(\AIJ[J][][\zeta]\).
For instance, with \(\zeta=\alpha\), 
the relation is 
\[
1+\AIJ[I][\beta][\alpha]+\AIJ[J][\gamma][\alpha]=0.
\]
The explicit definition in this case is \(\AIJ[I][\beta][\alpha]=C_\alpha^\beta \frac{\beta^{3}}{\alpha^{d_G}}\) and
\(\AIJ[J][\gamma][\alpha]=C_\alpha^\gamma \frac{\gamma^{2}}{\alpha^{d_G}}\).
Or, with \(\zeta=\beta\), 
the relation is 
\[
\AIJ[I][\alpha][\beta]+1+\AIJ[J][\gamma][\beta]=0.
\]
The explicit definition in this case is \(\AIJ[I][\alpha][\beta]=C_\beta^\alpha \frac{\alpha^{d_G}}{\beta^3}\) and
\(\AIJ[J][\gamma][\beta]=C_\beta^\gamma \frac{\gamma^{2}}{\beta^3}\).

A relative invariant \(\zeta\) is identified with the orbit of a specific group element \(g_\zeta\) of order \(\nu_\zeta\), such that \(d_\zeta\nu_\zeta=|G|\).
For each representation $W$ of the group one defines \(\kappa_\zeta=\nicefrac{1}{2}\,\codim W^{\langle g_\zeta\rangle}\). 
In Table \ref{tab:codimg} (Section \ref{sec:Invariants}) the numbers \( \coa,\cob, \coc\) are given for different Lie algebras $W=\mf{g}(V)$.

We use \(\AIJ[J][][]\) for the invariant related to the relative invariant with the lowest \(\kappa\).
If there is equality, for instance if \(\coa=\cob\), then in \(\AIJ[I][\alpha][\gamma]\) and \(\AIJ[J][\beta][\gamma]\),
one can interchange the \(\AIJ[I][][]\) and the \(\AIJ[J][][]\).
The fully adorned \(\AIJ[J][\gamma][\beta]\) is overloaded with indices and one can replace it by \(\AIJ[J][][\beta]\),
or one could have simply called it \(\AIJ[I][\gamma][\beta]\). The reason for the use of the \(\AIJ[J][][]\) notation at all, is that we later on want to be able to make statements about the Chevalley normal form (see Section \ref{sec:Chev}) and their isomorphism.

\begin{Remark}
In the \(\mf{sl}(V)\) case, the relative invariant of the highest degree identifies a lowest \(\kappa\) (there could be more than one, see Table \ref{tab:codimg}). In other words, $\kappa_\zeta \leq \kappa_{\zeta'}$ if $\deg_\zeta \geq \deg_{\zeta'}$. 
\end{Remark}

%%%%%%%%%%%%%%%%%%%%%%%%%%%%%%%%%%%%%%%%%%%%%%%%%%%%%%%%%%%%%%%%%%%%%%%%%%%%%%%%%%%%%%%%%%%%%%%%%%%%%%%%%%%%%%%%%%%%%%%%%%%%%%%%%%%%%%%%%%%%%%%%%%%%%%%%%%%%%%%%%%%%%%%%%%%%%%%%%%%%%%%%%%%%%%%%%%%%%%%%%%%%%%%%%%%%%%%%%%%%%%
%%%%%%%%%%%%%%%%%%%%%                MATRICES IF INVARIANTS

\section{Matrices of invariants}
\label{sec:moi}
By constructing a basis of invariant vectors for each irreducible representation (see Tables \ref{tab:vecbasest}-\ref{tab:vecbasesi}), we prepare ourselves for the next step, the computation of the \emph{matrices of invariants}: we change from the standard basis of an irreducible representation to the basis of invariant vectors. The matrices in the new basis will now have their coefficients in the space of invariants.
There are two reasons to make this change of basis.\\
The first is computational: it is much easier to work with the matrices of invariants, e.g.~when computing the structure constants.
In the computation of the Chevalley normal form for the Lie algebra we need to find eigenvalues (see Section \ref{sec:Chev}) and this is easier in this new basis.
The second reason is that when the algebra is in Chevalley normal form, it will be natural to ask
whether the algebra is isomorphic to another case.
This \emph{isomorphism question} is difficult to settle, unless one has an explicit way to go from one case to the next. And this is exactly what the matrices of invariants provide.
When everything is in Chevalley normal form, the matrices of invariants have been reduced to elementary matrices with invariant coefficients. 
To analyse them
one can now use permutations and scalings with \(\mathbb{I}\) and \(\mathbb{J}\).
This limits the problem enough that one can finally answer the isomorphism question.

\begin{center}
\begin{table}[H]
\begin{center}
\begin{tabular}{|c|c|c|c|c|} \hline
irrep &   Molien & ground form & invariant  vector & multiplier \\ 
\hline 
\hline
$\bbbt_{2}$&$t^4$&$\mf{a}_{2}^4$&$\mf{v}_{2}^4=\mf{a}_{2}^4$&$1$\\
\hline
$\bbbt_{3}$&$t^4$&$\mf{a}_{3}^4$&$\mf{v}_{3}^4=\mf{a}_{3}^4$&$1$\\
\hline
$\bbbt_{4}^\flat$&$t$&$\mf{a}_{4}^1$&$\mf{v}_{4}^1=\mf{a}_{4}^1$&$\balphat$\\
&$t^7$&&$\mf{v}_{4}^7=(\bbetat,\mf{a}_{4}^1)^1$&$1$\\
\hline
$\bbbt_{5,6}^\flat$&$t^3$&$\mf{a}_{5,6}^3$&$\mf{v}_{5,6}^3=\mf{a}_{5,6}^3$&$\bbetat$\\
&$t^5$&&$\mf{v}_{5,6}^5=(\balphat,\mf{a}_{5,6}^3)^2$&$\balphat$\\
\hline
$\bbbt_7$&$t^2$&$\mf{a}_7^2$&$\mf{v}_7^2=\mf{a}_7^2$&$\bbetat$\\
&$t^4$&&$\mf{v}_7^4=(\balphat,\mf{a}_7^2)^2$&$\balphat$\\
&$t^{10}$&&$\mf{v}_7^{10}=(\bgammat,\mf{a}_7^2)^2$&$1$\\
\hline 
\end{tabular}
\end{center}
\caption{Bases of invariant  vectors for $\bbbt^\flat$.}
\label{tab:vecbasest}
\end{table}
\end{center}
\begin{Example} 
In the case of
\(\rlie{s}{2}{\bbbt_5^\flat}\)
one has
the invariant matrix
\[
%\psi_5(\mf{A}_7^2)= 
\vartheta^{\rlie{s}{2}{\bbbt_5^\flat}}(\mf{A}_7^2)=
\begin{pmatrix}
  Y^2-2(1+\omega_3) XY+(\omega_3-1) X^2 & 
-Y^2 +2\omega_3 XY +(\omega_3+1) X^2 \\ 
Y^2  -2 (2+\omega_3) XY+3(\omega_3+1) X^2&
 - Y^2 +2(1+\omega_3) XY+(1-\omega_3) X^2 
 \end{pmatrix}
%\begin{pmatrix} 2/3(- Y^2 - (2+\omega_3)XY + \omega_3 Y^2  )
%& -2/3(  (1+2\omega_3) Y^2 - (2 +\omega_3) XY + (1-\omega_3)X^2) \\ 
% - 2/3(2+\omega_3)Y^2 - 2\omega_3 XY &
%  2/3Y ((1-\omega_3) Y + (2+\omega_3)X)\end{pmatrix}
\]
(cf.~Example \ref{ex:psiGround}).
We consider the basis of invariant  vectors
\[
\vartheta^{\bbbt_5^\flat}(\mf{v}_5^3)=
%\begin{pmatrix}Y^3 - 4XY^2 + 5X^2Y - X^3 + \omega_3   X Y^2 + \omega_3 X^2 Y \\ Y^3 + 2 X Y^2 - 3 X^2 Y + X^3 +
%      \omega_3 X Y^2 - 3 \omega_3 X^2 Y + 2 \omega_3 X^3\end{pmatrix},
\begin{pmatrix}Y^3 +(\omega_3 - 4)XY^2 + (5+\omega_3) X^2Y - X^3   \\ 
Y^3 + (2+\omega_3) X Y^2 - 3(1+\omega_3) X^2 Y  +(1+ 2 \omega_3) X^3\end{pmatrix},
\]
\[
\vartheta^{\bbbt_5^\flat}(\mf{v}_5^5)=245760 \begin{pmatrix} XY^4 - 2 (1+\omega_3)X^2Y^3  + 2\omega_3 X^3Y^2  + X^4Y\\
XY^4 - 2(2+\omega_3)X^2Y^3+ 4(1+\omega_3 )X^3Y^2 - (1+2\omega_3 X^4Y)\end{pmatrix}\,.
\]
After making everything zero-homogeneous, the matrix of invariants of $\mf{M}_7^{4}=(\balphat,\mf{A}_7^{2})^2$ becomes
\[
 5898240 \begin{pmatrix}- 1 &   983040\,\AIJ[J][\gamma][\alpha] \\ 6/5898240 & 1 \end{pmatrix}.
\]
\end{Example}
\begin{center}
\begin{table}[H]
\begin{center}
\begin{tabular}{|c|c|c|c|c|} \hline
irrep &   Molien & ground form & invariant  vector & multiplier \\ 
\hline 
\hline
$\bbbo_{2}$&$t^6$&$\mf{a}_{2}^6$&$\mf{v}_{2}^6=\mf{a}_{2}^6$&$1$\\
\hline  
$\bbbo_{3}$&$t^4$&$\mf{a}_{3}^4$&$\mf{v}_{3}^4=\mf{a}_{3}^4$&$\bbetao$\\
&$t^8$&&$\mf{v}_{3}^8=(\balphao,\mf{a}_{3}^4)^2$&$\balphao$\\
\hline
$\bbbo_{4}^\flat$&$t$&$\mf{a}_{4}^1$&$\mf{v}_{4}^1=\mf{a}_{4}^1$&$\balphao^2$\\
&$t^{17}$&&$\mf{v}_{4}^{17}=(\bgammao,\mf{a}_{4}^1)^1$&$1$\\
\hline
$\bbbo_{5}^\flat$&$t^5$&$\mf{a}_{5}^5$&$\mf{v}_{5}^{5}=\mf{a}_{5}^5$&$ \balphao$\\
&$t^{13}$&$$ &$\mf{v}_{5}^{13}=(\bbetao,\mf{a}_{5}^5)^2$&$1$\\
\hline
$\bbbo_6$&$t^4$&$\mf{a}_6^4$&$\mf{v}_6^4=\mf{a}_6^4$&$\balphao^2$\\
&$t^8$&&$\mf{v}_6^8=(\balphao,\mf{a}_6^4)^2$&$\bbetao$\\
&$t^{12}$&&$\mf{v}_6^{12}=(\bbetao,\mf{a}_6^4)^2$&$\balphao$\\
\hline
$\bbbo_7$&$t^2$&$\mf{a}_7^2$&$\mf{v}_7^2=\mf{a}_7^2$&$\balphao^2$\\
&$t^6$&&$\mf{v}_7^6=(\balphao,\mf{a}_7^2)^2$&$\bbetao$\\
&$t^{10}$&$$&$\mf{v}_7^{10}=(\bbetao,\mf{a}_7^2)^2$&$\balphao$\\
\hline
$\bbbo_8^\flat$&$t^5$&$\mf{a}_8^3$&$\mf{v}_8^5=(\balphao,\mf{a}_8^3)^3$&$\balphao^2$\\
&$t^9$&&$\mf{v}_8^9=(\balphao,\mf{a}_8^3)^1$&$\bbetao$\\
&$t^9$&&$\mf{v}_8^9=(\bbetao,\mf{a}_8^3)^3$&$\bbetao$\\
&$t^{13}$&&$\mf{v}_8^{13}=(\bbetao,\mf{a}_8^3)^1$&$\balphao$\\
\hline 
\end{tabular}
\end{center}
\caption{Bases of invariant  vectors for $\bbbo^\flat$. }
\label{tab:eqveco}
\end{table}
\end{center}
 \begin{center}
\begin{table}[h!]
\begin{center}
\begin{tabular}{|c|c|c|c|c|} \hline
%Irrep &   Molien &Multiplier & Ground form &Transvectant  \\ 
irrep &   Molien  & ground form & invariant  vector & multiplier \\ 
\hline 
\hline
$\bbbi_2^\flat$&$t^{11}$&$\mf{a}_2^1$&$\mf{v}_2^{11}=(\alphai,\mf{a}_2^1)^1$&$\alphai^4$\\
&$t^{19}$&&$\mf{v}_2^{19}=(\betai,\mf{a}_2^1)^1$&$\betai^2$\\
\hline
$\bbbi_{3}^\flat$&$t^{13}$&$\mf{a}_{3}^7$&$\mf{v}_{3}^{13}=(\alphai,\mf{a}_{3}^7)^3$&$\alphai^4$\\
&$t^{17}$&&$\mf{v}_{3}^{17}=(\betai,\mf{a}_{3}^7)^1$&$\alphai^2\betai$\\
\hline
$\bbbi_{4}$&$t^{6}$&$\mf{a}_{4}^6$&$\mf{v}_4^6=\mf{a}_4^6$&$\betai^2$\\
&$t^{10}$&&$\mf{v}_{4}^{10}=(\alphai,\mf{a}_{4}^6)^4$&$\alphai^3$\\
&$t^{14}$&&$\mf{v}_{4}^{14}=(\alphai,\mf{a}_{4}^6)^2$&$\alphai\betai$\\
\hline
$\bbbi_{5}$&$t^2$&$\mf{a}_{5}^2$&$\mf{v}_5^2=\mf{a}_{5}^2$&$ \betai^2$\\
&$t^{10}$&&$\mf{v}_{5}^{10}=(\alphai,\mf{a}_{5}^2)^2$&$\alphai\betai$\\
&$t^{18}$&&$\mf{v}_{5}^{18}=(\betai,\mf{a}_{5}^2)^2$&$\alphai^2$\\
\hline
$\bbbi_6$&$t^8$&$\mf{a}_6^6$&$\mf{v}_6^8=(\alphai,\mf{a}_6^6)^5$&$\betai^2$\\
&$t^{12}$&&$\mf{v}_6^{12}=(\alphai,\mf{a}_6^6)^3$&$\alphai^3$\\
&$t^{16}$&&$\mf{v}_6^{16}=(\alphai,\mf{v}_6^6)^1$&$\alphai\betai$\\
&$t^{24}$&&$\mf{v}_6^{24}=(\betai,\mf{v}_6^6)^1$&$\alphai^2$\\
\hline
$\bbbi_7^\flat$&$t^3$&$\mf{a}_7^3$&$\mf{v}_7^3=\mf{a}_7^3$&$\alphai^4$\\
&$t^{11}$&&$\mf{v}_7^{11}=(\alphai,\mf{a}_7^3)^2$&$\betai^2$\\
&$t^{19}$&&$\mf{v}_7^{19}=(\betai,\mf{a}_7^3)^2$&$\alphai\betai$\\
&$t^{27}$&&$\mf{v}_7^{27}=(\gammai,\mf{a}_7^3)^3$&$\alphai^2$\\
\hline
$\bbbi_8$&$t^4$&$\mf{a}_8^4$&$\mf{v}_8^4=\mf{a}_8^4$&$\alphai^4$\\
&$t^8$&&$\mf{v}_8^8=(\alphai,\mf{a}_8^4)^4$&$\alphai\betai$\\
&$t^{12}$&&$\mf{v}_8^{12}=(\alphai,\mf{a}_8^4)^2$&$\betai^2$\\
&$t^{16}$&&$\mf{v}_8^{16}=(\betai,\mf{a}_8^4)^4$&$\alphai^3$\\
&$t^{20}$&&$\mf{v}_8^{20}=(\betai,\mf{a}_8^4)^2$&$\alphai\betai$\\
\hline
$\bbbi_9^\flat$&$t^7$&$\mf{a}_9^5$&$\mf{v}_9^7=(\alphai,\mf{a}_9^5)^5$&$\alphai^4$\\
&$t^{11}$&&$\mf{v}_9^{11}=(\alphai,\mf{a}_9^5)^3$&$\alphai^2\betai$\\
&$t^{15}$&&$\mf{v}_9^{15}=(\alphai,\mf{a}_9^5)^1$&$\betai^2$\\
&$t^{15}$&&$\mf{w}_9^{15}=(\betai,\mf{a}_9^5)^5$&$\betai^2$\\
&$t^{19}$&&$\mf{v}_9^{19}=(\betai,\mf{a}_9^5)^3$&$\alphai^3$\\
&$t^{23}$&&$\mf{v}_9^{23}=(\betai,\mf{a}_9^5)^1$&$\alphai\betai$\\
\hline 
\end{tabular}
\end{center}
\caption{Bases of invariant  vectors for $\bbbi^\flat$.}
\label{tab:vecbasesi}
\end{table}
\end{center}

%%%%%%%%%%%%%%%%%%%%%%%%%%%%%%%%%%%%%%%%

In sections \ref{sec:it}--\ref{sec:ii} we produced the invariant, zero homogeneous matrices $ \hat{M}^1\,,\cdots\,,\hat{M}^{n^2-1}$. 
We now produce a list of invariant, homogeneous vectors $\inv{v}_1$,...,$\inv{v}_n$, by taking an invariant vector $\mf{v}$ multiplied by the corresponding invariant multiplier (see Tables \ref{tab:vecbasest}-\ref{tab:vecbasesi}). 
The resulting set $\{\inv{v}_i\}$ generates  the invariant vectors over the polynomial invariants. 
If $\bbbt_i^\flat$ is not a representation of $\bbbt$, there are no invariants in $\bbbt_i^\flat\otimes \splitk[X,Y]$ of degree $|\bbbt|$. In this case
one can try as an alternative the lowest degree for which the dimension is the same as the dimension of the irreducible representation. 
This is listed in Table \ref{tab:vecbasest}-\ref{tab:vecbasesi}. 

When we compute \(\inv{M}^j \inv{v}_{i}\) the result is an invariant  vector, that is, a linear combination with invariant coefficients of degree \(|G|\) of the basic vectors $\inv{v}_1$,...,$\inv{v}_n$.
We denote the coefficient of \( \inv{v}_{k}\) by \( \psi(\inv{M}^j)_{k,i}\) and obtain the following representation of \(\inv{M}^j\):
\[
\inv{M}^j \inv{v}_{i}=\sum_{k=1}^n\, \psi(\inv{M}^j)_{k,i}\, \inv{v}_{k}\,.
\]
This defines the matrix $(\psi(\inv{M}^j))_{k,i}$ which is called the \emph{matrix of invariants} corresponding to $\inv{M}^j$, and we extend $\psi$ linearly. We check that the resulting $n^2 -1$ matrices $\psi(\inv{M}^j)$ are linearly independent over \(\splitk[\mathbb{I}]\). 
Observe that the matrices $\psi(\inv{M}^j)$  are not themselves invariants under the standard action, as defined in Section \ref{sec:motivations}. 
Consider two invariant matrices $\inv{M}$ and $\inv{N}$
\begin{eqnarray*}
\inv{N}\inv{M}\inv{v}_{i} &=& \sum_{k}\, \inv{N}\,\psi(\inv{M})_{k,i}\, \inv{v}_{k}=\sum_{k}\psi(\inv{M})_{k,i}\sum_{l}\psi(\inv{N})_{l,k}\, \inv{v}_{l}
\\&=&\sum_{l}\sum_{k}\psi(\inv{N})_{l,k}\,\psi(\inv{M})_{k,i}\, \inv{v}_{l}=\sum_{l} (\psi(\inv{N})\psi(\inv{M}))_{l,i}\, \inv{v}_{l}\,.
\end{eqnarray*}
It follows then that 
\[
[\inv{N} , \inv{M}] \inv{v}_{i}=\sum_{l}[\psi(\inv{N}),\psi(\inv{M})]_{l,i} \,\inv{v}_{l}
\]
that is,
\[\psi([\inv{N} , \inv{M}])=[\psi(\inv{N}) ,\psi(\inv{M}) ],
\]
in other words, $\psi$ is a Lie algebra homomorphism.

From the computational point of view and in preparation of the next step (namely the compuation of  Chevalley normal forms), once one has matrices with invariant coefficients it makes sense
to simplify them eliminating as many \(\mathbb{I}\)s as possible by taking linear combinations,
while taking care not to change those matrices of invariants with a \(\mathbb{I}\)-independent
characteristic polynomial (see the next Section \ref{sec:Chev}).

%%%%%%%%%%%%%%%%%%%%%%%%%%%%%%%%%%%%%%%%%%%%%%%%%%%%%%%%%%%%%%%%%%%%%%%%%%%%%%%%%%%%%%%%%%%%%%%%%%%%%%%%%%%%%%%%%%%%%%%%%%%%%%%%%%%%%%%%%%%%%%%%%%%%%%%%%%%%

\section{Chevalley normal form for Automorphic Lie Algebras}
\label{sec:Chev}

Even the 
most detailed Lie algebra books are a bit vague when it comes down to put a concrete Lie algebra
into Chevalley normal form over \(\bbbc\). 
In \cite{MR559927} the theory is derived for arbitrary fields, so this is getting closer to our problem.
One can imagine how much is written on how to do this over a polynomial ring.
In Bourbaki \cite{mr2001g:17006} the switch from the general set up to fields is quickly made in Chapter 1 after Section 3 (even though  this is relaxed again at times later on).

The original Lie algebra \(\mf{sl}(V)\) is of classical type and belongs to an isomorphism class \(A_h
\), with a corresponding $h\times h$ Cartan matrix.
Following the way the Chevalley normal form is computed over \(\bbbc\), the first task is to collect \(h\) 
commuting semisimple elements from the Lie algebra,
the \emph{Cartan subalgebra} or CSA (see e.g.~\cite{ful91a,MR1920389}), denoted by \(\mf{h}\). %CHECK THEY ARE IN

\begin{Remark}\label{rem:CSA}
In a simple Lie algebra over \(\bbbc\), a generic element will be semisimple
and one can construct a CSA around it.
In the automorphic case one requires not only semisimplicity but also that the eigenvalues of the matrices in the CSA are in the field extension \(\splitk\),
thus restricting the choice considerably.
In this sense one could say that Automorphic Lie Algebras are easier to deal with,
which is also reflected by the fact that, at least in the \(\mf{sl}(V)\) case, 
the characteristic equations could always be solved explicitly over \(\splitk\).
Working over the field extension of the irreducible representations of the group makes it easier to find explicit solutions,
even when the degree of the polynomial is five or six.
Of course, the computations are made more intricate by the fact that one works
not over \(\splitk\), but over \(\splitk[\AIJ[I][][\Gamma]]\).
\end{Remark}
The construction of the CSA \(\mf{h}\) starts with the search of a semisimple element in the Lie algebra of matrices of invariants such that all its eigenvalues are in  \(\splitk\).  Once such a matrix is found, it is tested for semisimplicity. This is done by considering the reduced characteristic polynomial, and checking that the matrix itself satisfies it 
(in the usual theory over \(\bbbc\) one looks for an element without degenerate eigenvalues, but this strategy proved not practical in our case).
Such an element, once found, can be put in diagonal form. The eigenvalue computation is done by Singular \cite{MR2363237}. 
We call this element \(h_1\) and store it in \(\mf{h}\).
We then proceed inductively.
We find a semisimple element \(h_i\) commuting with all the elements in \(\mf{h}\), but $\splitk$-linearly independent of the elements in \(\mf{h}\).
We then diagonalise \(h_i\) (leaving the other elements in \(\mf{h}\) diagonal).
Then we add \(h_i\) to \(\mf{h}\). We stop when we have $h$ elements in \(\mf{h}\). By construction, they are all linearly independent and diagonal matrices.
Next, one considers a $\splitk$-linear combination of these matrices to insure that their eigenvalues are constants and equal to the one prescribed by the Cartan matrix \cite[Plate I]{MR1890629}
 (corresponding to $\sln$ in the classification of simple Lie algebras).
 
We now give an algorithm to put the elements in \(\mf{h}\) in canonical form in the case of \(\sln\).
To this end, for every element $h_j$ in \(\mf{h}\) one computes the differences of the subsequent eigenvalues \[\alpha_i(h_j)= \lambda_i^j -\lambda_{i+1}^j,\quad i=1,\ldots,n-1.\] The canonical basis is the set of elements $H_k=\sum_{j=1}^{n-1}c_{j,k} h_j$ satisfying $\alpha_i(H_k)=a_{i,k}$, where $a_{i,k}$ are entries of the Cartan matrix of $A_{n-1}$. Since the matrix $(\alpha_i(h_j))_{i,j}$ is nondegenerate one can solve $c_{j,k}$, for each fixed $k$, in the equation \[\alpha_i(H_k)=\alpha_i(\sum_{j=1}^{n-1}c_{j,k} h_j)=\sum_{j=1}^{n-1}\alpha_i(h_j) c_{j,k} =a_{i,k}\]
and obtain $H_k$.
\begin{Example}
Consider, as an example, the case \(\alian{\bbbi_4}{G}{}\); 
one finds the elements
\(h_1=\diag\{-1,1,0\}\) and \(h_2=\diag\{1,0,-1\}\in\mf{sl}_3\). Let $A$ be the $\mf{sl}_3$ Cartan matrix and let \(E_{i,i}\) be the diagonal elementary matrix with \(1\) in the \(i\)th position.
We would like to have the CSA basis in the standard form \(H_1=E_{1,1}-E_{2,2}\) and \(H_2=E_{2,2}-E_{3,3}\).
We compute
\[
\alpha(h)=\begin{pmatrix} \alpha_1(h_1) & \alpha_1(h_2) \\ \alpha_2(h_1)& \alpha_2(h_2)\end{pmatrix}=\begin{pmatrix} -2 & 1 \\ 1 & 1 \end{pmatrix}\,.
\]
The matrix \(c\)  is then 
\[ \alpha(h)^{-1} A = -\frac{1}{3} \begin{pmatrix} 1 & -1 \\ -1 & -2 \end{pmatrix}\begin{pmatrix} 2 & -1 \\ -1 & 2 \end{pmatrix}=-\frac{1}{3} \begin{pmatrix} 3 & -3 \\ 0 & -3 \end{pmatrix}=\begin{pmatrix} -1 & 1 \\ 0 & 1 \end{pmatrix}\,,
\] 
i.e.~$H_1=-h_1$ and  $H_2=h_1+h_2$. $H_1$ and $H_2$ form a realisation of $A$ in the sense of Kac \cite{MR1104219}.
\end{Example}

\begin{Remark}\label{rem:alpha}
Here and in the following we will use the symbol $\alpha$ to denote the roots of the Lie algebra. This should be clear from the context and should not create confusion with the invariants introduced in the previous sections. 
\end{Remark}

Let $M_{\alpha_j}$ be a $\splitk[\AIJ[I][][\Gamma]]$-linear combination of the generators of the ALiA under investigation;
one computes them by solving
\[
[H_i, M_{\pm\alpha_j }]=\pm a_{j,i}M_{\pm\alpha_j }\,.
\]
The $M_{\alpha_j}$ are called weight vectors (of weight \(\alpha_j\)). 
Next one computes $[M_{\pm\alpha_j},M_{\pm\alpha_k}]$, $\alpha_j \neq \alpha_k$; 
if the commutator is not zero, the equation
\[
[H_i, M_{\pm(\alpha_j+\alpha_k) }]=\pm( a_{j,i}+a_{k,i}) M_{\pm(\alpha_j+\alpha_k)}\, 
\]
is solved.
Recursively, one computes all the weight vectors in the Chevalley normal form.
When all weight vectors have been computed, it is explicitly checked that the transformation from the old generators to this new basis is invertible over $\splitk[\AIJ[I][][\Gamma]]$.

Notice that we do not have an existence proof of a Chevalley normal form, however the computation finds always a suitable set of generators such that the algebra is in normal form, so the existence is proven by construction. Since we restrict ourselves to irreducible representations, we only have a finite number of cases to consider.

In the next Sections \ref{sec:dim3}--\ref{sec:dim35} we list Chevalley normal forms for \((\mf{sl}(V)\otimes \splitk(\lambda))_{\zeta}^{G}\) and we prove the following main result:
\begin{Theorem}\label{teo:AL}
Let $V$ be an irreducible representations of $G^\flat$ and $V^\prime$ be an irreducible representation of $G^{\prime\flat}$, where $G$ and $G^\prime$ are isomorphic to the tetrahedral group $\bbbt$, the octahedral group $\bbbo$ or the icosahedral  group $\bbbi$. 
Let $\zeta$ and $\zeta^\prime$ be $G$, $G^\prime$- classical relative invariants (see Example \ref{ex:cit}); 
then \((\mf{g}(V)\otimes \splitk(\lambda))_{\zeta}^{G}\)  is isomorphic to \((\mf{g'}(V')\otimes \splitk'(\lambda))_{\zeta'}^{G'}\)  if and only if $\mf{g}(V)$ is isomorphic to $  \mf{g}^\prime(V^\prime)$ as Lie algebra, 
where \(\mf{g}, \mf{g}^\prime=\mf{sl}\), and \(\kappa_\zeta=\kappa_{\zeta'}\), where the \(\kappa_\zeta\)s can be found in Table \ref{tab:codimg}.
\end{Theorem}
\begin{Corollary}\label{cor:AL}
The statement of Theorem \ref{teo:AL} is true also if one includes the dihedral group $\bbbd_N$ in the list of groups (see \cite{KLS-DN}).
\end{Corollary}\
%%%%%%%%%%%%%%%%%%%%%%%%%%%%%%%%%%%%%%%%%%%%%%%%%%%%%%%%%%%%%%%%%%%%%%%%%%%%%

\subsection{Notation} 
Before formulating our result, let us introduce some notation which will be handy in the following; consider, as an example, the case \(\alian{V}{G}{}\), where $V=\bbbt_4^\flat$. 
After computing the Chevalley normal form as described in Section \ref{sec:Chev}, we find
\[
M_{\alpha_1}=\begin{pmatrix}0 & \AIJ[J][\gamma][\alpha]\\ 0 &0\end{pmatrix}\,,\quad 
M_{-\alpha_1}=\begin{pmatrix}0 & 0\\ \AIJ[I][\beta][\alpha] &0\end{pmatrix}\,,\quad 
H_1=\begin{pmatrix}1 & 0\\ 0&-1\end{pmatrix},
\]
where the symbol \(\alpha_i\) stands for the the root and \(\alpha\) stands for the ground form. 
%%%%%
In terms of the original invariant matrices this Cartan-Weyl basis reads (see also Table \ref{tab:MolienT}):
\begin{eqnarray*}
&&H_{1}= - 1/2949120\ \mf{M}_7^4 - 1/9437184\ \mf{M}_7^6 + 1/5505024\ \mf{N}_7^6\,,\\
&&M_{\alpha_1}= + 1/11796480\ \mf{M}_7^4\ \ \AIJ[J][\gamma][\alpha] 
      +1/37748736\ \mf{M}_7^6\ \ \AIJ[J][\gamma][\alpha] - 1/22020096\ \mf{N}_7^6 + 1/22020096\ \mf{N}_7^6\ 
      \ \AIJ[I][\beta][\alpha]\,,\\
&&M_{-\alpha_1}= 1/2949120\ \mf{M}_7^4\ \ \AIJ[J][\gamma][\alpha] + 1/9437184\ \mf{M}_7^6 - 1/
      9437184\ \mf{M}_7^6\ \ \AIJ[I][\beta][\alpha] - 1/5505024\ \mf{N}_7^6\ \ \AIJ[J][\gamma][\alpha]\,.\\
\end{eqnarray*}
%%%%%
We introduce the following short-hand notation %(which nicely illustrates the possible \(\alpha\)-confusion)
\[
\nf{\bbbt_4^\flat}=M_{\alpha_1}+M_{-\alpha_1}=\begin{bmatrix}0 & \AIJ[J][\gamma][\alpha]\\ \AIJ[I][\beta][\alpha]&0\end{bmatrix}
\]
where we take the sum of all weight vectors;
we will refer to this as the \emph{Chevalley model} of the Automorphic Lie Algebra.
\begin{Remark}
\(\nf{\bbbt_4^\flat}\) can be considered as a \(1\)-form with arguments in the root system \(A_1\) and values in the space of monomials in \(\AIJ[I][\beta][\alpha]\) and \(\AIJ[J][\gamma][\alpha]\), the coboundary operator \(\mathsf{d}^1\) of which determines
the occurrence of these monomials in the structure constants of the ALiA.
\end{Remark}
\begin{Remark}
We recall that \(\AIJ[J][][]\) is the invariant related to the relative invariant with the lowest \(\kappa\), see Section \ref{sec:af}.
If there is equality, for instance if \(\coa=\cob\), then in \(\AIJ[I][\alpha][\gamma]\) and \(\AIJ[J][\beta][\gamma]\),
one can interchange the \(\AIJ[I][][]\) and the \(\AIJ[J][][]\), without changing the isomorphism type of the Chevalley normal form.
\end{Remark}
The Chevalley normal form can be reconstructed from the Cartan matrix (in this case the $1\times 1$ matrix $\begin{pmatrix} 2\end{pmatrix}$) and from the Chevalley model above. The Lie brackets are
\[
\begin{array}{l}
[M_{\alpha_1},M_{-\alpha_1}]=\AIJ[I][\beta][\alpha]\AIJ[J][\gamma][\alpha] H_1\\
\\
{[}H_1,M_{\pm\alpha_1}]=\pm 2M_{\alpha_1}\,.
\end{array}
\]
For any \(A_h\), given \(\alpha=\sum_{k=1}^{h} m_k \alpha_k\) and \(m_k\in\bbbn, k=1,\ldots,h\),
the following holds:
\[
[ M_{\alpha} , M_{-\alpha} ]=\langle M_{\alpha} , M_{-\alpha} \rangle H_\alpha,
\]
where \(H_\alpha= \sum_{k=1}^{h} m_k H_k\) and \(\langle\cdot,\cdot\rangle\) is the traceform.

In the following we list all cases, ordered by \(\dim\mf{g}(V)\).
The Chevalley normal form will be compared to a model computed from the structure constants of one  of the computed Lie algebras with the given Dynkin diagram and written as, for example, \(\|A_2\|\). This model is not unique.
\newcommand\aliam[4]{(\mf{sl}_{#3}\otimes \splitk(\lambda))_{#4}^{#2}}

%%%%%%%%%%%%%%%%%%%%%%%%%%%%%%%%%%%%%%%%%%%%%%%%%%%%%%%%
\begin{Definition}
We denote by  \(\|A_{n}^{(k,l)}\|\) the Automorphic Lie Algebra model based on \(\mf{sl}_{n+1}\) and with \(k\) \(\AIJ[I][][]\)s and \(l\) \(\AIJ[J][][]\)s in its
Cartan-Weyl basis. This defines the ALiA type \(A_{n}^{(k,l)}\). 
It will have the same Cartan matrix as \(A_n\) and the specifics of the particular Chevalley model, that is to say, which elements have an \(\AIJ[I][][]\) and which have a  \(\AIJ[J][][]\),
will be fixed in the sequel.
\end{Definition}
Let $\Phi$ be the root system of the base Lie algebra and let $\Phi^+$ be a choice of positive roots; together with the model  \(\|A_{n}^{(k,l)}\|\)  we also consider 
\[
K_{\mf{b}}(\mf{sl}_{n})_\zeta= \sum_{\alpha\in \Phi^+}\, \langle e_\alpha,e_{-\alpha}\rangle =a+b\AIJ[I][][]+c \AIJ[J][][]+d\AIJ[IJ][][]\,.
\]
In the example above the sum equals $\mathbbm{I}\mathbbm{J}$. Computational evidence suggests that this is an invariant.
\begin{Definition}
We denote by  \(\aliam{V}{G}{n}{\zeta}\) the $G$-Automorphic Lie Algebra based on \(\mf{sl}(V)\), $\dim(V)=n$,  
with poles confined at the $G$-orbit $\Gamma_\zeta$, \(\zeta=\alpha\), \(\beta\) or \(\gamma\). 
\end{Definition}
%%%%%%%%%%%%%%%%%%%%%%%%%%%%%%%%%%%%%%%%%%%%%%%%%%%%%%%%%%%

\subsection{Automorphic Lie Algebras \(\aliam{V}{G}{2}{\zeta}\)}
%\subsection{$\dim\mf{g}(V)=3$}
\label{sec:dim3}
Let the model for \(\aliam{V}{G}{2}{\zeta}\) be
\[
\|A_1^{(1,1)}\|=\begin{bmatrix}
0 & \mathbbm{I} \\
\mathbbm{J} & 0
\end{bmatrix}\,,\quad
K_{\mf{b}}(\mf{sl}_{2})_\zeta=\mathbbm{I}\mathbbm{J}
\]
%\input\ifnum\where=1 /lhome/jan/Dropbox/WORK200/\fi ala/Result\dir/A\dynnum/KillingNormalForm\dirp\pole.tex
% The normal form  of Ax5c2
% on Tue Feb 17 16:57:26 2015
% in ResultAx5c2/A1/KillingNormalFormAx5c2pa.tex
where 
\(\zeta=\alpha\), \(\beta\) or \(\gamma\).

\begin{Theorem}[\(\aliam{V}{G}{2}{\zeta}\)]
\label{teo:dim3}
All Automorphic Lie Algebras \(\aliam{V}{G}{2}{\zeta}\), \(\zeta=\alpha,\beta,\gamma\),  are of type \(A_1^{(1,1)}\) and therefore isomorphic.
\end{Theorem}
%%%%%%%%%%%%%%%%%%%%%%%%%%%%%%%%%%%%%%%%%%%%%%%%%%%%%%%%%%%
\begin{proof}
\def\dynnum{1}
We give the Chevalley model together with its intertwining operator \({\cal{I}}_{\mf{sl}(V)}\) with respect to
\(\|A_\dynnum^{(1,1)}\|\),
i.e.
\[
\|\mf{sl}(V)\| {\cal{I}}_{\mf{sl}(V)}={\cal{I}}_{\mf{sl}(V)} \|A_{\dynnum}^{(1,1)}\|.
\]
\renewcommand{\arraystretch}{1}
%\def\pole{a}
%\def\dynnum{1}
%\def\dir{Ax3c4}
%\def\dirp{Ax3c4p\pole}
%%\input\ifnum\where=1 /lhome/jan/Dropbox/WORK200/\fi ala/Result\dir/A\dynnum/Chevalley\dirp.tex 
%$\begin{bmatrix}
%0&\AIJ[J][\gamma][\alpha]\\
%\AIJ[I][\beta][\alpha]&0\end{bmatrix}$
%\def\pole{a}
%\def\dynnum{1}
%\def\dir{Ax3c6}
%\def\dirp{Ax3c6p\pole}
%%\input\ifnum\where=1 /lhome/jan/Dropbox/WORK200/\fi ala/Result\dir/A\dynnum/Chevalley\dirp.tex 
%&
%$\begin{bmatrix}
%0&\AIJ[I][\beta][\alpha]\\
%\AIJ[J][\gamma][\alpha]&0\end{bmatrix}$
%\\
%& & \\
%\hline
%%& & \\
%Intertwining operator&& \\
%${\cal{I}}_{\mf{sl}(V)}$
%& \def\pole{a}
%\def\dynnum{1}
%\def\dir{Ax3c4}
%\def\dirp{Ax3c4p\pole}
%%\input\ifnum\where=1 /lhome/jan/Dropbox/WORK200/\fi ala/Result\dir/A\dynnum/Intertwiner\dirp.tex 
%$\begin{pmatrix}
%\AIJ[J][\gamma][\alpha]&0\\
%0&\AIJ[I][\beta][\alpha]\end{pmatrix}$
%&
%\def\pole{a}
%\def\dynnum{1}
%\def\dir{Ax3c6}
%\def\dirp{Ax3c6p\pole}
%%\input\ifnum\where=1 /lhome/jan/Dropbox/WORK200/\fi ala/Result\dir/A\dynnum/Intertwiner\dirp.tex 
%$\begin{pmatrix}
%1&0\\
%0&1\end{pmatrix}$
%\\
%& & \\
%\hline
%\end{tabular}
%\end{table}
%\begin{center}
\begin{table}[H]
\begin{center}
\begin{tabular}{c|cr}
\hline
%\begin{tabular}{|c|c|c|c|c|c|c|c|} \hline
%& & \\
%poles at $\Gamma_\alpha$\\
Irreducible representation \\
$V$
 & $\bbbt_{4}$
, $\bbbt_{5}$
, $\bbbo_{3}$
, $\bbbo_{5}$
, $\bbbi_{2}$
, $\bbbi_{3}$
&
$\bbbt_{6}$
, $\bbbo_{4}$
\\
& & \\
 \hline\hline
% & & \\
Chevalley model \\
$\|\mf{sl}(V)\|$ &
\def\pole{a}
\def\dynnum{1}
\def\dir{Ax3c4}
\def\dirp{Ax3c4p\pole}
$\begin{bmatrix}
0&\AIJ[J][\gamma][\alpha]\\
\AIJ[I][\beta][\alpha]&0\end{bmatrix}$
%\input\ifnum\where=1 /lhome/jan/Dropbox/WORK200/\fi ala/Result\dir/A\dynnum/Chevalley\dirp.tex 
%&
%\def\pole{a}
%\def\dynnum{1}
%\def\dir{Ax3c5}
%\def\dirp{Ax3c5p\pole}
%\input\ifnum\where=1 /lhome/jan/Dropbox/WORK200/\fi ala/Result\dir/A\dynnum/Chevalley\dirp.tex 
&
\def\pole{a}
\def\dynnum{1}
\def\dir{Ax3c6}
\def\dirp{Ax3c6p\pole}
$\begin{bmatrix}
0&\AIJ[I][\beta][\alpha]\\
\AIJ[J][\gamma][\alpha]&0\end{bmatrix}$
\\
& & \\
\hline
%& & \\
Intertwining operator \\
${\cal{I}}_{\mf{sl}(V)}$
& \def\pole{a}
\def\dynnum{1}
\def\dir{Ax3c4}
\def\dirp{Ax3c4p\pole}
$\begin{pmatrix}
\AIJ[J][\gamma][\alpha]&0\\
0&\AIJ[I][\beta][\alpha]\end{pmatrix}$
%\input\ifnum\where=1 /lhome/jan/Dropbox/WORK200/\fi ala/Result\dir/A\dynnum/Intertwiner\dirp.tex 
%&
%\def\pole{a}
%\def\dynnum{1}
%\def\dir{Ax3c5}
%\def\dirp{Ax3c5p\pole}
%\input\ifnum\where=1 /lhome/jan/Dropbox/WORK200/\fi ala/Result\dir/A\dynnum/Intertwiner\dirp.tex 
&
\def\pole{a}
\def\dynnum{1}
\def\dir{Ax3c6}
\def\dirp{Ax3c6p\pole}
$\begin{pmatrix}
1&0\\
0&1\end{pmatrix}$

\\
& & \\
\hline
\end{tabular}
\end{center}
\caption{Chevalley models and intertwining operators for  \(\aliam{V}{G}{2}{\alpha}\).}
\end{table}
%\end{center}
%%%%%%%%%%%%%%%%%%%%%%%%%%
\begin{table}[H]
\begin{center}
\begin{tabular}{c|ccc}
\hline
Irreducible representation \\
$V$
 & $\bbbt_{4}$
, $\bbbt_{5}$
, $\bbbo_{3}$
, $\bbbo_{5}$
&$\bbbt_{6}$
, $\bbbi_{2}$
, $\bbbi_{3}$
&
$\bbbo_{4}$
\\
& & & \\
 \hline\hline
% & & \\
Chevalley model \\
$\|\mf{sl}(V)\|$ &
$\begin{bmatrix}
0&\AIJ[I][\alpha][\beta]\\
\AIJ[J][\gamma][\beta]&0\end{bmatrix}$
&
$\begin{bmatrix}
0&\AIJ[J][\gamma][\beta]\\
\AIJ[I][\alpha][\beta]&0\end{bmatrix}$
%\input\ifnum\where=1 /lhome/jan/Dropbox/WORK200/\fi ala/Result\dir/A\dynnum/Chevalley\dirp.tex 
%&
%\def\pole{a}
%\def\dynnum{1}
%\def\dir{Ax3c5}
%\def\dirp{Ax3c5p\pole}
%\input\ifnum\where=1 /lhome/jan/Dropbox/WORK200/\fi ala/Result\dir/A\dynnum/Chevalley\dirp.tex 
&
$\begin{bmatrix}
0& 1\\
\AIJ[I][\alpha][\beta]\AIJ[J][\gamma][\beta]&0\end{bmatrix}$
\\
& & & \\
\hline
%& & \\
Intertwining operator \\
${\cal{I}}_{\mf{sl}(V)}$
& 
$\begin{pmatrix}
1&0\\
0&1\end{pmatrix}$
&
$\begin{pmatrix}
\AIJ[J][\gamma][\beta]&0\\
0&\AIJ[I][\alpha][\beta]\end{pmatrix}$
&
$\begin{pmatrix}
1&0\\
0&\AIJ[I][\alpha][\beta]\end{pmatrix}$
%\input\ifnum\where=1 /lhome/jan/Dropbox/WORK200/\fi ala/Result\dir/A\dynnum/Intertwiner\dirp.tex 
%&
%\def\pole{a}
%\def\dynnum{1}
%\def\dir{Ax3c5}
%\def\dirp{Ax3c5p\pole}
%\input\ifnum\where=1 /lhome/jan/Dropbox/WORK200/\fi ala/Result\dir/A\dynnum/Intertwiner\dirp.tex 
\\
& & & \\
\hline
\end{tabular}
\end{center}
%}
\caption{Chevalley models and intertwining operators for  \(\aliam{V}{G}{2}{\beta}\).}
\end{table}
%%%%%%%%%%%%%%%%%%%%%%%%%%
%\vspace{0.3cm}
\begin{table}[H]
\begin{center}
\begin{tabular}{c|cccc}
\hline
%\begin{tabular}{|c|c|c|c|c|c|c|c|} \hline
%& & & &\\
%poles at $\Gamma_\alpha$\\
Irreducible representation \\
$V$
 & $\bbbt_{4}$
, $\bbbt_{5}$
&
$\bbbt_{6}$
&
$\bbbo_{3}$
, $\bbbi_{2}$
, $\bbbi_{3}$
&
$\bbbo_{4}$,
$\bbbo_{5}$
\\
& & & &\\
 \hline\hline
% & & \\
Chevalley model \\
$\|\mf{sl}(V)\|$ &
$\begin{bmatrix}
0& 1\\
\AIJ[I][\alpha][\gamma]\AIJ[J][\beta][\gamma]&0\end{bmatrix}$
&
$\begin{bmatrix}
0&\AIJ[I][\alpha][\gamma]\AIJ[J][\beta][\gamma]\\
1&0\end{bmatrix}$
&
$\begin{bmatrix}
0&\AIJ[J][\beta][\gamma]\\
\AIJ[I][\alpha][\gamma]&0\end{bmatrix}$
%\input\ifnum\where=1 /lhome/jan/Dropbox/WORK200/\fi ala/Result\dir/A\dynnum/Chevalley\dirp.tex 
%&
%\def\pole{a}
%\def\dynnum{1}
%\def\dir{Ax3c5}
%\def\dirp{Ax3c5p\pole}
%\input\ifnum\where=1 /lhome/jan/Dropbox/WORK200/\fi ala/Result\dir/A\dynnum/Chevalley\dirp.tex 
&
$\begin{bmatrix}
0& \AIJ[I][\alpha][\gamma]\\
\AIJ[J][\beta][\gamma]&0\end{bmatrix}$
\\
& & & &\\
\hline
%& & \\
Intertwining operator \\
${\cal{I}}_{\mf{sl}(V)}$
& 
$\begin{pmatrix}
1&0\\
0&\AIJ[I][\alpha][\gamma]\end{pmatrix}$
&
$\begin{pmatrix}
\AIJ[J][\beta][\gamma]& 0\\
0&1\end{pmatrix}$
&
$\begin{pmatrix}
\AIJ[J][\beta][\gamma]& 0\\
0&\AIJ[I][\alpha][\gamma]\end{pmatrix}$
&
$\begin{pmatrix}
1& 0\\
0&1\end{pmatrix}$
%\input\ifnum\where=1 /lhome/jan/Dropbox/WORK200/\fi ala/Result\dir/A\dynnum/Intertwiner\dirp.tex 
%&
%\def\pole{a}
%\def\dynnum{1}
%\def\dir{Ax3c5}
%\def\dirp{Ax3c5p\pole}
%\input\ifnum\where=1 /lhome/jan/Dropbox/WORK200/\fi ala/Result\dir/A\dynnum/Intertwiner\dirp.tex 
\\
& & & \\
\hline
\end{tabular}
\end{center}
\caption{Chevalley models and intertwining operators for  \(\aliam{V}{G}{2}{\gamma}\).}
\end{table}
\end{proof}

%%%%%%%%%%%%%%%%%%%%%%%%%%%%%%%%%%%%%%%%%%%%%%%%%%%%%%%%%%%
For the proofs of the following theorems we refer to Appendix \ref{app:Chev}.
%%%%%%%%%%%%%%%%%%%%%%%%%%%%%%%%%%%%%%%%%%%%%%%%%%%%%%%%%%%
\subsection{Automorphic Lie Algebras \(\aliam{V}{G}{3}{\zeta}\)} 
%\subsection{$\dim\mf{g}(V)=8$}
\label{sec:dim8}
\def\pole{a}
\def\dynnum{2}
\def\dir{Ax5c4}
\def\dirp{Ax5c4p}
%\input\ifnum\where=1 /lhome/jan/Dropbox/WORK200/\fi ala/Result\dir/A\dynnum/SerreSNormalForm\dirp\pole.tex
% The normal form  of Ax5c4
% on Thu Feb 12 13:07:36 2015
% file ResultAx5c4/A2/SerreSNormalFormAx5c4.tex
\subsubsection{Poles in $\alpha$ and $\beta$}
\label{sec:dim8ab}

Let the model for \(\aliam{V}{G}{3}{\zeta}\), $\zeta=\alpha,\beta$, be
\[
\|A_2^{(3,2)}\|=\begin{bmatrix}
0 & \mathbbm{I} & \mathbbm{I} \\
\mathbbm{J} & 0 & \mathbbm{I} \\
\mathbbm{J} & 1 & 0 \end{bmatrix}\,,\quad
K_{\mf{b}}(\mf{sl}_{4})_{\alpha,\beta}=\mathbbm{I}+2\mathbbm{I}\mathbbm{J}\,.
%\input\ifnum\where=1 /lhome/jan/Dropbox/WORK200/\fi ala/Result\dir/A\dynnum/KillingNormalForm\dirp\pole.tex
% The normal form  of Ax5c4
% on Thu Feb 12 13:07:36 2015
% in ResultAx5c4/A2/KillingNormalFormAx5c4pa.tex
\]

%%\input\ifnum\where=1 /lhome/jan/Dropbox/WORK200/\fi ala/Result\dir/A\dynnum/KillingNormalForm\dirp\pole.tex
%% The normal form  of Ax5c4
%% on Thu Feb 12 13:07:36 2015
%% in ResultAx5c4/A2/KillingNormalFormAx5c4pa.tex

%%%%%%%%%%%%%%%%%%%%%%%%%%%%%%%%%%%%%%%%%%%%%%%%%%%%%%%%%%%
%\subsubsection{Chevalley model for \(\aliam{V}{G}{3}{\alpha,\beta}\)}
\begin{Theorem}[\(\aliam{V}{G}{3}{\zeta}\), $\zeta=\alpha,\beta$]
\label{teo:dim8ab}
All Automorphic Lie Algebras \(\aliam{V}{G}{3}{\zeta}\), $\zeta=\alpha,\beta$,  are isomorphic and of 
type  \(A_{2}^{(3,2)}\).
\end{Theorem}
\def\pole{c}
%\input\ifnum\where=1 /lhome/jan/Dropbox/WORK200/\fi ala/Result\dir/A\dynnum/SerreSNormalForm\dirp\pole.tex
% The normal form  of Ax5c4
% on Sun Feb  8 15:46:55 2015
% file /lhome/jan/Dropbox/WORK200/ala/ResultAx5c4/A2/SerreSNormalFormAx5c4.tex
\subsubsection{Poles in $\gamma$}
\label{sec:dim3\pole}
Let the model for \(\aliam{V}{G}{3}{\gamma}\) be
\[
\|A_2^{(3,3)}\|=\begin{bmatrix}
0 & \mathbbm{I} & \mathbbm{I} \\
\mathbbm{J} & 0 & \mathbbm{I} \\
\mathbbm{J} & \mathbbm{J} & 0
\end{bmatrix}\,,\quad
K_{\mf{b}}(\mf{sl}_{4})_{\gamma}=3\mathbbm{I}\mathbbm{J}\,.
%\input\ifnum\where=1 /lhome/jan/Dropbox/WORK200/\fi ala/Result\dir/A\dynnum/KillingNormalForm\dirp\pole.tex
% The normal form  of Ax5c4
% on Sun Feb  8 15:46:55 2015
% in /lhome/jan/Dropbox/WORK200/ala/ResultAx5c4/A2/KillingNormalFormAx5c4pc.tex
\]
%% in /lhome/jan/Dropbox/WORK200/ala/ResultAx5c4/A2/KillingNormalFormAx5c4pc.tex
\begin{Theorem}[\(\aliam{V}{G}{3}{\gamma}\)]
\label{teo:dim8c}
All Automorphic Lie Algebras \(\aliam{V}{G}{3}{\gamma}\) are isomorphic and of 
type  \(A_{2}^{(3,3)}\).
\end{Theorem}
%%%%%%%%%%%%%%%%%%%%%%%%%%%%%%%%%%%%%%%%%%%%%%%%%%%%%%%%%%%

\subsection{Automorphic Lie Algebras \(\aliam{V}{G}{4}{\zeta}\)}
\label{sec:dim15}
\def\pole{a}
\def\dir{Ax5c6}
\def\dirp{Ax5c6p}
\def\dynnum{3}
%\input\ifnum\where=1 /lhome/jan/Dropbox/WORK200/\fi ala/Result\dir/A\dynnum/SerreSNormalForm\dirp\pole.tex 
% The normal form  of Ax5c6
% on Thu Feb 12 13:27:10 2015
% file ResultAx5c6/A3/SerreSNormalFormAx5c6.tex
\subsubsection{Poles in $\alpha$}
\label{sec:dim15a}
Let the model for \(\aliam{V}{G}{4}{\alpha}\) be
\[
\|A_3^{(5,4)}\|=\begin{bmatrix}
0 & \mathbbm{I} & \mathbbm{I} & \mathbbm{I} \\
1 & 0 & 1 & \mathbbm{I} \\
\mathbbm{J} & \mathbbm{J} & 0 & \mathbbm{I} \\
\mathbbm{J} & \mathbbm{J} & 1 & 0
\end{bmatrix}\,,\quad
K_{\mf{b}}(\mf{sl}_{4})_{\alpha}=2\mathbbm{I}+\mathbbm{J}+3\mathbbm{I}\mathbbm{J}\,.
\]
%%\input\ifnum\where=1 /lhome/jan/Dropbox/WORK200/\fi ala/Result\dir/A\dynnum/KillingNormalForm\dirp\pole.tex
%% The normal form  of Ax5c6
%% on Thu Feb 12 13:27:10 2015
%% in ResultAx5c6/A3/KillingNormalFormAx5c6pa.tex

\begin{Theorem}[\(\aliam{V}{G}{4}{\alpha}\)]
\label{teo:dim15a}
All Automorphic Lie Algebras \(\aliam{V}{G}{4}{\alpha}\) are isomorphic and of 
type  \(A_{3}^{(5,4)}\).
\end{Theorem}

%%\input\ifnum\where=1 /lhome/jan/Dropbox/WORK200/\fi ala/Result\dir/A\dynnum/SerreSNormalForm\dirp\pole.tex 
%% The normal form  of Ax5c6
%% on Thu Feb 12 14:06:20 2015
%% in ResultAx5c6/A3/KillingNormalFormAx5c6pb.tex
%\input\ifnum\where=1 /lhome/jan/Dropbox/WORK200/\fi ala/Result\dir/A\dynnum/KillingNormalForm\dirp\pole.tex
% The normal form  of Ax5c6
% on Thu Feb 12 14:06:20 2015
% in ResultAx5c6/A3/KillingNormalFormAx5c6pb.tex

%%%%%%%%%%%%%%%%%%%%%%%%%%%%%%%%%%%%%%%%%%%%%%%%%%%%%%%%%%%
\def\pole{b}
\subsubsection{Poles in $\beta$}
\label{sec:dim15b}
Let the model for \(\aliam{V}{G}{4}{\beta}\) be
% The normal form  of Ax5c6
% on Sun Apr  5 18:19:55 2015
% file ResultAx5c6/A3/SerreSNormalFormAx5c6.tex
\[
\|A_3^{(6,4)}\|=\begin{bmatrix}
0 & \mathbbm{I} & \mathbbm{I} & \mathbbm{I} \\
1 & 0 & \mathbbm{I} & \mathbbm{I} \\
\mathbbm{J} & \mathbbm{J} & 0 & \mathbbm{I} \\
\mathbbm{J} & \mathbbm{J} & 1 & 0
\end{bmatrix}\,,\quad
K_{\mf{b}}(\mf{sl}_{4})_{\gamma}=2\mathbbm{I}+4\mathbbm{I}\mathbbm{J}\,.
\]
\begin{Theorem}[\(\aliam{V}{G}{4}{\beta}\)]
\label{teo:dim15b}
All Automorphic Lie Algebras \(\aliam{V}{G}{4}{\beta}\) are isomorphic and of type  \(A_{3}^{(6,4)}\).
\end{Theorem}

\def\pole{c}
%\input\ifnum\where=1 /lhome/jan/Dropbox/WORK200/\fi ala/Result\dir/A\dynnum/SerreSNormalForm\dirp\pole.tex 
% The normal form  of Ax5c6
% on Sun Feb  8 16:39:31 2015
% file /lhome/jan/Dropbox/WORK200/ala/ResultAx5c6/A3/SerreSNormalFormAx5c6.tex
\subsubsection{Poles in $\gamma$}
\label{sec:dim15c}
Let the model for \(\aliam{V}{G}{4}{\gamma}\) be
\[
\|A_3^{(6,5)}\|=\begin{bmatrix}
0 & \mathbbm{I} & \mathbbm{I} & \mathbbm{I} \\
\mathbbm{J} & 0 & \mathbbm{I} & \mathbbm{I} \\
\mathbbm{J} & \mathbbm{J} & 0 & \mathbbm{I} \\
\mathbbm{J} & \mathbbm{J} & 1 & 0
\end{bmatrix}\,,\quad
K_{\mf{b}}(\mf{sl}_{4})_{\gamma}=\mathbbm{I}+5\mathbbm{I}\mathbbm{J}\,.
%\input\ifnum\where=1 /lhome/jan/Dropbox/WORK200/\fi ala/Result\dir/A\dynnum/KillingNormalForm\dirp\pole.tex
% The normal form  of Ax5c6
% on Sun Feb  8 16:39:31 2015
% in /lhome/jan/Dropbox/WORK200/ala/ResultAx5c6/A3/KillingNormalFormAx5c6pc.tex
\]

%%\input\ifnum\where=1 /lhome/jan/Dropbox/WORK200/\fi ala/Result\dir/A\dynnum/KillingNormalForm\dirp\pole.tex
%% The normal form  of Ax5c6
%% on Sun Feb  8 16:39:31 2015
%% in /lhome/jan/Dropbox/WORK200/ala/ResultAx5c6/A3/KillingNormalFormAx5c6pc.tex
\begin{Theorem}[\(\aliam{V}{G}{4}{\gamma}\)]
\label{teo:dim15c}
All Automorphic Lie Algebras \(\aliam{V}{G}{4}{\gamma}\) are isomorphic and of 
type  \(A_{3}^{(6,5)}\).
\end{Theorem}
%%%%%%%%%%%%%%%%%%%%%%%%%%%%%%%%%%%%%%%%%%%%%%%%%%%%%%%%%%%%%%%%%%%%%%%%%%%%%%%%%%%%%%%%%%%%%%%%%%%%%%%%%%%%%%%%%%%%%%
\subsection{Automorphic Lie Algebras \(\aliam{V}{G}{5}{\zeta}\)}
\label{sec:dim24}
\def\pole{a}
\def\dir{Ax5c8}
\def\dirp{Ax5c8p}
\def\dynnum{4}
%\input\ifnum\where=1 /lhome/jan/Dropbox/WORK200/\fi ala/Result\dir/A\dynnum/SerreSNormalForm\dirp\pole.tex 
% The normal form  of Ax5c8
% on Sun Feb  8 17:07:20 2015
% file /lhome/jan/Dropbox/WORK200/ala/ResultAx5c8/A4/SerreSNormalFormAx5c8.tex
\subsubsection{Poles in $\alpha$}
\label{sec:dim24a}
Let the model for \(\aliam{V}{G}{5}{\alpha}\) be
\[
\|A_4^{(8,6)}\|=\begin{bmatrix}
0 & 1 & \mathbbm{I} & \mathbbm{I} & \mathbbm{I} \\
1 & 0 & \mathbbm{I} & \mathbbm{I} & \mathbbm{I} \\
\mathbbm{J} & \mathbbm{J} & 0 &1 & \mathbbm{I} \\
\mathbbm{J} & \mathbbm{J} &1 & 0 & \mathbbm{I} \\
\mathbbm{J} & \mathbbm{J} &1 & 1& 0
\end{bmatrix}\,,\quad
K_{\mf{b}}(\mf{sl}_{5})_{\alpha}=2+2\mathbbm{I}+6\mathbbm{I}\mathbbm{J}\,.
%\input\ifnum\where=1 /lhome/jan/Dropbox/WORK200/\fi ala/Result\dir/A\dynnum/KillingNormalForm\dirp\pole.tex
% The normal form  of Ax5c8
% on Sun Feb  8 17:07:20 2015
% in /lhome/jan/Dropbox/WORK200/ala/ResultAx5c8/A4/KillingNormalFormAx5c8pa.tex
\]

%%\input\ifnum\where=1 /lhome/jan/Dropbox/WORK200/\fi ala/Result\dir/A\dynnum/KillingNormalForm\dirp\pole.tex
%% The normal form  of Ax5c8
%% on Sun Feb  8 17:07:20 2015
%% in /lhome/jan/Dropbox/WORK200/ala/ResultAx5c8/A4/KillingNormalFormAx5c8pa.tex

\begin{Theorem}[\(\aliam{V}{G}{5}{\alpha}\)]
\label{teo:dim24a}
All Automorphic Lie Algebras \(\aliam{V}{G}{5}{\alpha}\) are isomorphic and of 
type  \(A_{4}^{(8,6)}\).
\end{Theorem}
\def\pole{b}
%\input\ifnum\where=1 /lhome/jan/Dropbox/WORK200/\fi ala/Result\dir/A\dynnum/SerreSNormalForm\dirp\pole.tex 
% The normal form  of Ax5c8
% on Sun Feb  8 17:27:31 2015
% file /lhome/jan/Dropbox/WORK200/ala/ResultAx5c8/A4/SerreSNormalFormAx5c8.tex
%%%%%%%%%%%%%%%%%%%%%%%%%%%%%%%%%%%%%%%%%%%%%%%%%%%%%%%%%%%
\subsubsection{Poles in $\beta$}
\label{sec:dim24b}
Let the model for \(\aliam{V}{G}{5}{\beta}\) be
\[
\|A_4^{(10,6)}\|=\begin{bmatrix}
0 & \mathbbm{I} & \mathbbm{I} & \mathbbm{I} & \mathbbm{I} \\
1 & 0 & \mathbbm{I} & \mathbbm{I} & \mathbbm{I} \\
\mathbbm{J} & \mathbbm{J} & 0 & \mathbbm{I} & \mathbbm{I} \\
\mathbbm{J} & \mathbbm{J} & 1& 0 & \mathbbm{I} \\
\mathbbm{J} & \mathbbm{J} &1 & 1& 0
\end{bmatrix}\,,\quad
K_{\mf{b}}(\mf{sl}_{5})_{\beta}=4\mathbbm{I}+6\mathbbm{I}\mathbbm{J}\,.
%\input\ifnum\where=1 /lhome/jan/Dropbox/WORK200/\fi ala/Result\dir/A\dynnum/KillingNormalForm\dirp\pole.tex
% The normal form  of Ax5c8
% on Sun Feb  8 17:27:31 2015
% in /lhome/jan/Dropbox/WORK200/ala/ResultAx5c8/A4/KillingNormalFormAx5c8pb.tex
\]

\begin{Theorem}[\(\aliam{V}{G}{5}{\beta}\)]
\label{teo:dim24b}
All Automorphic Lie Algebras \(\aliam{V}{G}{5}{\beta}\) are isomorphic and of type  \(A_{4}^{(10,6)}\).
\end{Theorem}
\def\pole{c}
%\input\ifnum\where=1 /lhome/jan/Dropbox/WORK200/\fi ala/Result\dir/A\dynnum/SerreSNormalForm\dirp\pole.tex 
% The normal form  of Ax5c8
% on Mon Feb 23 12:29:12 2015
% file ResultAx5c8/A4/SerreSNormalFormAx5c8.tex
%%%%%%%%%%%%%%%%%%%%%%%%%%%%%%%%%%%%%%%%%%%%%%%%%%%%%%%%%%%
\subsubsection{Pole in $\gamma$}
\label{sec:dim24c}
Let the model for \(\aliam{V}{G}{5}{\gamma}\) be
\[
\|A_4^{(10,8)}\|=\begin{bmatrix}
0 & \mathbbm{I} & \mathbbm{I} & \mathbbm{I} & \mathbbm{I} \\
1 & 0 & \mathbbm{I} & \mathbbm{I} & \mathbbm{I} \\
\mathbbm{J} & \mathbbm{J} & 0 & \mathbbm{I} & \mathbbm{I} \\
\mathbbm{J} & \mathbbm{J} & \mathbbm{J} & 0 & \mathbbm{I} \\
\mathbbm{J} & \mathbbm{J} & \mathbbm{J} & 1& 0
\end{bmatrix}\,,\quad
K_{\mf{b}}(\mf{sl}_{5})_{\gamma}=2\mathbbm{I}+8\mathbbm{I}\mathbbm{J}\,.
%\input\ifnum\where=1 /lhome/jan/Dropbox/WORK200/\fi ala/Result\dir/A\dynnum/KillingNormalForm\dirp\pole.tex
% The normal form  of Ax5c8
% on Mon Feb 23 12:29:12 2015
% in ResultAx5c8/A4/KillingNormalFormAx5c8pc.tex
%K_{\mf{h}}(\mf{sl}_5)=\begin{bmatrix}
%2&-1&0&0\\
%-1&2&-1&0\\
%0&-1&2&-1\\
%0&0&-1&2\end{bmatrix}\,,
%\]
%\[
%K_{\mf{b}}(\mf{sl}_{5})_{\gamma}=\begin{bmatrix}
%\mathbbm{I},
%&\mathbbm{I}\mathbbm{J},
%&\mathbbm{I}\mathbbm{J},
%&\mathbbm{I}\mathbbm{J},
%&\mathbbm{I}\mathbbm{J},
%&\mathbbm{I}\mathbbm{J},
%&\mathbbm{I},
%&\mathbbm{I}\mathbbm{J},
%&\mathbbm{I}\mathbbm{J},
%&\mathbbm{I}\mathbbm{J}\end{bmatrix}\,.
\]
%Let
%\begin{eqnarray*}
%\|A_4^{(10,8)}\|&=&\begin{bmatrix}
%0 & \mathbbm{I} & \mathbbm{I} & \mathbbm{I} & \mathbbm{I} \\
%& 0 & \mathbbm{I} & \mathbbm{I} & \mathbbm{I} \\
%\mathbbm{J} & \mathbbm{J} & 0 & \mathbbm{I} & \mathbbm{I} \\
%\mathbbm{J} & \mathbbm{J} & \mathbbm{J} & 0 & \mathbbm{I} \\
%\mathbbm{J} & \mathbbm{J} & \mathbbm{J} & & 0
%\end{bmatrix}
%\end{eqnarray*}
%,\quad
%%\input\ifnum\where=1 /lhome/jan/Dropbox/WORK200/\fi ala/Result\dir/A\dynnum/KillingNormalForm\dirp\pole.tex
%% The normal form  of Ax5c8
%% on Mon Feb 23 12:29:12 2015
%% in ResultAx5c8/A4/KillingNormalFormAx5c8pc.tex
%\begin{eqnarray*}
%K_{\mf{h}}(\mf{sl}_5)&=&\begin{bmatrix}
%2&-1&0&0\\
%-1&2&-1&0\\
%0&-1&2&-1\\
%0&0&-1&2\end{bmatrix}
%\end{eqnarray*}
%\begin{eqnarray*}
%K_{\mf{b}}(\mf{sl}_{5})_{\gamma}&=&\begin{bmatrix}
%\mathbbm{I},
%&\mathbbm{I}\mathbbm{J},
%&\mathbbm{I}\mathbbm{J},
%&\mathbbm{I}\mathbbm{J},
%&\mathbbm{I}\mathbbm{J},
%&\mathbbm{I}\mathbbm{J},
%&\mathbbm{I},
%&\mathbbm{I}\mathbbm{J},
%&\mathbbm{I}\mathbbm{J},
%&\mathbbm{I}\mathbbm{J}\end{bmatrix}
%\end{eqnarray*}

%\subsubsection{Chevalley model for \(\aliam{V}{G}{5}{\gamma}\)}
\begin{Theorem}[\(\aliam{V}{G}{5}{\gamma}\)]
\label{teo:dim24c}
All Automorphic Lie Algebras \(\aliam{V}{G}{5}{\gamma}\) are isomorphic and of 
type  \(A_{4}^{(10,8)}\).
\end{Theorem}
%%%%%%%%%%%%%%%%%%%%%%%%%%%%%%%%%%%%%%%%%%%%%%%%%%%%%%%%%%%%%%%%%%%%%%%%%%%%%
\subsection{Automorphic Lie Algebras \(\aliam{V}{G}{6}{\zeta}\) }
\label{sec:dim35}
\def\pole{a}
\def\dynnum{5}
\def\dir{Ax5c9}
\def\dirp{Ax5c9p}
%\input\ifnum\where=1 /lhome/jan/Dropbox/WORK200/\fi ala/Result\dir/A\dynnum/SerreSNormalForm\dirp\pole.tex 
%\input\ifnum\where=1 /lhome/jan/Dropbox/WORK200/\fi ala/Result\dir/A\dynnum/KillingNormalForm\dirp\pole.tex
% The normal form  of Ax5c9
% on Thu Feb 12 19:25:32 2015
% file ResultAx5c9/A5/SerreSNormalFormAx5c9.tex
\subsubsection{Poles in $\alpha$}
\label{sec:dim35a}
Let the model for \(\aliam{V}{G}{6}{\alpha}\) be
\[
\|A_5^{(12,9)}\|=\begin{bmatrix}
0 & 1& \mathbbm{I} & \mathbbm{I} & \mathbbm{I} & \mathbbm{I} \\
1& 0 & \mathbbm{I} & \mathbbm{I} & \mathbbm{I} & \mathbbm{I} \\
1& 1 & 0 & 1& \mathbbm{I} & \mathbbm{I} \\ 
\mathbbm{J} & \mathbbm{J} & \mathbbm{J} & 0 & \mathbbm{I} & \mathbbm{I} \\
\mathbbm{J} & \mathbbm{J} & \mathbbm{J} & 1& 0 & 1 \\
\mathbbm{J} & \mathbbm{J} & \mathbbm{J} & 1& 1& 0
\end{bmatrix}\,,\quad
K_{\mf{b}}(\mf{sl}_{6})_{\alpha}=2+4\mathbbm{I}+\mathbbm{J}+8\mathbbm{I}\mathbbm{J}\,.
\]
%,\quad
%\subsubsection{Chevalley model for \(\aliam{V}{G}{6}{\alpha}\)}
\begin{Theorem}[\(\aliam{V}{G}{6}{\alpha}\)]
\label{teo:dim35a}
All Automorphic Lie Algebras \(\aliam{V}{G}{6}{\alpha}\) are isomorphic and of 
type  \(A_{5}^{(12,9)}\).
\end{Theorem}
\def\pole{b}
%\input\ifnum\where=1 /lhome/jan/Dropbox/WORK200/\fi ala/Result\dir/A\dynnum/SerreSNormalForm\dirp\pole.tex 
% The normal form  of Ax5c9
% on Mon Feb  9 10:07:38 2015
% file /lhome/jan/Dropbox/WORK200/ala/ResultAx5c9/A5/SerreSNormalFormAx5c9.tex
\subsubsection{Poles in $\beta$}
\label{sec:dim5b}
Let the model for \(\aliam{V}{G}{6}{\beta}\) be
\[
\|A_5^{(14,9)}\|=\begin{bmatrix}
0 & 1 & \mathbbm{I} & \mathbbm{I} & \mathbbm{I} & \mathbbm{I} \\
1 & 0 & \mathbbm{I} & \mathbbm{I} & \mathbbm{I} & \mathbbm{I} \\
1 & 1 & 0 & \mathbbm{I} & \mathbbm{I} & \mathbbm{I} \\
\mathbbm{J} & \mathbbm{J} & \mathbbm{J} & 0 & \mathbbm{I} & \mathbbm{I} \\
\mathbbm{J} & \mathbbm{J} & \mathbbm{J} & 1 & 0 & \mathbbm{I} \\
\mathbbm{J} & \mathbbm{J} & \mathbbm{J} & 1 & 1 & 0
\end{bmatrix}\,,\quad
K_{\mf{b}}(\mf{sl}_{6})_{\beta}=1+5\mathbbm{I}+9\mathbbm{I}\mathbbm{J}\,.
\]

%,\quad
%\subsubsection{Chevalley model for \(\aliam{V}{G}{6}{\beta}\)}
\begin{Theorem}[\(\aliam{V}{G}{6}{\beta}\)]
\label{teo:dim35b}
All Automorphic Lie Algebras \(\aliam{V}{G}{6}{\beta}\) are isomorphic and of 
type  \(A_{5}^{(14,9)}\).
\end{Theorem}
\def\pole{c}
%\input\ifnum\where=1 /lhome/jan/Dropbox/WORK200/\fi ala/Result\dir/A\dynnum/SerreSNormalForm\dirp\pole.tex 
% The normal form  of Ax5c9
% on Mon Feb  9 12:06:03 2015
% file /lhome/jan/Dropbox/WORK200/ala/ResultAx5c9/A5/SerreSNormalFormAx5c9.tex
\subsubsection{Poles in $\gamma$}
\label{sec:dim35c}
Let the model for \(\aliam{V}{G}{6}{\gamma}\) be
\[
\|A_5^{(14,12)}\|=\begin{bmatrix}
0 & 1 & \mathbbm{I} & \mathbbm{I} & \mathbbm{I} & \mathbbm{I} \\
1 & 0 & \mathbbm{I} & \mathbbm{I} & \mathbbm{I} & \mathbbm{I} \\
\mathbbm{J} & \mathbbm{J} & 0 & \mathbbm{I} & \mathbbm{I} & \mathbbm{I} \\
\mathbbm{J} & \mathbbm{J} & 1 & 0 & \mathbbm{I} & \mathbbm{I} \\
\mathbbm{J} & \mathbbm{J} & \mathbbm{J} & \mathbbm{J} & 0 & \mathbbm{I} \\
\mathbbm{J} & \mathbbm{J} & \mathbbm{J} & \mathbbm{J} & 1 & 0
\end{bmatrix}\,,\quad
K_{\mf{b}}(\mf{sl}_{6})_{\gamma}=1+2\mathbbm{I}+12\mathbbm{I}\mathbbm{J}\,.
\]
%,\quad
%\subsubsection{Chevalley model for \(\aliam{V}{G}{6}{\gamma}\)}
\begin{Theorem}[\(\aliam{V}{G}{6}{\gamma}\) ]\label{teo:dim35c}
All Automorphic Lie Algebras \(\aliam{V}{G}{6}{\gamma}\) are isomorphic and of 
type  \(A_{5}^{(14,12)}\).
\end{Theorem}
We have now proved Theorem \ref{teo:AL} modulo the proofs in Appendix \ref{app:Chev}.
%INSERT NORMAL FORM HERE - IT SEEMS MISSING - REMOVE THEOREM BELOW | DONE BY JAN [07.07]
%%%%%%%%%%%%%%%%%%%%%%%%%%%%%%%%%%%%%%%%%%%%%%%%%%%%%%%%%%%%%%%%%%%%%%%%%%%%%
\section{Invariants of Automorphic Lie Algebras}
\label{sec:Invariants}

In this section we consider \emph{invariants} of Automorphic Lie
Algebras $\cite{knibbeler2014}$. These are defined as properties of
Automorphic Lie Algebras $\alias{\mf{g}(V)}{\zeta}{3}{G}$
that are independent of the
particular reduction group $G$ and its representation $V$. That is,
properties which only depend on the base Lie algebra and the orbit of
poles. The isomorphism question asks whether the Lie algebra structure
is an invariant, and this paper affirms this for
$\mf{g}=\mf{sl}$, cf.~Theorem \ref{teo:AL}.

We saw already in Section \ref{sec:inv_mat_a} that the number of
generators is an invariant, related to the dimension of the underlying
vector space $V$. We will give here two more invariants, namely the
number of $\AIJ[I][\zeta'][\zeta]$s and $\AIJ[J][\zeta''][\zeta]$s in
the Chevalley model, $\zeta,\, \zeta',\, \zeta'' = \alpha,\, \beta$ or $\gamma$.

Let $E_{i,j}$ be the elementary matrix with entry equal to $1$ at the
$i,j$ positions, and zero elsewhere; since the \(H_i\) are by
construction of the type \(E_{i,i}-E_{i+1,i+1}\), the matrices \(M_{\pm\alpha_j }\) will be elementary with coefficients
in \(\splitk[\AIJ[I][\zeta'][\zeta]]\).
We find that the coefficients are always one of four types:
\(1,\, \AIJ[I][\zeta'][\zeta],\,\AIJ[J][\zeta''][\zeta]\) or
\(\AIJ[I][\zeta'][\zeta]\AIJ[J][\zeta''][\zeta]\).
We also find that the number of \(\AIJ[I][\zeta'][\zeta]\)s and
\(\AIJ[J][\zeta''][\zeta]\)s
is determined by the dimension of $\mf{sl}(V)$ and choice of $\zeta$ (see Table
\ref{tab:IJnumbers}) and consequently independent of the group.

\begin{table}[H]
\begin{center}
\begin{tabular}{|c|c|c|c|c|c|}
\hline
$\dim\,\mf{sl}(V)$&3&8&15&24&35\\
\hline
\hline
$\alpha$&(1,1)&(3,2)&(5,4)&(8,6)&(12,9)\\
$\beta$&(1,1)&(3,2)&(6,4)&(10,6)&(14,9)\\
$\gamma$&(1,1)&(3,3)&(6,5)&(10,8)&(14,12)\\
\hline
\end{tabular}
\end{center}
\caption{Numbers \(\left(\#\AIJ[I][\zeta'][\zeta],\# \AIJ[J][\zeta''][\zeta]\right)\) in the Chevalley model, $\zeta=\alpha,\,\beta$ or $\gamma$.}
\label{tab:IJnumbers}
\end{table}
\iffalse
\begin{table}[H]
\begin{center}
\begin{tabular}{|c|c|c|}
\hline
$\dim$ of $\mf{sl}(V)$&$\#$ of $\AIJ[I][\zeta'][\zeta]$&$\#$ of
$\AIJ[J][\zeta'][\zeta]$\\
\hline
3& 1&1\\
8&3&2\\
15&5&4\\
24&8&6\\
35&12&9\\
\hline
\end{tabular}
\end{center}
\caption{Number of \(\AIJ[I][][]\)s and \(\AIJ[J][][]\)s in the Chevalley model.}
\label{tab:IJnumbers}
\end{table}

The numbers in the following table
can be checked against the Chevalley model by considering the number
of \(\AIJ[I][]\)s and \(\AIJ[J][]\)s in \(\langle
M_{-\alpha},M_\alpha\rangle, \alpha\in \Phi^+\).

\begin{table}[H]
\begin{center}
\begin{tabular}{|c|c|c|}
\hline
$\dim$ of $\mf{so}(V)$&$\#$ of $\AIJ[I][]$&$\#$ of $\AIJ[J][]$\\
\hline
\hline
3& 1&1\\
6&2&2\\
10&3&3\\
\hline
\end{tabular}
\end{center}
\caption{Number of \(\AIJ[I][]\)s and \(\AIJ[J][]\)s in matrix of
\(\langle M_\alpha, M_{-\alpha}\rangle, \alpha \in \cal{R}\).}
\label{tab:soIJnumbers}
\end{table}
\fi
Computations suggest that the numbers in Table \ref{tab:IJnumbers} are
invariant from the choice of the CSA, from the choice of the group $G$
and its irreducible representation $V$. In \cite{knibbeler2014} this is in
fact shown to be true for general simple Lie algebras $\mf{g}(V)$,
where $V$ is an irreducible $G$-module. Moreover, for all base Lie
algebras the numbers can be easily derived with the formula
\cite{knibbeler2015isotypic}
\[
\kappa_{\zeta'}\equiv \#\; \text{of}\;
\AIJ[I][\zeta'][\zeta]=\nicefrac{1}{2}\,\codim\mf{g}(V)^{\langle
g_{\zeta'} \rangle}\,,
%\quad
%\#\; \text{of}\; \AIJ[J][]=\frac{1}{2}\,\codim\mf{g}(V)^{\langle\gga\rangle}
\]
where $\langle g_{\zeta'} \rangle$ is a stabiliser subgroup of $G$ at
a zero of $\zeta'$.
This formula enables us to extend the table counting the automorphic
functions in the representations for ALiAs to
undiscovered territory. The following table is taken from
\cite{knibbeler2014}, where further details can be found.
% ADD REF TOTHE THESIS [04.07]

%See \cite{knibbeler2014} for details, where also the next

\begin{center}
\begin{table}[H]
\begin{center}
\begin{tabular}{|c|cccccccc|}
\hline
$\mf{g}$&$\mf{sl}_2, \mf{so}_3, \mf{sp}_2$&$\mf{so}_4$&$\mf{sl}_3$&$\mf{so}_5 ,\mf{sp}_4$&$\mf{sl}_4$&$\mf{sp}_6$&$\mf{sl}_5$&$\mf{sl}_6 $\\
\hline
\hline
$\Phi$&$A_1$&$A_1\oplus A_1$&$A_2$&$B_2,C_2$&$A_3$&$C_3$&$A_4$&$A_5 $\\
\hline
$\coa$&$1$&$2$&$3$&$4$&$6$&$8$&$10$&$14$\\
$\cob$&$1$&$2$&$3$&$3$&$5$&$7$&$8$&$12$\\
$\coc$&$1$&$2$&$2$&$3$&$4$&$6$&$6$&$9$\\
\hline
$\dim\mf{g}%=\Sigma_{\alpha,\beta,\gamma}
$&$3$&$6$&$8$&$10$&$15$&$21$&$24$&$35$\\
\hline
\end{tabular}
\end{center}
%\caption{Half codimensions of single element invariants
%$\nicefrac{1}{2}\,\codim \mf{g}(V)^{\langle g_\zeta\rangle}$,
%$\zeta=\alpha,\beta,\gamma$.}
\caption{Number of automorphic functions in the Chevalley model:
$\kappa_{\zeta'}$,
%=\#\; \text{of}\;
%\AIJ[I][\zeta'][\zeta]=\nicefrac{1}{2}\,\codim\mf{g}(V)^{\langle
%g_{\zeta'} \rangle},
$\zeta'=\alpha,\beta,\gamma$.}
\label{tab:codimg}
\end{table}
\end{center}

This table extends Table \ref{tab:IJnumbers} as follows: the pair in the $\zeta$ row in Table \ref{tab:IJnumbers}, consists of $\kappa_{\zeta'}$ and $\kappa_{\zeta''}$ as found in Table \ref{tab:codimg}, where $\{\zeta,\zeta', \zeta''\}=\{\alpha,\beta,\gamma\}$.
Table \ref{tab:codimg} provides
predictions for the orthogonal and symplectic Lie algebras, which have been verified.

The fact that $\dim\mf{g}=\sum_{\zeta\in\{\alpha,\beta,\gamma\}}\nicefrac{1}{2}\,\codim
\mf{g}(V)^{\langle g_\zeta\rangle}$ is also stated in
\cite{lusztig2003homomorphisms} for the case $G=\mc{A}_5$, the
alternating group and attributed to Serre.
An algebraic proof is given in \cite{knibbeler2014}.

We conclude this section observing that the polynomial $K_{\mf{b}}(\mf{sl}_{n})_\zeta$ carry the information from Table \ref{tab:IJnumbers} and actually add extra information on how the $\AIJ[I][\zeta'][\zeta]$s and $\AIJ[J][\zeta''][\zeta]$s are distributed. Computational evidence suggests that these polynomials are also invariants of the ALiAs.

%%%%%%%%%%%%%%%%%%%%%%%%%%%%%%%%%%%%%%%%%%%%%%%%%%%%%%%%%%%%%%%%%%%%%%%%%%%%%%%%%%%%%%%%%%%%%%%%%%%%%%%%%%%%%%%%%%%%%%%%%%%%%%%%%%%%%%%%%%%%%%%%%%%%%%%%%%%%
\section{Conclusions}\label{sec:conclusions}
The paper addresses the problem of classification for Automorphic Lie Algebras $(\mf{g}\otimes\mc{M}(\overline{\mathbbm{C}}))^G_\Gamma$
where the symmetry group $G$ is finite and the orbit $\Gamma$ is any of the exceptional $G$-orbits in $\overline{\bbbc}$. It presents a complete classification for the case $\mf{sl}_n(\bbbc)$ and proposes a procedure which can be applied to any semi-simple Lie algebra $\mf{g}$, thus it is universal. The analysis makes use of notions from classical invariant theory, such as group forms, Molien series and transvectants, and combines the completely classical representation theory of finite groups with the slightly more modern Lie algebra theory over a polynomial ring. It is worth stressing that it is precisely the combination of these two subjects that poses the central questions in this study and makes the subject interesting and worth studying.\\
The \emph{procedure}, loosely speaking, comprises three steps: the first step consists in identifying the Riemann sphere with the complex projective line $\CCP^1$ consisting of quotients $\nicefrac{X}{Y}$ of two complex variables by setting $\lambda=\nicefrac{X}{Y}$ (Section \ref{sec:computing}). M\"obius transformations on $\lambda$ then correspond to linear transformations on the vector $(X,Y)$ by the same matrix. Classical invariant theory is then used to find the $G$-invariant subspaces of $\CC[X,Y]$-modules, where $\CC[X,Y]$ is the ring of polynomials in $X$ and $Y$. Step two consists in localising these ring-modules of invariants by a choice of multiplicative set of invariants. This choice corresponds to selecting a $G$-orbit $\Gamma_\zeta$ of poles, or equivalently, selecting a relative invariant \(\zeta\) vanishing at those points. The set of elements in the localisation of degree zero, i.e.~the set of elements which can be expressed as functions of $\lambda$, generate the ALiA (Section \ref{sec:comp_inv_mat}). Step one and two can be generalised to any Lie algebra $\mf{g}$, as they rely purely on $\mf{g}(V)$ being a vector space. Once the algebra is computed, it is transformed in the third step into a Chevalley normal form in the spirit of the standard Cartan-Weyl 
%Chevalley 
basis (Section \ref{sec:Chev}). This final step relies on the algebraic structure of $\mf{g}(V)$ and it can be extended to any semi-simple Lie algebra $\mf{g}$.

Through computational means, inspired  be the theory of semi-simple Lie algebras,
%underpinned by theoretical background, 
we demonstrated the existence of a Chevalley normal form for Automorphic Lie Algebras, generalising this classical notion to the case of Lie algebras over a polynomial ring. Moreover, we show that ALiAs associated to $\bbbt\bbbo\bbbi$ groups (namely, tetrahedral, octahedral and icosahedral groups) depend on the group through the automorphic functions only, thus they are group independent as Lie algebras. We prove furthermore that $(\mf{sl}\otimes\mc{M}(\overline{\mathbbm{C}}))^G_{\Gamma_{\zeta}}$ and $(\mf{sl}'\otimes\mc{M}(\overline{\mathbbm{C}}))^{G'}_{\Gamma'_{\zeta'}}$ are isomorphic as Lie algebras if and only if $\kappa_\zeta=\kappa_{\zeta'}$ (Theorem \ref{teo:AL}), and we conjecture a similar result for the cases $\mf{so}$ and $\mf{sp}$. This surprising uniformity of ALiAs is not yet completely understood. The study of ALiAs over finite fields could provide information on whether the uniformity is an algebraic or geometric phenomenon.

We also introduce the concept of \emph{matrices of invariants} (see Section \ref{sec:moi}); they describe the (multiplicative) action of invariant matrices on invariant vectors. The description of the invariant matrices in terms of this action yields a much simpler representation of the Lie algebra, reducing the computational cost considerably. 
%The description of the invariant matrices in terms of this action yields greatly simplified matrices, while preserving the structure constants of the Lie algebra, hence reducing the computational cost considerably. 
We believe that the introduction of matrices of invariants is a fundamental step in the problem of classification of ALiAs.\\
The Cartan-Weyl basis of the matrices of invariants can be seen as a
$1$-form, with arguments in \(\Phi\), the root system of the original
Lie algebra, and taking values in the abelian group of monomials in $\AIJ[I][][]$ and
$\AIJ[J][][]$. The structure constants of the ALiA are given by taking the coboundary operator 
$\mathsf{d}^1$ of this $1$-form. This leads to a formulation of the
isomorphism problem in terms of the action of $\Aut(\Phi)$ on the closed
$2$-forms.

Along with the rise of interest in Darboux transformations with finite reduction groups \cite{konstantinou2013darboux, mikhailov2014darboux} and applications (e.g.~\cite{degasperis2014darboux}), which suggests wide applications of ALiAs within and beyond integrability theory, 
this work encourages further study of the structure theory of ALiAs and proposes the notion of \emph{invariants} (Section \ref{sec:Invariants}). These invariants are polynomials in the coefficients of the computed
$1$-form that are invariant under $\Aut(\Phi)$ and the addition of trivial
terms. Whether these invariants determine the isomorphism is an open
question. From a more general perspective, the success of the structure theory and root system cohomology in absence of a field promises interesting theoretical developments for Lie algebras over a ring.

The theory of ALiAs gives a natural deformation of classical Lie theory
that might be of interest to physics. In particular, it retains the Cartan matrix, thus preserving the finitely generated character of the classical theory.

%%%%%%%%%%%%%%%%%%%%%%%%%%%%%%%%%%%%%%%%%%%%%%%%%%%%%%%%%%%

\textbf{Acknowledgements}
The result presented here are the culmination of a long standing quest and report on work done over a numbers of years.
S. L. gratefully acknowledges financial support from EPSRC (EP/E044646/1 and EP/E044646/2) and from NWO VENI (016.073.026).

%%%%%%%%%%%%%%%%%%%%%%%%%%%%%%%%%%%%%%%%%%%%%%%%%%%%%%%%%%%%%%%%%%%%%%%%%%%%%%%%%%%%%%%%%%%%%%%%%%%%%%%%%%%%%%%%%%%%%%%%%%%%%%%%%%%%%%%%%%%%%%%%%%%%%%%%%%%%%%%%%%%%%%%%%%%%%%%
\appendix
\section{Projective representations and double covering groups}
\label{app:cover}
Let \(G\)  be a finite group and let \(\sigma\)  be a faithful projective representation of \(G\)  in \(\bbbc^2\), that is, \(\sigma\) is a mapping from \(G\) to \(GL_2(\bbbc)\) obeying the following 
\begin{equation}
\label{eq:pr_a}
\sigma(g)\,\sigma(h)=c(g,h)\,\sigma(g\,h)\,,\quad \forall g\,,h\in G\,.
\end{equation}
%This restricts \(G\)  to one of the following groups \cite{MR0080930,MR1315530}: 
%\begin{equation}\label{eq:list}
%\Zn{M} \,,\quad\bbbd_M \,,\quad \bbbt\, ,\quad \bbbo\, ,\quad \bbbi,
%\end{equation}
%i.e.~the cyclic group $\Zn{M}$, the dihedral group $\bbbd_M$, the tetrahedral group $\bbbt$, the octahedral group $\bbbo$ and the icosahedral group $\bbbi$; we refer to \(\bbbd_M\), \(\bbbt\), \(\bbbo\), \(\bbbi\) as to \emph{Platonic} groups.
Here
\(c(g,h)\,:\; G\times G\to\bbbc^{\ast}\) in (\ref{eq:pr_a}) is a nontrivial 2-cocycle over \(\bbbc^{\ast}\), the multiplicative group of \(\bbbc\) (see for example \cite{MR1157729}), satisfying the cocycle identity 
\[
c(x,y)c(xy,z)=c(y,z)c(x,yz).
\]
%and taking its values in the abelian group \(H^2(G,\bbbz)\).
% which is known to be \(S_2\) in all three cases.
It follows from the cocycle condition that \(c(1,1)=c(1,z)\) and \(c(x,1)=c(1,1)\).
So if one defines \(\tilde{c}(x,y)=c(x,y)c(1,1)^{-1}\),
then \(\tilde{c}\) is again a cocycle, but now with \(\tilde{c}(x,1)\) and \(\tilde{c}(1,x)\) equal to \(1\). It follows that \(c(x,y)\) is a root of unity, the order of which divides the group order. If the cocycle is trivial one can view the projective representation as a representation.

For each of the Platonic groups \(\bbbt, \bbbo\) and \(\bbbi\) consider a projective representation \(\sigma\).
In order to use GAP to compute generating elements, character tables and Molien functions,
we need to replace the projective representation by a representation.
The time-honored method to do this is by constructing the covering group \(G^\flat\),
which is an extension of the group with its second cohomology group: the sequence
\[
0\rightarrow H^2(G,\bbbz) \rightarrow G^\flat \rightarrow G \rightarrow 0
\]
is exact. The actual construction runs as follows.
One defines (with trivial group action) the group cohomology with values in \(\bbbz\) as follows
(written in the usual additive way, followed by multiplication as in the definition of the projective representation):
\begin{eqnarray*}
&&\mathsf{d}^0 a(x)=a-a=0\equiv 1\\
&&\mathsf{d}^1 b(x,y)=b(xy)-b(x)-b(y)\equiv \frac{b(xy)}{b(x)b(y)}\\
&&\mathsf{d}^2 c(x,y,z)=c(y,z)-c(xy,z)+c(x,yz)-c(x,y)\equiv \frac{c(y,z)c(x,yz)}{c(xy,z)c(x,y)}
%\\d^i g(x_0,\cdots,x_i)&=&(-1)^ i g(x_1,\cdots,x_i)+\sum_{j=1}^i (-1)^{i+j} g(x_0,\cdots,x_{j-1},x_j,\cdots,x_i)-g(x_0,\cdots,x_{i-1}).
%\\&=& \left\{ \begin{array}{lr} \frac{ g(x_0 x_1,\cdots,x_i)}{g(x_1,\cdots,x_i)}\frac{g(x_0,x_1 x_2,\cdots, x_i)}{g(x_0,x_1,x_2x_3,\cdots,x_i)}\cdots \frac{g(x_0,\cdots,x_{i-2}x_{i-1},x_i)}{g(x_0,\cdots,x_{i-1}x_i)} \frac{1}{g(x_0,\cdots,x_{i-1})}&\text{i odd}\\
%\frac{g(x_1,\cdots,x_i)}{ g(x_0 x_1,\cdots,x_i)}\frac{g(x_0,x_1,x_2x_3,\cdots,x_i)}{g(x_0,x_1 x_2,\cdots, x_i)}\cdots\frac{g(x_0,\cdots,x_{i-3}x_{i-2},x_{i-1},x_i)}{g(x_0,\cdots,x_{i-2}x_{i-1},x_i)} \frac{g(x_0,\cdots,x_{i-1}x_i)}{g(x_0,\cdots,x_{i-1})}&\text{i even}
%\end{array}\right.
 \end{eqnarray*}
\iffalse
We require \(c(x,y+z)\geq c(x,y)\geq c(x+z,y)\).
If \(c=d^1b\) this reduces to
\( b(x+y+z)+b(y)+b(x+z)\geq b(x+y)+b(x+z)+b(y+z)\geq b(x+y+z) +b(y+z)+b(x)\).
For a form with an odd number of arguments we require
\[ g(x_0,\cdots,x_{i-j-2}x_{i-j-1},\cdots,x_i)\geq  g(x_0,\cdots,x_{i-j-1}x_{i-j},\cdots,x_i).\]
For a form with an even number of arguments we require
\[ g(x_0,\cdots,x_{i-j-2}x_{i-j-1},\cdots,x_i)\leq  g(x_0,\cdots,x_{i-j-1}x_{i-j},\cdots,x_i)\] and, moreover,
\[g(x_0,\cdots,x_{i-2}x_{i-1},x_i)\geq g(x_0,\cdots,x_{i-1}x_i)+g(x_0,\cdots,x_{i-1}).\]
\fi
Then the second cohomology group \(H^2(G,\bbbz)\) is defined as the quotient of \(\ker \mathsf{d}^2\) over \(\mathrm{im\ } \mathsf{d}^1\),
which is well defined since \(\mathsf{d}^2 \mathsf{d}^1\) maps to unity.
%The projective representation is defined using a nontrivial cocycle, satisfying the cocycle identity \[c(x,y)c(xy,z)=c(y,z)c(x,yz),\]
%and taking its values in the abelian group
%\(H^2(G,\bbbz)\) which is known to be \(S_2\) in all three cases.
%It follows from the cocycle condition that \(c(1,1)=c(1,z)\) and \(c(x,1)=c(1,1)\).
%So if one defines \(\tilde{c}(x,y)=c(x,y)c(1,1)^{-1}\),
%then \(\tilde{c}\) is again a cocycle, but now with \(\tilde{c}(x,1)\) and \(\tilde{c}(1,x)\) equal to \(1\).
%We use the \(\tilde{c}\) in what follows, but call it \(c\) again.
%In order to apply the usual representation theory one does the following.
%Let \(E=id_{\bbbc^2}, R=\pi(r)\) and \(S=\pi(s)\) and consider \(E, R\) and \(S\) as the generators
%of a group, which we denote by \(G^\flat\).
%For instance, we compute the product of \(R\) and \(S\) as follows:
%\[
%R S= \pi(r)\pi(s)=c(r,s) \pi(rs)
%\]
%By the definition of cocycle this defines an associative multiplication.
We can consider \(G^\flat\) as the group generated by the pairs
\( (r,\rho)\), with \(r\in G\) and \(\rho\in H^2(G,\bbbz)=\bbbz/2=\langle \pm 1 \rangle\) \cite{Schur1904, Schur1911},
with multiplication given by
\[
(x,\xi)(y,\upsilon)=(xy,\xi\upsilon \tilde{c}(x,y)).
\]
Then the identity is \((e,1)\), since \(\tilde{c}(x,1)\) and \(\tilde{c}(1,x)\) are both equal to \(1\).
Let us check associativity (and see what motivated the cocycle identity):
\begin{eqnarray*}
((x,\xi)(y,\upsilon))(z,\zeta)&=&(xy,\xi\upsilon \tilde{c}(x,y))(z,\zeta)
\\&=&((xy)z,\xi\upsilon \tilde{c}(x,y)\zeta \tilde{c}(xy,z))
\\&=&(x(yz),\xi\upsilon \zeta \tilde{c}(y,z)\tilde{c}(x,yz))
\\&=&(x,\xi)(yz,\upsilon \zeta  \tilde{c}(y,z) )
\\&=&(x,\xi)((y,\upsilon)(z,\zeta)).
\end{eqnarray*}
One defines the inverse of an element by
\[
(x,\xi)^{-1}=(x^{-1},\xi^{-1} \tilde{c}(x,x^{-1})^{-1}).
\]
On \(G^\flat\) we now define a representation \(\sigma^\flat((x,\xi))=\xi c(1,1)^{-1}\sigma(x)\).
We have  indeed
\begin{eqnarray*}
\sigma^\flat((x,\xi))\sigma^\flat((y,\upsilon))&=&c(1,1)^{-2}\xi\upsilon\sigma(x)\sigma(y)
\\&=&
c(1,1)^{-2}\xi\upsilon c(x,y)\sigma(xy)\\&=&\sigma^\flat((xy,c(1,1)^{-1}\xi\upsilon c(x,y)))\\&=&\sigma^\flat((xy,\xi\upsilon \tilde{c}(x,y)))=
\sigma^\flat((x,\xi)(y,\upsilon)).
\end{eqnarray*}
In practice one can compute the cocycle the other way around, by considering given
\(\sigma(r)\) and \(\sigma(s)\) as generators of \(G^\flat\) and computing the group multiplication table.
\begin{Remark}
Suppose there exists a section \(s:G\rightarrow G^\flat\).
This would imply the existence of an element \(\zeta\in C^1(G,\bbbz)\), 
such that \(s(g)=(g,\zeta(g))\).
Can we do this so that \(s(gh)=s(g)s(h)\)? In that case \(G\) can be viewed as a subgroup of \(G^\flat\)).
This would imply
\begin{eqnarray*}
s(gh)&=& (gh,\zeta(gh)) \\
s(g)s(h)&=& (g,\zeta(g))(h,\zeta(h))=(gh,\zeta(g)\zeta(h)c(g,h))
\end{eqnarray*}
But this would in turn imply that \(c=\mathsf{d}^1 \zeta\) is a coboundary, where in fact the assumption was that \(c\) was nontrivial.
\end{Remark}
%%%%%%%%%%%%%%%%%%%%%%%%%%%%%%%%%%%%%%%%%%%%%%%%%%%%%%%%%%%%%%%%%%%%%%%%%%%%%%%%%%%%%%%%%%%%%%%%%%%%%%%%%%%%%%%%%%%%%%%%%%%%%%%%%%%%%%%%%%%%%%%%%%%%%%%%%%%%%%%%%%%%%%%%%%%%%%%%%%%%%%%%%%%%%%%%%%%%%%%%%%%%%%%%%%%%%%%%%%%%%%%%%%%%%%%%%%%%%

%\section{Chevalley normal form of all cases}
\section{Chevalley normal forms}
\label{app:Chev}
\setcounter{table}{0}
\renewcommand{\thetable}{B\arabic{table}}
%%%%%%%%%%%%%%%%%%%%%%%%%%%%%%%%%%%%%%%%%%%%%%%%%%%%%%%%%%%%%%%%%%%%
\def\dynnum{2}
\def\dym{3}
\begin{Theorem}[\(\aliam{V}{G}{\dym}{\zeta}\), \(\zeta=\alpha,\beta\)]
\label{teo:dim8app_ab}
All Automorphic Lie Algebras \(\aliam{V}{G}{\dym}{\zeta}\), \(\zeta=\alpha,\beta\),  are of type \(A_\dynnum^{(3,2)}\) and therefore isomorphic.
\end{Theorem}
\begin{proof}
\def\dynnum{2}
We give the Chevalley model together with its intertwining operator \({\cal{I}}_{\mf{sl}(V)}\) with respect to
\(\|A_\dynnum^{(3,2)}\|\) (see Tables B1 and B2),
% (see Tables \ref{tab:dim8app_a} and \ref{tab:dim8app_b}),
i.e.
\[
\|\mf{sl}(V)\| {\cal{I}}_{\mf{sl}(V)}={\cal{I}}_{\mf{sl}(V)} \|A_{\dynnum}^{(3,2)}\|.
\]

\begin{table}[h!]
\label{tab:dim8app_a}
\begin{center}
\scalebox{0.86}{
\begin{tabular}{c|cccc} 
\hline
Irreducible representation\\
$V$
&
$\bbbt_{7}$, $\bbbi_{5}$
&$\bbbo_{6}$
&$\bbbo_{7}$
&$\bbbi_{4}$
\\
& & & & \\
\hline\hline
Chevalley model \\
$\|\mf{sl}(V)\|$ &
\def\pole{a}
\def\dynnum{2}
\def\dir{Ax3c7}
\def\dirp{Ax3c7p\pole}
$\begin{bmatrix}
0&\AIJ[J][\gamma][\alpha]&\AIJ[I][\beta][\alpha]\\
\AIJ[I][\beta][\alpha]&0&\AIJ[I][\beta][\alpha]\\
1&\AIJ[J][\gamma][\alpha]&0\end{bmatrix}$
&
\def\pole{a}
\def\dynnum{2}
\def\dir{Ax4c6}
\def\dirp{Ax4c6p\pole}
$\begin{bmatrix}
0&1&\AIJ[J][\gamma][\alpha]\\
\AIJ[I][\beta][\alpha]&0&\AIJ[I][\beta][\alpha]\AIJ[J][\gamma][\alpha]\\
\AIJ[I][\beta][\alpha]&1&0\end{bmatrix}$
&
\def\pole{a}
\def\dynnum{2}
\def\dir{Ax4c7}
\def\dirp{Ax4c7p\pole}
$\begin{bmatrix}
0&\AIJ[I][\beta][\alpha]&\AIJ[J][\gamma][\alpha]\\
1&0&\AIJ[J][\gamma][\alpha]\\
\AIJ[I][\beta][\alpha]&\AIJ[I][\beta][\alpha]&0\end{bmatrix}$
&
\def\pole{a}
\def\dynnum{2}
\def\dir{Ax5c4}
\def\dirp{Ax5c4p\pole}
$\begin{bmatrix}
0&\AIJ[I][\beta][\alpha]&\AIJ[I][\beta][\alpha]\AIJ[J][\gamma][\alpha]\\
1&0&\AIJ[I][\beta][\alpha]\AIJ[J][\gamma][\alpha]\\
1&1&0\end{bmatrix}$
\\
& & &  &\\
\hline
Intertwining operator \\
${\cal{I}}_{\mf{sl}(V)}$
& 
\def\pole{a}
\def\dynnum{2}
\def\dir{Ax3c7}
\def\dirp{Ax3c7p\pole}
$\begin{pmatrix}
0&1&0\\
1&0&0\\
0&0&1\end{pmatrix}$
&
\def\pole{a}
\def\dynnum{2}
\def\dir{Ax4c6}
\def\dirp{Ax4c6p\pole}
$\begin{pmatrix}
0&1&0\\
0&0&\AIJ[I][\beta][\alpha]\\
1&0&0\end{pmatrix}$
&
\def\pole{a}
\def\dynnum{2}
\def\dir{Ax4c7}
\def\dirp{Ax4c7p\pole}
$\begin{pmatrix}
0&1&0\\
0&0&1\\
1&0&0\end{pmatrix}$
&
\def\pole{a}
\def\dynnum{2}
\def\dir{Ax5c4}
\def\dirp{Ax5c4p\pole}
$\begin{pmatrix}
0&\AIJ[I][\beta][\alpha]&0\\
0&0&\AIJ[I][\beta][\alpha]\\
1&0&0\end{pmatrix}$
\\
& & &  &\\
\hline
\end{tabular}
}
\end{center}
\caption{Chevalley models and intertwining operators for \(\aliam{V}{G}{\dym}{\alpha}\).}
\end{table}
%%%%%%%%%%%%%%%%%%%%%%%%%%%%%%%%%%%%%%%%%%%%%%%%%%%%%%%%%%%%%%%%%%%%
\begin{table}[H]
\label{tab:dim8app_b}
\begin{center}
\scalebox{0.75}{
\begin{tabular}{c|ccccc} 
\hline
Irrep\\
$V$
&
$\bbbt_{7}$
&$\bbbo_{6}$
&$\bbbo_{7}$
&$\bbbi_{4}$
&$\bbbi_{5}$
\\
& & & & \\
\hline\hline
Chevalley model \\
$\|\mf{sl}(V)\|$ &
\def\pole{b}
\def\dynnum{2}
\def\dir{Ax3c7}
\def\dirp{Ax3c7p\pole}
$\begin{bmatrix}
0&\AIJ[I][\alpha][\beta]&\AIJ[I][\alpha][\beta]\\
\AIJ[J][\gamma][\beta]&0&\AIJ[I][\alpha][\beta]\\
\AIJ[J][\gamma][\beta]&1&0\end{bmatrix}$
&
\def\pole{b}
\def\dynnum{2}
\def\dir{Ax4c6}
\def\dirp{Ax4c6p\pole}
$\begin{bmatrix}
0&\AIJ[J][\gamma][\beta]&\AIJ[I][\alpha][\beta]\\
\AIJ[I][\alpha][\beta]&0&\AIJ[I][\alpha][\beta]\\
1&\AIJ[J][\gamma][\beta]&0\end{bmatrix}$
&
\def\pole{b}
\def\dynnum{2}
\def\dir{Ax4c7}
\def\dirp{Ax4c7p\pole}
$\begin{bmatrix}
0&1&\AIJ[J][\gamma][\beta]\\
\AIJ[I][\alpha][\beta]&0&\AIJ[I][\alpha][\beta]\AIJ[J][\gamma][\beta]\\
\AIJ[I][\alpha][\beta]&1&0\end{bmatrix}$
&
\def\pole{b}
\def\dynnum{2}
\def\dir{Ax5c4}
\def\dirp{Ax5c4p\pole}
$\begin{bmatrix}
0&1&\AIJ[I][\alpha][\beta]\AIJ[J][\gamma][\beta]\\
\AIJ[I][\alpha][\beta]&0&\AIJ[I][\alpha][\beta]\AIJ[J][\gamma][\beta]\\
1&1&0\end{bmatrix}$
&
\def\pole{b}
\def\dynnum{2}
\def\dir{Ax5c5}
\def\dirp{Ax5c5p\pole}
$\begin{bmatrix}
0&\AIJ[J][\gamma][\beta]&1\\
\AIJ[I][\alpha][\beta]&0&1\\
\AIJ[I][\alpha][\beta]&\AIJ[I][\alpha][\beta]\AIJ[J][\gamma][\beta]&0\end{bmatrix}$
\\
& & &  &\\
\hline
Intertwining operator \\
${\cal{I}}_{\mf{sl}(V)}$
& 
\def\pole{b}
\def\dynnum{2}
\def\dir{Ax3c7}
\def\dirp{Ax3c7p\pole}
$\begin{pmatrix}
1&0&0\\
0&1&0\\
0&0&1\end{pmatrix}$
&
\def\pole{b}
\def\dynnum{2}
\def\dir{Ax4c6}
\def\dirp{Ax4c6p\pole}
$\begin{pmatrix}
0&1&0\\
1&0&0\\
0&0&1\end{pmatrix}$
&
\def\pole{b}
\def\dynnum{2}
\def\dir{Ax4c7}
\def\dirp{Ax4c7p\pole}
$\begin{pmatrix}
0&1&0\\
0&0&\AIJ[I][\alpha][\beta]\\
1&0&0\end{pmatrix}$
&
\def\pole{b}
\def\dynnum{2}
\def\dir{Ax5c4}
\def\dirp{Ax5c4p\pole}
$\begin{pmatrix}
0&0&\AIJ[I][\alpha][\beta]\\
0&\AIJ[I][\alpha][\beta]&0\\
1&0&0\end{pmatrix}$
&
\def\pole{b}
\def\dynnum{2}
\def\dir{Ax5c5}
\def\dirp{Ax5c5p\pole}
$\begin{pmatrix}
0&1&0\\
1&0&0\\
0&0&\AIJ[I][\alpha][\beta]\end{pmatrix}$
\\
& & & &\\
\hline
\end{tabular}
}
\end{center}
\caption{Chevalley models and intertwining operators for \(\aliam{V}{G}{\dym}{\beta}\).}
\end{table}
\end{proof}

%%%%%%%%%%%%%%%%%%%%%%%%%%%%%%%%%%%%%%%%%%%%%%%%%%%%%%%%%%%%%%%%%%%%
\begin{Theorem}[\(\aliam{V}{G}{\dym}{\gamma}\)]
\label{teo:dim8app_g}
\def\dynnum{2}
\def\dym{3}
All Automorphic Lie Algebras \(\aliam{V}{G}{\dym}{\gamma}\) are of type \(A_\dynnum^{(3,3)}\) and therefore isomorphic.
\end{Theorem}
\begin{proof}
\def\dynnum{2}
We give the Chevalley model together with its intertwining operator \({\cal{I}}_{\mf{sl}(V)}\) with respect to
\(\|A_\dynnum^{(3,3)}\|\) (see Table B3),
%(see Table \ref{tab:dim8app_g}),
i.e.
\[
\|\mf{sl}(V)\| {\cal{I}}_{\mf{sl}(V)}={\cal{I}}_{\mf{sl}(V)} \|A_{\dynnum}^{(3,3)}\|.
\]
\begin{table}[H]
\label{tab:dim8app_g}
\begin{center}
\scalebox{0.85}{
\begin{tabular}{c|ccccc} 
\hline
Irrep\\
$V$
&
$\bbbt_{7}$
&$\bbbo_{6}$, $\bbbi_{5}$
&$\bbbo_{7}$
&$\bbbi_{4}$
\\
& & & & \\
\hline\hline
Chevalley model \\
$\|\mf{sl}(V)\|$ &
\def\pole{c}
\def\dynnum{2}
\def\dir{Ax3c7}
\def\dirp{Ax3c7p\pole}
$\begin{bmatrix}
0&1&1\\
\AIJ[I][\alpha][\gamma]\AIJ[J][\beta][\gamma]&0&\AIJ[J][\beta][\gamma]\\
\AIJ[I][\alpha][\gamma]\AIJ[J][\beta][\gamma]&\AIJ[I][\alpha][\gamma]&0\end{bmatrix}$
&
\def\pole{c}
\def\dynnum{2}
\def\dir{Ax4c6}
\def\dirp{Ax4c6p\pole}
$\begin{bmatrix}
0&\AIJ[J][\beta][\gamma]&\AIJ[J][\beta][\gamma]\\
\AIJ[I][\alpha][\gamma]&0&\AIJ[J][\beta][\gamma]\\
\AIJ[I][\alpha][\gamma]&\AIJ[I][\alpha][\gamma]&0\end{bmatrix}$
&
\def\pole{c}
\def\dynnum{2}
\def\dir{Ax4c7}
\def\dirp{Ax4c7p\pole}
$\begin{bmatrix}
0&\AIJ[J][\beta][\gamma]&\AIJ[I][\alpha][\gamma]\\
\AIJ[I][\alpha][\gamma]&0&\AIJ[I][\alpha][\gamma]\\
\AIJ[J][\beta][\gamma]&\AIJ[J][\beta][\gamma]&0\end{bmatrix}$
&
\def\pole{c}
\def\dynnum{2}
\def\dir{Ax5c4}
\def\dirp{Ax5c4p\pole}
$\begin{bmatrix}
0&\AIJ[I][\alpha][\gamma]&\AIJ[I][\alpha][\gamma]\\
\AIJ[J][\beta][\gamma]&0&\AIJ[J][\beta][\gamma]\\
\AIJ[J][\beta][\gamma]&\AIJ[I][\alpha][\gamma]&0\end{bmatrix}$
\\
& & &  &\\
\hline
Intertwining operator \\
${\cal{I}}_{\mf{sl}(V)}$
& 
\def\pole{c}
\def\dynnum{2}
\def\dir{Ax3c7}
\def\dirp{Ax3c7p\pole}
$\begin{pmatrix}
1&0&0\\
0&0&\AIJ[I][\alpha][\gamma]\\
0&\AIJ[I][\alpha][\gamma]&0\end{pmatrix}$
&
\def\pole{c}
\def\dynnum{2}
\def\dir{Ax4c6}
\def\dirp{Ax4c6p\pole}
$\begin{pmatrix}
\AIJ[J][\beta][\gamma]&0&0\\
0&0&\AIJ[I][\alpha][\gamma]\\
0&\AIJ[I][\alpha][\gamma]&0\end{pmatrix}$
&
\def\pole{c}
\def\dynnum{2}
\def\dir{Ax4c7}
\def\dirp{Ax4c7p\pole}
$\begin{pmatrix}
\AIJ[J][\beta][\gamma]&0&0\\
0&0&\AIJ[I][\alpha][\gamma]\\
0&\AIJ[J][\beta][\gamma]&0\end{pmatrix}$
&
\def\pole{c}
\def\dynnum{2}
\def\dir{Ax5c4}
\def\dirp{Ax5c4p\pole}
$\begin{pmatrix}
1&0&0\\
0&0&1\\
0&1&0\end{pmatrix}$
\\
& & & &\\
\hline
\end{tabular}
}
\end{center}
\caption{Chevalley models and intertwining operators for \(\aliam{V}{G}{\dym}{\gamma}\).}
\end{table}
\end{proof}
%\newpage
%%%%%%%%%%%%%%%%%%%%%%%%%%%%%%%%%%%%%%%%%%%%%%%%%%%%%%%%%%%%%%%%%%%%
\def\dym{4}
\begin{Theorem}[\(\aliam{V}{G}{\dym}{\alpha}\)]
\label{teo:dim15app_a}
\def\dynnum{3}
All Automorphic Lie Algebras \(\aliam{V}{G}{\dym}{\alpha}\) are of type \(A_\dynnum^{(5,4)}\) and therefore isomorphic.
\end{Theorem}
\begin{proof}
\def\dynnum{3}
We give the Chevalley model together with its intertwining operator \({\cal{I}}_{\mf{sl}(V)}\) with respect to
\(\|A_\dynnum^{(5,4)}\|\) (see Table B4),
%(see Table \ref{tab:dim15app_a}),
i.e.
\[
\|\mf{sl}(V)\| {\cal{I}}_{\mf{sl}(V)}={\cal{I}}_{\mf{sl}(V)} \|A_{\dynnum}^{(5,4)}\|.
\]
\begin{table}[H]
\label{tab:dim15app_a}
\begin{center}
\scalebox{0.9}{
\begin{tabular}{c|ccc} 
\hline
Irrep\\
$V$
&
$\bbbo_{8}$
&$\bbbi_{6}$
&$\bbbi_{7}$
\\
& & &  \\
\hline\hline
Chevalley model \\
$\|\mf{sl}(V)\|$ 
&
\def\pole{a}
\def\dynnum{3}
\def\dir{Ax4c8}
\def\dirp{Ax4c8p\pole}
$\begin{bmatrix}
0&\AIJ[I][\beta][\alpha]&1&\AIJ[I][\beta][\alpha]\\
1&0&1&1\\
\AIJ[J][\gamma][\alpha]&\AIJ[I][\beta][\alpha]\AIJ[J][\gamma][\alpha]&0&\AIJ[I][\beta][\alpha]\\
\AIJ[J][\gamma][\alpha]&\AIJ[I][\beta][\alpha]\AIJ[J][\gamma][\alpha]&1&0\end{bmatrix}$
&
\def\pole{a}
\def\dynnum{3}
\def\dir{Ax5c6}
\def\dirp{Ax5c6p\pole}
$\begin{bmatrix}
0&\AIJ[I][\beta][\alpha]&1&1\\
\AIJ[J][\gamma][\alpha]&0&\AIJ[J][\gamma][\alpha]&1\\
\AIJ[I][\beta][\alpha]&\AIJ[I][\beta][\alpha]&0&1\\
\AIJ[I][\beta][\alpha]\AIJ[J][\gamma][\alpha]&\AIJ[I][\beta][\alpha]&\AIJ[J][\gamma][\alpha]&0\end{bmatrix}$
&
\def\pole{a}
\def\dynnum{3}
\def\dir{Ax5c7}
\def\dirp{Ax5c7p\pole}
$\begin{bmatrix}
0&1&1&1\\
\AIJ[I][\beta][\alpha]\AIJ[J][\gamma][\alpha]&0&1&\AIJ[J][\gamma][\alpha]\\
\AIJ[I][\beta][\alpha]\AIJ[J][\gamma][\alpha]&\AIJ[I][\beta][\alpha]&0&\AIJ[J][\gamma][\alpha]\\
\AIJ[I][\beta][\alpha]&\AIJ[I][\beta][\alpha]&1&0\end{bmatrix}$
\\
& & &  \\
\hline
Intertwining operator \\
${\cal{I}}_{\mf{sl}(V)}$
&
\def\pole{a}
\def\dynnum{3}
\def\dir{Ax4c8}
\def\dirp{Ax4c8p\pole}
$\begin{pmatrix}
0&\AIJ[I][\beta][\alpha]&0&0\\
1&0&0&0\\
0&0&\AIJ[I][\beta][\alpha]&0\\
0&0&0&\AIJ[I][\beta][\alpha]\end{pmatrix}$
&
\def\pole{a}
\def\dynnum{3}
\def\dir{Ax5c6}
\def\dirp{Ax5c6p\pole}
$\begin{pmatrix}
0&0&0&\AIJ[I][\beta][\alpha]\\
\AIJ[J][\gamma][\alpha]&0&0&0\\
0&0&\AIJ[I][\beta][\alpha]&0\\
0&\AIJ[I][\beta][\alpha]\AIJ[J][\gamma][\alpha]&0&0\end{pmatrix}$
&
\def\pole{a}
\def\dynnum{3}
\def\dir{Ax5c7}
\def\dirp{Ax5c7p\pole}
$\begin{pmatrix}
1&0&0&0\\
0&0&0&\AIJ[I][\beta][\alpha]\\
0&0&\AIJ[I][\beta][\alpha]&0\\
0&\AIJ[I][\beta][\alpha]&0&0\end{pmatrix}$
\\
& & &  \\
\hline
\end{tabular}
}
\end{center}
\caption{Chevalley models and intertwining operators for \(\aliam{V}{G}{\dym}{\alpha}\).}
\end{table}
\end{proof}

%%%%%%%%%%%%%%%%%%%%%%%%%%%%%%%%%%%%%%%%%%%%%%%%%%%%%%%%%%%%%%%%%%%%
\def\dym{4}
\begin{Theorem}[\(\aliam{V}{G}{\dym}{\beta}\)]
\label{teo:dim15app_b}
\def\dynnum{3}
All Automorphic Lie Algebras \(\aliam{V}{G}{\dym}{\beta}\) are of type \(A_\dynnum^{(6,4)}\) and therefore isomorphic.
\end{Theorem}
\begin{proof}
\def\dynnum{3}
We give the Chevalley model together with its intertwining operator \({\cal{I}}_{\mf{sl}(V)}\) with respect to
\(\|A_\dynnum^{(6,4)}\|\) (see Table B5),
% (see Table \ref{tab:dim15app_b}),
i.e.
\[
\|\mf{sl}(V)\| {\cal{I}}_{\mf{sl}(V)}={\cal{I}}_{\mf{sl}(V)} \|A_{\dynnum}^{(6,4)}\|.
\]
\begin{table}[H]
\label{tab:dim15app_b}
\begin{center}
\scalebox{0.9}{
\begin{tabular}{c|ccc} 
\hline
Irreducible representation\\
$V$
&
$\bbbo_{8}$
&$\bbbi_{6}$
&$\bbbi_{7}$
\\
& & &  \\
\hline\hline
Chevalley model \\
$\|\mf{sl}(V)\|$ 
&
\def\pole{b}
\def\dynnum{3}
\def\dir{Ax4c8}
\def\dirp{Ax4c8p\pole}
$\begin{bmatrix}
0&1&1&1\\
\AIJ[I][\alpha][\beta]\AIJ[J][\gamma][\beta]&0&\AIJ[I][\alpha][\beta]&\AIJ[J][\gamma][\beta]\\
\AIJ[I][\alpha][\beta]\AIJ[J][\gamma][\beta]&1&0&\AIJ[J][\gamma][\beta]\\
\AIJ[I][\alpha][\beta]&\AIJ[I][\alpha][\beta]&\AIJ[I][\alpha][\beta]&0\end{bmatrix}$
&
\def\pole{b}
\def\dynnum{3}
\def\dir{Ax5c6}
\def\dirp{Ax5c6p\pole}
$\begin{bmatrix}
0&\AIJ[J][\gamma][\beta]&1&\AIJ[J][\gamma][\beta]\\
\AIJ[I][\alpha][\beta]&0&\AIJ[I][\alpha][\beta]&1\\
\AIJ[I][\alpha][\beta]&\AIJ[J][\gamma][\beta]&0&\AIJ[J][\gamma][\beta]\\
\AIJ[I][\alpha][\beta]&\AIJ[I][\alpha][\beta]&\AIJ[I][\alpha][\beta]&0\end{bmatrix}$
&
\def\pole{b}
\def\dynnum{3}
\def\dir{Ax5c7}
\def\dirp{Ax5c7p\pole}
$\begin{bmatrix}
0&\AIJ[I][\alpha][\beta]&\AIJ[I][\alpha][\beta]&\AIJ[I][\alpha][\beta]\\
\AIJ[J][\gamma][\beta]&0&\AIJ[I][\alpha][\beta]\AIJ[J][\gamma][\beta]&1\\
1&1&0&1\\
\AIJ[J][\gamma][\beta]&\AIJ[I][\alpha][\beta]&\AIJ[I][\alpha][\beta]\AIJ[J][\gamma][\beta]&0\end{bmatrix}$
\\
& & & \\
\hline
Intertwining operator \\
${\cal{I}}_{\mf{sl}(V)}$
& 
\def\pole{b}
\def\dynnum{3}
\def\dir{Ax4c8}
\def\dirp{Ax4c8p\pole}
$\begin{pmatrix}
1&0&0&0\\
0&0&\AIJ[I][\alpha][\beta]&0\\
0&0&0&\AIJ[I][\alpha][\beta]\\
0&\AIJ[I][\alpha][\beta]&0&0\end{pmatrix}$
&
\def\pole{b}
\def\dynnum{3}
\def\dir{Ax5c6}
\def\dirp{Ax5c6p\pole}
$\begin{pmatrix}
0&\AIJ[J][\gamma][\beta]&0&0\\
0&0&0&\AIJ[I][\alpha][\beta]\\
\AIJ[J][\gamma][\beta]&0&0&0\\
0&0&\AIJ[I][\alpha][\beta]&0\end{pmatrix}$
&
\def\pole{b}
\def\dynnum{3}
\def\dir{Ax5c7}
\def\dirp{Ax5c7p\pole}
$\begin{pmatrix}
0&\AIJ[I][\alpha][\beta]&0&0\\
0&0&0&\AIJ[I][\alpha][\beta]\\
1&0&0&0\\
0&0&\AIJ[I][\alpha][\beta]&0\end{pmatrix}$
\\
& & & \\
\hline
\end{tabular}
}
\end{center}
\caption{Chevalley models and intertwining operators for \(\aliam{V}{G}{\dym}{\beta}\).}
\end{table}
\end{proof}

%%%%%%%%%%%%%%%%%%%%%%%%%%%%%%%%%%%%%%%%%%%%%%%%%%%%%%%%%%%%%%%%%%%%
\def\dym{4}
\begin{Theorem}[\(\aliam{V}{G}{\dym}{\gamma}\)]
\label{teo:dim15app_c}
\def\dynnum{3}
All Automorphic Lie Algebras \(\aliam{V}{G}{\dym}{\gamma}\) are of type \(A_\dynnum^{(6,5)}\) and therefore isomorphic.
\end{Theorem}
\begin{proof}
\def\dynnum{3}
We give the Chevalley model together with its intertwining operator \({\cal{I}}_{\mf{sl}(V)}\) with respect to
\(\|A_\dynnum^{(6,5)}\|\) (see Table B6),
%(see Table \ref{tab:dim15app_c}),
i.e.
\[
\|\mf{sl}(V)\| {\cal{I}}_{\mf{sl}(V)}={\cal{I}}_{\mf{sl}(V)} \|A_{\dynnum}^{(6,5)}\|.
\]
\begin{table}[H]
\label{tab:dim15app_c}
\begin{center}
\scalebox{0.9}{
\begin{tabular}{c|ccc} 
\hline
Irreducible representation\\
$V$
&
$\bbbo_{8}$
&$\bbbi_{6}$
&$\bbbi_{7}$
\\
& & &  \\
\hline\hline
Chevalley model \\
$\|\mf{sl}(V)\|$ &
\def\pole{c}
\def\dynnum{3}
\def\dir{Ax4c8}
\def\dirp{Ax4c8p\pole}
$\begin{bmatrix}
0&\AIJ[I][\alpha][\gamma]&\AIJ[I][\alpha][\gamma]&\AIJ[I][\alpha][\gamma]\\
\AIJ[J][\beta][\gamma]&0&\AIJ[J][\beta][\gamma]&1\\
\AIJ[J][\beta][\gamma]&\AIJ[I][\alpha][\gamma]&0&\AIJ[I][\alpha][\gamma]\\
\AIJ[J][\beta][\gamma]&\AIJ[I][\alpha][\gamma]&\AIJ[J][\beta][\gamma]&0\end{bmatrix}$
&
\def\pole{c}
\def\dynnum{3}
\def\dir{Ax5c6}
\def\dirp{Ax5c6p\pole}
$\begin{bmatrix}
0&\AIJ[I][\alpha][\gamma]&\AIJ[I][\alpha][\gamma]&\AIJ[I][\alpha][\gamma]\\
\AIJ[J][\beta][\gamma]&0&\AIJ[I][\alpha][\gamma]&\AIJ[I][\alpha][\gamma]\\
\AIJ[J][\beta][\gamma]&\AIJ[J][\beta][\gamma]&0&1\\
\AIJ[J][\beta][\gamma]&\AIJ[J][\beta][\gamma]&\AIJ[I][\alpha][\gamma]&0\end{bmatrix}$
&
\def\pole{c}
\def\dynnum{3}
\def\dir{Ax5c7}
\def\dirp{Ax5c7p\pole}
$\begin{bmatrix}
0&\AIJ[J][\beta][\gamma]&\AIJ[J][\beta][\gamma]&\AIJ[I][\alpha][\gamma]\\
\AIJ[I][\alpha][\gamma]&0&\AIJ[I][\alpha][\gamma]&\AIJ[I][\alpha][\gamma]\\
\AIJ[I][\alpha][\gamma]&\AIJ[J][\beta][\gamma]&0&\AIJ[I][\alpha][\gamma]\\
1&\AIJ[J][\beta][\gamma]&\AIJ[J][\beta][\gamma]&0\end{bmatrix}$
\\
& & & \\
\hline
Intertwining operator \\
${\cal{I}}_{\mf{sl}(V)}$
& 
\def\pole{c}
\def\dynnum{3}
\def\dir{Ax4c8}
\def\dirp{Ax4c8p\pole}
$\begin{pmatrix}
1&0&0&0\\
0&0&0&1\\
0&1&0&0\\
0&0&1&0\end{pmatrix}$
&
\def\pole{c}
\def\dynnum{3}
\def\dir{Ax5c6}
\def\dirp{Ax5c6p\pole}
$\begin{pmatrix}
1&0&0&0\\
0&1&0&0\\
0&0&0&1\\
0&0&1&0\end{pmatrix}$
&
\def\pole{c}
\def\dynnum{3}
\def\dir{Ax5c7}
\def\dirp{Ax5c7p\pole}
$\begin{pmatrix}
0&0&1&0\\
1&0&0&0\\
0&1&0&0\\
0&0&0&1\end{pmatrix}$
\\
& & & \\
\hline
\end{tabular}
}
\end{center}
\caption{Chevalley models and intertwining operators for \(\aliam{V}{G}{\dym}{\gamma}\).}
\end{table}

\end{proof}

%%%%%%%%%%%%%%%%%%%%%%%%%%%%%%%%%%%%%%%%%%%%%%%%%%%%%%%%%%%%%%%%%%%%
\def\dym{5}
\begin{Theorem}[\(\aliam{V}{G}{\dym}{\alpha}\)]
\label{teo:dim25app_a}
\def\dynnum{4}
All Automorphic Lie Algebras \(\aliam{V}{G}{\dym}{\alpha}\) are of type \(A_\dynnum^{(8,6)}\) and therefore isomorphic.
\end{Theorem}

\begin{proof}
\def\dynnum{4}
We give the Chevalley model together with its intertwining operator \({\cal{I}}_{\mf{sl}(V)}\) with respect to
\(\|A_\dynnum^{(8,6)}\|\) (see Table B7),
%(see Table \ref{tab:dim25app_abc}),
i.e.
\[
\|\mf{sl}(V)\| {\cal{I}}_{\mf{sl}(V)}={\cal{I}}_{\mf{sl}(V)} \|A_{\dynnum}^{(8,6)}\|.
\]
\begin{table}[h!]
\label{tab:dim25app_abc}
\begin{center}
\scalebox{0.81}{
\begin{tabular}{c|ccc} 
\hline
Poles at \\
$\Gamma_{\zeta}$
&
$\Gamma_{\alpha}$
&
$\Gamma_{\beta}$
&
$\Gamma_{\gamma}$
\\
& & & \\
\hline\hline
Chevalley model \\
$\|\mf{sl}(V)\|$ 
&
\def\pole{a}
\def\dynnum{4}
\def\dir{Ax5c8}
\def\dirp{Ax5c8p\pole}
$\begin{bmatrix}
0&\AIJ[J][\gamma][\alpha]&\AIJ[J][\gamma][\alpha]&1&1\\
\AIJ[I][\beta][\alpha]&0&1&\AIJ[I][\beta][\alpha]&\AIJ[I][\beta][\alpha]\\
\AIJ[I][\beta][\alpha]&1&0&\AIJ[I][\beta][\alpha]&\AIJ[I][\beta][\alpha]\\
\AIJ[I][\beta][\alpha]&\AIJ[J][\gamma][\alpha]&\AIJ[J][\gamma][\alpha]&0&1\\
\AIJ[I][\beta][\alpha]&\AIJ[J][\gamma][\alpha]&\AIJ[J][\gamma][\alpha]&1&0\end{bmatrix}$
&
\def\pole{b}
\def\dynnum{4}
\def\dir{Ax5c8}
\def\dirp{Ax5c8p\pole}
$\begin{bmatrix}
0&\AIJ[I][\alpha][\beta]\AIJ[J][\gamma][\beta]&\AIJ[I][\alpha][\beta]\AIJ[J][\gamma][\beta]&1&1\\
1&0&\AIJ[I][\alpha][\beta]&1&1\\
1&1&0&1&1\\
\AIJ[I][\alpha][\beta]&\AIJ[I][\alpha][\beta]\AIJ[J][\gamma][\beta]&\AIJ[I][\alpha][\beta]\AIJ[J][\gamma][\beta]&0&1\\
\AIJ[I][\alpha][\beta]&\AIJ[I][\alpha][\beta]\AIJ[J][\gamma][\beta]&\AIJ[I][\alpha][\beta]\AIJ[J][\gamma][\beta]&\AIJ[I][\alpha][\beta]&0\end{bmatrix}$
&
\def\pole{c}
\def\dynnum{4}
\def\dir{Ax5c8}
\def\dirp{Ax5c8p\pole}
$\begin{bmatrix}
0&1&\AIJ[I][\alpha][\gamma]&\AIJ[J][\beta][\gamma]&\AIJ[I][\alpha][\gamma]\\
\AIJ[I][\alpha][\gamma]&0&\AIJ[I][\alpha][\gamma]&\AIJ[J][\beta][\gamma]&\AIJ[I][\alpha][\gamma]\\
\AIJ[J][\beta][\gamma]&\AIJ[J][\beta][\gamma]&0&\AIJ[J][\beta][\gamma]&1\\
\AIJ[I][\alpha][\gamma]&\AIJ[I][\alpha][\gamma]&\AIJ[I][\alpha][\gamma]&0&\AIJ[I][\alpha][\gamma]\\
\AIJ[J][\beta][\gamma]&\AIJ[J][\beta][\gamma]&\AIJ[I][\alpha][\gamma]&\AIJ[J][\beta][\gamma]&0\end{bmatrix}$
\\
& & & \\
\hline
Intertwining operator \\
${\cal{I}}_{\mf{sl}(V)}$
& 
\def\pole{a}
\def\dynnum{4}
\def\dir{Ax5c8}
\def\dirp{Ax5c8p\pole}
$\begin{pmatrix}
0&0&0&0&1\\
1&0&0&0&0\\
0&1&0&0&0\\
0&0&1&0&0\\
0&0&0&1&0\end{pmatrix}$
&
\def\pole{b}
\def\dynnum{4}
\def\dir{Ax5c8}
\def\dirp{Ax5c8p\pole}
$\begin{pmatrix}
0&0&0&0&\AIJ[I][\alpha][\beta]\\
1&0&0&0&0\\
0&1&0&0&0\\
0&0&0&\AIJ[I][\alpha][\beta]&0\\
0&0&\AIJ[I][\alpha][\beta]&0&0\end{pmatrix}$
&
\def\pole{c}
\def\dynnum{4}
\def\dir{Ax5c8}
\def\dirp{Ax5c8p\pole}
$\begin{pmatrix}
0&0&0&0&\AIJ[I][\alpha][\gamma]\\
0&0&0&\AIJ[I][\alpha][\gamma]&0\\
0&\AIJ[J][\beta][\gamma]&0&0&0\\
0&0&\AIJ[I][\alpha][\gamma]&0&0\\
\AIJ[J][\beta][\gamma]&0&0&0&0\end{pmatrix}$
\\
& & & \\
\hline
\end{tabular}
}
\end{center}
\caption{$V=\bbbi_{8}$; Chevalley models and intertwining operators for \(\aliam{V}{G}{\dym}{\zeta}\), \(\zeta=\alpha,\beta,\gamma\).} 
\end{table}

\end{proof}

\begin{Theorem}[\(\aliam{V}{G}{\dym}{\beta}\)]
\label{teo:dim25b}
\def\dynnum{4}
\def\dym{5}
All Automorphic Lie Algebras \(\aliam{V}{G}{\dym}{\beta}\) are of type \(A_\dynnum^{(10,6)}\) and therefore isomorphic.
\end{Theorem}
\begin{proof}
\def\dynnum{4}
We give the Chevalley model together with its intertwining operator \({\cal{I}}_{\mf{sl}(V)}\) with respect to
\(\|A_\dynnum^{(10,6)}\|\) (see Table B7),
%(see Table \ref{tab:dim25app_abc}),
i.e.
\[
\|\mf{sl}(V)\| {\cal{I}}_{\mf{sl}(V)}={\cal{I}}_{\mf{sl}(V)} \|A_{\dynnum}^{(10,6)}\|.
\]

\end{proof}

%%%%%%%%%%%%%%%%%%%%%%%%%%%%%%%%%%%%%%%%%%%%%%%%%%%%%%%%%%%%%%%%%%%%

\begin{Theorem}[\(\aliam{V}{G}{\dym}{\gamma}\)]
\label{teo:dim24app_c}
\def\dynnum{4}
\def\dym{5}
All Automorphic Lie Algebras \(\aliam{V}{G}{\dym}{\gamma}\) are of type \(A_\dynnum^{(10,8)}\) and therefore isomorphic.
\end{Theorem}
\begin{proof}
\def\dynnum{4}
We give the Chevalley model together with its intertwining operator \({\cal{I}}_{\mf{sl}(V)}\) with respect to
\(\|A_\dynnum^{(10,8)}\|\) (see Table B7),
i.e.
\[
\|\mf{sl}(V)\| {\cal{I}}_{\mf{sl}(V)}={\cal{I}}_{\mf{sl}(V)} \|A_{\dynnum}^{(10,8)}\|.
\]
\end{proof}

%%%%%%%%%%%%%%%%%%%%%%%%%%%%%%%%%%%%%%%%%%%%%%%%%%%%%%%%%%%%%%%%%%%%
\def\dym{6}
\begin{Theorem}[\(\aliam{V}{G}{\dym}{\alpha}\)]
\label{teo:dim35app_a}
\def\dynnum{5}
All Automorphic Lie Algebras \(\aliam{V}{G}{\dym}{\alpha}\) are of type \(A_\dynnum^{(12,9)}\) and therefore isomorphic.
\end{Theorem}
\begin{proof}
\def\dynnum{5}
We give the Chevalley model together with its intertwining operator \({\cal{I}}_{\mf{sl}(V)}\) with respect to
\(\|A_\dynnum^{(12,9)}\|\) (see Table B8),
i.e.
\[
\|\mf{sl}(V)\| {\cal{I}}_{\mf{sl}(V)}={\cal{I}}_{\mf{sl}(V)} \|A_{\dynnum}^{(12,9)}\|.
\]

\begin{table}[h!]
\label{tab:dim35app_abc}
\begin{center}
\scalebox{0.78}{
\begin{tabular}{c|ccc} 
\hline
Poles at\\
$\Gamma_{\zeta}$
&
$\Gamma_{\alpha}$
&
$\Gamma_{\beta}$
&
$\Gamma_{\gamma}$
\\
& & & \\
\hline\hline
Chevalley model \\
$\|\mf{sl}(V)\|$ 
&
\def\pole{a}
\def\dynnum{5}
\def\dir{Ax5c9}
\def\dirp{Ax5c9p\pole}
$\begin{bmatrix}
0&\AIJ[I][\beta][\alpha]&\AIJ[I][\beta][\alpha]&\AIJ[I][\beta][\alpha]&\AIJ[I][\beta][\alpha]&1\\
\AIJ[J][\gamma][\alpha]&0&1&\AIJ[J][\gamma][\alpha]&1&\AIJ[J][\gamma][\alpha]\\
\AIJ[J][\gamma][\alpha]&\AIJ[I][\beta][\alpha]&0&\AIJ[J][\gamma][\alpha]&\AIJ[I][\beta][\alpha]&\AIJ[J][\gamma][\alpha]\\
1&\AIJ[I][\beta][\alpha]&1&0&\AIJ[I][\beta][\alpha]&1\\
\AIJ[J][\gamma][\alpha]&1&1&\AIJ[J][\gamma][\alpha]&0&\AIJ[J][\gamma][\alpha]\\
1&\AIJ[I][\beta][\alpha]&\AIJ[I][\beta][\alpha]&\AIJ[I][\beta][\alpha]&\AIJ[I][\beta][\alpha]&0\end{bmatrix}$
&
\def\pole{b}
\def\dynnum{5}
\def\dir{Ax5c9}
\def\dirp{Ax5c9p\pole}
$\begin{bmatrix}
0&\AIJ[I][\alpha][\beta]&\AIJ[J][\gamma][\beta]&1&\AIJ[J][\gamma][\beta]&\AIJ[J][\gamma][\beta]\\
1&0&\AIJ[J][\gamma][\beta]&1&\AIJ[J][\gamma][\beta]&\AIJ[J][\gamma][\beta]\\
\AIJ[I][\alpha][\beta]&\AIJ[I][\alpha][\beta]&0&\AIJ[I][\alpha][\beta]&\AIJ[I][\alpha][\beta]&1\\
\AIJ[I][\alpha][\beta]&\AIJ[I][\alpha][\beta]&\AIJ[J][\gamma][\beta]&0&\AIJ[J][\gamma][\beta]&\AIJ[J][\gamma][\beta]\\
\AIJ[I][\alpha][\beta]&\AIJ[I][\alpha][\beta]&1&\AIJ[I][\alpha][\beta]&0&1\\
\AIJ[I][\alpha][\beta]&\AIJ[I][\alpha][\beta]&1&\AIJ[I][\alpha][\beta]&\AIJ[I][\alpha][\beta]&0\end{bmatrix}$
&
\def\pole{c}
\def\dynnum{5}
\def\dir{Ax5c9}
\def\dirp{Ax5c9p\pole}
$\begin{bmatrix}
0&1&\AIJ[J][\beta][\gamma]&\AIJ[I][\alpha][\gamma]&\AIJ[I][\alpha][\gamma]&\AIJ[J][\beta][\gamma]\\
\AIJ[I][\alpha][\gamma]&0&\AIJ[J][\beta][\gamma]&\AIJ[I][\alpha][\gamma]&\AIJ[I][\alpha][\gamma]&\AIJ[J][\beta][\gamma]\\
\AIJ[I][\alpha][\gamma]&\AIJ[I][\alpha][\gamma]&0&\AIJ[I][\alpha][\gamma]&\AIJ[I][\alpha][\gamma]&1\\
\AIJ[J][\beta][\gamma]&\AIJ[J][\beta][\gamma]&\AIJ[J][\beta][\gamma]&0&\AIJ[I][\alpha][\gamma]&\AIJ[J][\beta][\gamma]\\
\AIJ[J][\beta][\gamma]&\AIJ[J][\beta][\gamma]&\AIJ[J][\beta][\gamma]&1&0&\AIJ[J][\beta][\gamma]\\
\AIJ[I][\alpha][\gamma]&\AIJ[I][\alpha][\gamma]&1&\AIJ[I][\alpha][\gamma]&\AIJ[I][\alpha][\gamma]&0\end{bmatrix}$
\\
& & & \\
\hline
Inter operator \\
${\cal{I}}_{\mf{sl}(V)}$
& 
\def\pole{a}
\def\dynnum{5}
\def\dir{Ax5c9}
\def\dirp{Ax5c9p\pole}
$\begin{pmatrix}
1&0&0&0&0&0\\
0&0&0&0&1&0\\
0&0&0&1&0&0\\
0&0&1&0&0&0\\
0&0&0&0&0&1\\
0&1&0&0&0&0\end{pmatrix}$
&
\def\pole{b}
\def\dynnum{5}
\def\dir{Ax5c9}
\def\dirp{Ax5c9p\pole}
$\begin{pmatrix}
0&0&0&0&1&0\\
0&0&0&0&0&1\\
1&0&0&0&0&0\\
0&0&0&1&0&0\\
0&0&1&0&0&0\\
0&1&0&0&0&0\end{pmatrix}$
&
\def\pole{c}
\def\dynnum{5}
\def\dir{Ax5c9}
\def\dirp{Ax5c9p\pole}
$\begin{pmatrix}
0&0&0&1&0&0\\
0&0&1&0&0&0\\
1&0&0&0&0&0\\
0&0&0&0&1&0\\
0&0&0&0&0&1\\
0&1&0&0&0&0\end{pmatrix}$
\\
& & & \\
\hline
\end{tabular}
}
\end{center}
\caption{$V=\bbbi_{9}$; Chevalley models and intertwining operators for \(\aliam{V}{G}{\dym}{\zeta}\), \(\zeta=\alpha,\beta,\gamma\).} 
\end{table}
\end{proof}

%%%%%%%%%%%%%%%%%%%%%%%%%%%%%%%%%%%%%%%%%%%%%%%%%%%%%%%%%%%%%%%%%%%%

\begin{Theorem}[\(\aliam{V}{G}{\dym}{\beta}\)]
\label{teo:dim35app_b}
\def\dynnum{5}
\def\dym{6}
All Automorphic Lie Algebras \(\aliam{V}{G}{\dym}{\beta}\) are of type \(A_\dynnum^{(14,9)}\) and therefore isomorphic.
\end{Theorem}
\begin{proof}
\def\dynnum{5}
\def\dym{6}
We give the Chevalley model together with its intertwining operator \({\cal{I}}_{\mf{sl}(V)}\) with respect to
\(\|A_\dym^{(14,9)}\|\) (see Table B8),
i.e.
\[
\|\mf{sl}(V)\| {\cal{I}}_{\mf{sl}(V)}={\cal{I}}_{\mf{sl}(V)} \|A_{\dynnum}^{(14,9)}\|.
\]
\end{proof}

%%%%%%%%%%%%%%%%%%%%%%%%%%%%%%%%%%%%%%%%%%%%%%%%%%%%%%%%%%%%%%%%%%%%
\def\dym{6}
\begin{Theorem}[ \(\aliam{V}{G}{\dym}{\gamma}\)]
\label{teo:dim35app_c}
\def\dynnum{5}

All Automorphic Lie Algebras \(\aliam{V}{G}{\dym}{\gamma}\) are of type \(A_\dynnum^{(14,12)}\) and therefore isomorphic.
\end{Theorem}
\begin{proof}
\def\dynnum{5}
\def\dym{6}
We give the Chevalley model together with its intertwining operator \({\cal{I}}_{\mf{sl}(V)}\) with respect to
\(\|A_\dym^{(14,12)}\|\) (see Table B8),
i.e.
\[
\|\mf{sl}(V)\| {\cal{I}}_{\mf{sl}(V)}={\cal{I}}_{\mf{sl}(V)} \|A_{\dynnum}^{(14,12)}\|.
\]
\end{proof}

%%%%%%%%%%%%%%%%%%%%%%%%%%%%%%%%%%%%%%%%%%%%%%%%%%%%%%%%%%%%%%%%%%%%%%%%%%%%%%%%%%%%%%%%%%%%%%%%%%%%%%%%%%%%%%%%%%%%%%

\def\cprime{$'$}

\end{document}